\documentclass[12pt]{report}

\usepackage[margin=1in]{geometry}         
\usepackage{setspace}                     

\usepackage[caption=false,font=normalsize,labelfont=sf,textfont=sf]{subfig}
\usepackage{textcomp}
\usepackage{stfloats}
\usepackage{multicol}
\usepackage{algorithm,algorithmic,multirow}
\usepackage{bbm, relsize}
\usepackage{amsmath}
\usepackage{amsfonts}
\usepackage{amssymb}
\usepackage{graphicx}
\usepackage{color}
\usepackage[normalem]{ulem}
\usepackage{nomencl}
\makenomenclature

\newcommand{\yadi}{\nomenclature}
\newenvironment{proof}{\noindent{\sc Proof.}}{\qed}
\newtheorem{theorem}{Theorem}[section]
\newtheorem{lemma}{Lemma}[section]

\newtheorem{remark}{Remark}[section]
\newtheorem{rem}{Remark}[section]
\newtheorem{definition}{Definition}[section]
\newtheorem{prop}{Proposition}[section]

\newcommand{\qed}{$\blacksquare$}

\def\bhag#1{\noindent
\setcounter{equation}{0}
\section{#1}
}

\def\RR{{\mathbb R}}
\def\CC{{\mathbb C}}
\def\ZZ{{\mathbb Z}}

\def\TT{\mathbb T}

\def\bs#1{{\boldsymbol{#1}}}
\def\x{\mathbf{x}}

\def\w{\mathbf{w}}

\def\O{{\cal O}}

\usepackage{listings}

\definecolor{dkgreen}{rgb}{0,0.6,0}
\definecolor{gray}{rgb}{0.5,0.5,0.5}
\definecolor{mauve}{rgb}{0.58,0,0.82}

\def\be{\begin{equation}}
\def\ee{\end{equation}}
\def\bea{\begin{eqnarray}}
\def\eea{\end{eqnarray}}

\def\disp{\displaystyle}

\def\donchitre#1#2{\vskip 6.5cm\noindent
\parbox[t]{1in}{\special{eps:#1.eps x=6.5cm y=5.5cm}}
\hbox to 7cm{}\parbox[t]{0.0cm}{\special{eps:#2.eps x=6.5cm y=5.5cm}}}

\def\bs#1{{\boldsymbol{#1}}}

\def\gs{\gtrsim}
\def\ls{\lesssim}

\def\hati{i}

\begin{document}
\thispagestyle{empty}
\vspace*{2in}

\begin{center}
\LARGE
Localized Kernel Methods for Signal Processing

\vspace{1.5in}
\large
By\\[1em]
Sippanon Kitimoon

\vfill
Claremont Graduate University\\
2025

\end{center}

\newpage

\thispagestyle{empty}

\vspace*{\fill}

\begin{center}
\copyright\ Copyright Sippanon Kitimoon, 2025.\\
All rights reserved
\end{center}

\newpage

\thispagestyle{empty}
\begin{center}
\textbf{Approval of the Dissertation Committee}
\end{center}

\vspace{1em}

This dissertation has been duly read, reviewed, and critiqued by the Committee listed below, which hereby approves the manuscript of Sippanon Kitimoon as fulfilling the scope and quality requirements for meriting the degree of Doctor of Philosophy in Mathematics.

\vspace{1em}

\begin{center}
Hrushikesh Mhaskar, Chair \\
Claremont Graduate University \\
Distinguished Research Professor of Mathematics

\vspace{1em}

Qidi Peng \\
Claremont Graduate University \\
Research Associate Professor of Mathematics

\vspace{1em}

Allon Percus \\
Claremont Graduate University \\
Joseph H. Pengilly Professor of Mathematics

\end{center}

\vfill

\newpage

\thispagestyle{empty}
\begin{center}
\textbf{Abstract}

\vspace{1em}

Localized Kernel Methods for Signal Processing \\
By \\
Sippanon Kitimoon \\

\vspace{1em}

Claremont Graduate University: 2025
\end{center}

\doublespacing

This dissertation presents two signal processing methods using specially designed localized kernels for parameter recovery under noisy condition. The first method addresses the estimation of frequencies and amplitudes in multidimensional exponential models. It utilizes localized trigonometric polynomial kernels to detect the multivariate frequencies, followed by a more detailed parameter estimation. We compare our method with MUSIC and ESPRIT, which are classical subspace-based algorithms widely used for estimating the parameters of exponential signals. In the univariate case, the method outperforms MUSIC and ESPRIT under low signal-to-noise ratios. For the multivariate case, we develop a coordinate-wise projection and registration approach that achieves high recovery accuracy using significantly fewer samples than other methods.

The second method focuses on separating linear chirp components from time-localized signal segments. A variant of the Signal Separation Operator (SSO) is constructed using a localized kernel. Instantaneous frequency estimates are obtained via FFT-based filtering, then clustered and fitted with piecewise linear regression. The method operates without prior knowledge of the number of components and is shown to recover intersecting and discontinuous chirps at SNR levels as low as $-30$ dB.

Both methods share an idea based on localized kernels and efficient FFT-based implementation, and neither requires subspace decomposition or sparsity regularization. Experimental results confirm the robustness and tractability of the proposed approaches across a range of simulated data conditions. Potential extensions include application to nonlinear chirps, adaptive kernel design, and signal classification using extracted features.

\thispagestyle{empty}

\newpage

\pagenumbering{roman}
\setcounter{page}{5}

\begin{center}
\textbf{Acknowledgements}
\end{center}

\vspace{1em}
\doublespacing

I would like to express my deepest gratitude to my Ph.D. advisor, Prof. Hrushikesh Mhaskar, for his unwavering support, insightful guidance, and patient mentorship throughout every stage of this work. His mathematical rigor, clarity of thought, and encouragement have been essential to my development as a researcher. Beyond his academic mentorship, I am especially thankful for the parental care and support he provided during my Ph.D. journey, for which I am sincerely grateful.

I am also grateful to the other members of my dissertation committee, Prof. Allon Percus and Prof. Qidi Peng, for agreeing to serve on my dissertation committee. I greatly appreciate their time, dedication, and encouragement during my doctoral journey.

I would like to thank Dr. Raghu G. Raj and Dr. Eric Mason for their essential contributions to this work. Dr. Raj provided helpful direction and recommendations during the development of the methods presented in multivariate the exponential analysis portion of the dissertation. Dr. Mason generated simulation data and offered key guidance and suggestions in the formulation and implementation of the chirp signal separation approach. Their involvement improved both the clarity and practical relevance of the results.

Finally, I am thankful to all those who have supported and encouraged me in both work and recreational environment. I wish them all the best.

\newpage
\tableofcontents
\newpage

\begin{thenomenclature} 
\nomgroup{A}
  \item [$\zeta(S)$] Riemann zeta function
  \item [{$\delta(t)$}]\begingroup Dirac delta function\nomeqref {2.3}\nompageref{3}
  \item [{$\delta_y$}]\begingroup Dirac delta supported at $y$\nomeqref {3.0}\nompageref{3}
  \item [{$\Delta_j$}]\begingroup Independent vectors in $\RR^q$\nomeqref {5.0}\nompageref{10}
  \item [{$P(z)$}]\begingroup Prony polynomial\nomeqref {2.3}\nompageref{3}
  \item [{$c_j$}]\begingroup Coefficients of the Prony polynomial\nomeqref {2.3}\nompageref{3}
  \item [$\|\cdot\|_2$] L2 norm
  \item [{$\mathbf{H}_{L, N-L+1}$}]\begingroup Rectangular Hankel matrix\nomeqref {2.3}\nompageref{3}
  \item [{$e_L(\varphi)$}]\begingroup The column space of $\mathbf{H}_{L, N-L+1}$ \nomeqref {2.3}\nompageref{3}
  \item [{$\mathcal{S}_L$}]\begingroup The signal space\nomeqref {2.3}\nompageref{3}
  \item [{$\mathcal{N}_L$}]\begingroup The noise space\nomeqref {2.3}\nompageref{3}
  \item [{$\widetilde{W}_{N-L+1,K}$}]\begingroup The first $K$ right singular vectors of $\mathbf{H}_{L, N-L+1}$ \nomeqref {2.3}\nompageref{3}
  \item [{$V_{N,K}(z)$}]\begingroup The Vandermonde matrix\nomeqref {2.3}\nompageref{3}
  \item [{$\epsilon $}]\begingroup Random variable (noise)\nomeqref {3.2}\nompageref{4}
  \item [{$\eta$}]\begingroup Minimal separation among unitless frequencies\nomeqref {4.1}\nompageref{5}
  \item [{$\hat{\mu}$}]\begingroup Fourier coefficient of $\mu$\nomeqref {3.0}\nompageref{3}
  \item [{$\hbar_n$}]\begingroup Normalizing constant, Eqn. \eqref{eq:lockerndef}\nomeqref {3.2}\nompageref{4}
  \item [{$\lambda_k$}]\begingroup Univariate unitless frequencies\nomeqref {3.1}\nompageref{3}
  \item [{$\Lambda_{j,k}$}]\begingroup Estimated IF for chirp $j$ at $t_k$\nomeqref {4.2}\nompageref{8}
  \item [{$\mathbf{M}$}]\begingroup Sum of absolute values of amplitudes\nomeqref {3.12}\nompageref{5}
  \item [{$\mathcal{T}_{n,R}$}]\begingroup Signal separation operator\nomeqref {3.25}\nompageref{6}
  \item [{$\mathbb{G}$}]\begingroup Support of the thresholded power spectrum\nomeqref {4.7}\nompageref{6}
  \item [{$\mathbb{G}_\ell$}]\begingroup Cluster definded in Theorem~\ref{theo:main}\nomeqref {4.7}\nompageref{6}
  \item [{$\mathfrak{m}$}]\begingroup Minimum absolute value of coefficients\nomeqref {4.1}\nompageref{5}
  \item [{$\omega_{j,k}$, $\gamma_{j,k}$, $a_{j,k}$}]\begingroup Chirp parameters for chirp $j$ in $I_k$\nomeqref {4.2}\nompageref{8}
  \item [{$\mu$}]\begingroup Discretely supported measure/distribution on $\TT$\nomeqref {3.0}\nompageref{3}
  \item [{$\Phi_n$}]\begingroup Localized kernel of degree $n$\nomeqref {3.2}\nompageref{4}
  \item [{$\RR$}]\begingroup Set of real numbers\nomeqref {3.0}\nompageref{3}
  \item [{$\sigma_n$}]\begingroup Reconstruction operator, Section~\ref{section:theosect}\nomeqref {3.2}\nompageref{4}
  \item [{$\tilde {\mu }$}]\begingroup Fourier coefficients of $\mu $ plus an additive noise\nomeqref {3.2}\nompageref{4}
  \item [{$\TT$}]\begingroup Quotient space of real line modulo $2\pi$\nomeqref {3.0}\nompageref{3}
  \item [{$\w_k$}]\begingroup Multivariate frequencies/points\nomeqref {1.1}\nompageref{1}
  \item [{$\ZZ$}]\begingroup Set of integers\nomeqref {3.0}\nompageref{3}
  \item [{$A_k$}]\begingroup Coefficients (complex amplitudes)\nomeqref {3.1}\nompageref{3}
  \item [{$C$}]\begingroup Eqn. \eqref{eq:thresholdCdef}\nomeqref {4.7}\nompageref{6}
  \item [{$E_n$}]\begingroup Reconstruction operator used with noise alone Eqn. \eqref{eq:noisespectrum}\nomeqref {4.1}\nompageref{5}
  \item [{$H$}]\begingroup Smooth low pass filter\nomeqref {3.2}\nompageref{4}
  \item [{$i$}]\begingroup $\sqrt{-1}$\nomeqref {3.1}\nompageref{3}
  \item [{$K$}]\begingroup Number of exponential signals\nomeqref {3.1}\nompageref{3}
  \item [{$L$}]\begingroup Constant in Eqn. \eqref{eq:locest}\nomeqref {3.8}\nompageref{5}
  \item [{$M$}]\begingroup Sum of absolute values of coefficients\nomeqref {4.1}\nompageref{5}
  \item [{$q$}]\begingroup Dimension of the observations\nomeqref {1.1}\nompageref{1}
  \item [{$S$}]\begingroup Localization power Eqn. \eqref{eq:locest}\nomeqref {3.8}\nompageref{5}
  \item [{$\tau$}]\begingroup Threshold for SSO\nomeqref {4.2}\nompageref{8}
  \item [{$\theta$, $\gamma$, $B$, $d$}]\begingroup Chirp parameters\nomeqref {1.1}\nompageref{1}
  \item [{$A_j(t)$, $\phi_j(t)$}]\begingroup Instantaneous complex amplitudes and phases\nomeqref {2.2}\nompageref{3}
  \item [{$B_{\mbox{\scriptsize{rec}}}$}]\begingroup Receiver bandwidth\nomeqref {4.2}\nompageref{8}
  \item [{$I_k$}]\begingroup Time interval for snippet $[t_k-\Delta,t_k+\Delta]$\nomeqref {3.18}\nompageref{5}
  \item [{$R$}]\begingroup Sampling frequency\nomeqref {3.18}\nompageref{5}
  \item [{$t_k$}]\begingroup Center of the time interval for snippet\nomeqref {3.18}\nompageref{5}
  \item [{$\nu$}]\begingroup Parameter (variance) of sub-Gaussian variable\nomeqref {3.2}\nompageref{4}
  \item [{PRI}]\begingroup Pulse repetition interval\nomeqref {1.2}\nompageref{2}

\end{thenomenclature}

\newpage
\listoffigures
\newpage
\listoftables
\newpage

\pagenumbering{arabic}


\chapter{Introduction}\label{background}

This dissertation focuses on two core applications in signal processing:
\begin{itemize}
    \item Robust and Tractable Multidimensional Exponential Analysis
    \item Localized Kernel Method for Separation of Linear Chirps
\end{itemize}
The foundation of both applications is built upon the constant parameter problem formally defined in Section~\ref{section:problem}, where the goal is to recover unknown parameters of signals composed of complex exponentials or chirp components from noisy observations. By designing and utilizing localized kernel methods, we develop robust algorithms that perform well in the presence of sub-Gaussian noise. Since the two applications addressed in this dissertation focus on solving distinct signal processing problems, the literature reviews for each application are presented separately in their respective chapters. 

In this chapter, we will focus on defining the constant parameter problem that underlies both applications and introducing the state-of-the-art methods used to solve the problem.

\newpage

\bhag{Background and problem statement}\label{section:theosect}

We begin with a brief review of the sub-Gaussian noise model in Section~\ref{section:subgaussian}, followed by a definition of the signal-to-noise ratio used throughout this work in Section~\ref{bhag:datagen}. Section~\ref{section:problem} concludes with a formal statement of the constant parameter problem, which serves as the foundational model for the applications developed in later chapters.

\subsection{Sub-Gaussian noise}\label{section:subgaussian}

In this section, we review certain properties of sub-Gaussian random variables. Intuitively, a sub-Gaussian random variable is one whose probability tails decay at least as fast as those of a Gaussian random variable; i.e. extreme deviations are exponentially unlikely and no heavier than in the normal distribution. The material in this section is based on \cite[Section~2.3]{boucheron2013concentration}, with a slight change of notation. 
A mean zero real valued random variable $X$ is called sub-Gaussian with parameter $\nu$ ($X\in\mathcal{G}(\nu)$) if \be\label{subgaussian} \log\mathbb{E}(\exp(tX))\le (t\nu)^2/2.\ee \yadi{$\nu$}{(Variance) parameter of sub-Gaussian variable}
Examples include Gaussian variables and all bounded random variables.
For a sub-Gaussian variable, it is shown in  \cite[Section~2.3]{boucheron2013concentration} that condition \eqref{subgaussian} implies the tail bound
\be \label{tail_bound}
\mathsf{Prob}(|X|>t)\le 2\exp(-t^2/(2\nu^2)).
\ee
Conversely, by \cite[Section~2.3, Theorem 2.1]{boucheron2013concentration}, the tail bound in \eqref{tail_bound} implies
$X\in\mathcal{G}(4\nu)$.

We will say that a complex valued random variable $X$ is in $\mathcal{G}(\nu)$ if both the real and imaginary parts of $X$ are in $\mathcal{G}(\nu)$. 
We observe that if $z\in\CC$ and $|z|>t$ then $\max(|\Re e (z)|, |\Im m (z)|)>t/\sqrt{2}$. 
So, for such variables, we have
\begin{equation}\label{eq:subgaussiantail}
\mathsf{Prob}(|X|>t)\le 4\exp(-t^2/(4\nu^2)).
\end{equation}
It is not difficult to see that if $X_1,\cdots,X_n$ are i.i.d., complex valued variables all in $\mathcal{G}(\nu)$, $\mathbf{a}=(a_1,\cdots,a_n)\in\RR^n$, $|\mathbf{a}|_n^2=\sum_{\ell=1}^n a_\ell^2$, then $\sum_{\ell=1}^n a_\ell X_\ell \in \mathcal{G}(|\mathbf{a}|_n\nu)$. 
Therefore, \eqref{eq:subgaussiantail} implies
\begin{equation}\label{eq:subgaussian_sum_tail}
\mathsf{Prob}\left(\left|\sum_{\ell=1}^n a_\ell X_\ell\right| >t\right)\le 4\exp\left(-\frac{t^2}{4|\mathbf{a}|_n^2\nu^2}\right), \qquad t>0.
\end{equation}

\subsection{Signal to noise ratio}\label{bhag:datagen}

The signal-to-noise ratio (SNR) measures how strong the signal is compared to the noise; i.e., higher SNR values indicate clearer signals, while lower values indicate signals that are more difficult to distinguish from noise. In this dissertation, we used a synthetically generated dataset to evaluate our methods under controlled conditions.  The original signal $f(t)$ was constructed, and Gaussian noise was introduced to simulate noisy environments.

The SNR in dB is defined as
\be
\text{SNR} = 20 \log_{10} \left( \frac{\|\mathbf{f}\|_2}{\|\boldsymbol{\epsilon}\|_2} \right),
\ee
where $\|\cdot\|_2$ is an L2 norm, $\mathbf{f}$ is the vector of the original clean signal $f(t)$, and $\boldsymbol{\epsilon}$ is the vector of noise $\epsilon(t)$.

To control the SNR, we generate a vector of Gaussian noise $\mathbf{n}$ and compute
\begin{equation}
    \boldsymbol{\epsilon} = \frac{\|\mathbf{f}\|_2}{10^{\text{SNR}/20} \cdot \|\mathbf{n}\|_2} \cdot \mathbf{n},
    \label{eq:noise_scaling}
\end{equation}
where SNR is the desired signal-to-noise ratio in decibels (dB).

We then obtain the noisy signal using the formula
\begin{equation}
    \mathbf{F} = \mathbf{f} + \boldsymbol{\epsilon},
    \label{eq:noised_signal}
\end{equation}
where $\mathbf{F}$ is the vector of the noised signal $F(t)$.

By scaling the Gaussian noise $\mathbf{n}$ using Equation~\ref{eq:noise_scaling}, we ensure that the resulting noisy signal $\mathbf{F}$ achieves the desired SNR. This allows us to evaluate the performance of our experiments under varying noise levels.

\subsection{Problem statement}\label{section:problem}

In this section, we consider the following univariate set up.
Let $\TT=\RR/(2\pi\ZZ)$ (i.e., $x, y$ are considered equal if $x=y \mbox{ mod } 2\pi$),\yadi{$\TT$}{Quotient space of real line modulo $2\pi$} \yadi{$\RR$}{Set of real numbers} \yadi{$\ZZ$}{set of integers} $K\ge 1$ be an integer, $\lambda_k\in \TT$, $k=1,\cdots, K$, $A_k\in\CC$ for $k=1,\cdots, K$. 
We define \yadi{$\delta_y$}{Dirac delta supported at $y$} \yadi{$\mu$}{discretely supported measure/distribution on $\TT$} \yadi{$\hat{\mu}$}{Fourier coefficient of $\mu$}
\begin{equation}\label{eq:fourmoments}
\mu=\sum_{k=1}^K A_k\delta_{\lambda_k}, \quad \hat{\mu}(\ell) =\sum_{k=1}^K A_k\exp(-i\ell\lambda_k), \qquad \ell\in\ZZ,
\end{equation}
where $i=\sqrt{-1}$, \yadi{$i$}{$\sqrt{-1}$} $\delta_\lambda$ denotes the Dirac delta supported at $\lambda$. \yadi{$A_k$}{coefficients (complex amplitudes)} \yadi{$\lambda_k$}{Univariate unitless frequencies} \yadi{$K$}{number of exponential signals} We then formulate our problem statement as follow:

\vspace{2ex}

\noindent\textbf{Point source separation problem.}

\vspace{2ex}

\textit{Given finitely many noisy samples
\begin{equation}\label{eq:noisysamples}
\tilde{\mu}(\ell)=\hat{\mu}(\ell)+\epsilon_\ell, \qquad |\ell|<n,
\end{equation}
where $n\ge 1$ is an integer, and $\epsilon_\ell$ are realizations of a sub-Gaussian random variable, determine $K$, $A_k$, and $\lambda_k$, $k=1,\cdots, K$. \yadi{$\tilde{\mu}$}{Fourier coefficients of $\mu$ plus an additive noise} \yadi{$\epsilon$}{Random variable (noise)}
}

\vspace{2ex}

The harder part of the problem is to determine $K$ and the $\lambda_k$'s. 
The coefficients $A_k$ can then be determined by solving a linear system of equations. 
Many algorithms to do this are known, e.g., \cite{pottstaschesingdet}. 
Therefore, we will focus on the task of finding $K$ and the $\lambda_k$'s.

\bhag{Well-known solutions}

In this section, we will briefly review classical and modern methods developed to solve the point source separation problem introduced in Section \ref{section:problem}. These include Prony's \cite{prony_original}, MUSIC \cite[Section 10.2.1]{plonka2018numerical}, and ESPRIT \cite[Section 10.2.2]{plonka2018numerical} methods, which have served as state of the art tools for parameter recovery in exponential models.

\subsection{Prony's method}

Prony's method, originally proposed in \cite{prony_original}, is a classical approach for estimating the parameters of exponential sums from uniformly spaced data. 

Let $\mu(t)=\sum_{k=1}^K A_k\delta_{\lambda_k}(t)$, where $A_k$'s are complex numbers,  $\lambda_k$'s are complex numbers with $-\pi < \Im m (\lambda_k)\le \pi$, and $\delta_\lambda$ is the Dirac delta supported at $\lambda$.
The Prony method seeks to find $K$ and $\lambda_k$'s given data of the form
\be\label{eq:pronydata}
\hat{\mu}(\ell)=\sum_{k=1}^K A_k\exp(-i\lambda_k\ell), \qquad \ell=0,\cdots, n-1,
\ee
where $n\ge 2K$ is an integer.

Writing $z_k=\exp(-i\lambda_k)$, we introduce the Prony polynomial
\be\label{eq:pronypoly}
P(z)=\prod_{k=1}^K (z-z_k)=\sum_{j=0}^K c_jz^j, 
\ee
where $c_K=1$.
Obviously, for every integer $\ell$, $0\le \ell\le K-1$,
\be\label{eq:pronyortho}
\int z^\ell P(z)d\mu(z)=\sum_{k=1}^K A_kz_k^\ell P(z_k)=0.
\ee
Since
\be
\sum_{k=1}^K A_kz_k^\ell P(z_k)=\sum_{k=1}^K A_kz_k^\ell\sum_{j=1}^K c_jz_k^j=\sum_{j=1}^K c_j\sum_{k=1}^K A_kz_k^{\ell+j}=\sum_{j=1}^K c_j\hat{\mu}(j+\ell),
\ee
we see that the coefficients of the Prony polynomial are solutions of the equations
\be\label{eq:pronyeqn}
\sum_{j=0}^K c_j\hat{\mu}(j+\ell)=0, \qquad \ell=0,1,\cdots, n-K-1.
\ee
Prony's method consists of solving this system for the coefficients $c_j$, $j=1,\cdots, K-1$, and then find the zeros $z_k=\exp(-i\lambda_k)$ of the Prony polynomial.

The points $z_k$ can be determined directly as joint eigenvalues.
Let $z\in\CC$. 
We consider the matrix $\mathbf{M}(z)$ with entries given by
\be
M_{k,\ell}(z)=\hat{\mu}(\ell+k+1)-z\hat{\mu}(\ell+k), \qquad k=1,\cdots, K, \ \ell=0,\cdots, K-1,
\ee
and observe that the determinant
$Q(z)=\det(\mathbf{M}(z))$ is a polynomial in $z$ of degree $\le K$.
We consider the lower triangular matrix $\mathbf{L}(z)$ given by
\be
L_{j,k}=\begin{cases}
z^{j-k}, &\mbox{ if $0 \leq k \leq j \leq K-1$},\\
0, &\mbox{otherwise}.
\end{cases}
\ee
Since $\det(\mathbf{L}(z))=1$ for all $z$, we have
\be
\det(\mathbf{L}(z)\mathbf{M}(z))=\det(\mathbf{M}(z))=Q(z).
\ee
We note that
\be
\begin{aligned}
(\mathbf{L}(z)\mathbf{M}(z))_{j,\ell}&=\sum_{k=0}^{K-1}L_{j,k}M_{k,\ell}
=\sum_{k=0}^j z^{j-k}\hat{\mu}(k+\ell+1)-\sum_{k=0}^j z^{j-k+1}\hat{\mu}(k+\ell)\\
&=\sum_{k=1}^{j+1}z^{j-k+1}\hat{\mu}(k+\ell)-\sum_{k=0}^j z^{j-k+1}\hat{\mu}(k+\ell)=\hat{\mu}(\ell+j+1)-z^{j+1}\hat{\mu}(\ell).
\end{aligned}
\ee
Consequently,
\be
\int z^\ell (\mathbf{L}(z)\mathbf{M}(z))_{j,\ell}d\mu(z)=0, \qquad j=1,\cdots, K, \ \ell=0, \cdots,K-1,
\ee
and hence, also that
\be
\int z^\ell Q(z)d\mu(z)=0, \qquad \ell=0,\cdots, K-1.
\ee
This is the same system of equations as \eqref{eq:pronyortho}.
So, $Q$ is a constant multiple of $P$.

Thus, the $z_k$'s are joint eigenvalues as indicated by the matrix $\mathbf{M}(z)$; i.e. the zeros of $\det(\mathbf{M}(z))$.


\subsection{MUSIC method}

The Multiple Signal Classification (MUSIC) algorithm \cite[Section 10.2.1]{plonka2018numerical} is a technique for recovering the parameters of an exponential sum, particularly when the data is affected by noise. It operates on the principle of decomposing a trajectory matrix (a Hankel matrix formed from the data) into signal and noise subspaces using singular value decomposition (SVD).

To proceed, we assume that a parameter \( L \in \mathbb{N} \) is selected as an upper bound on \( K \), satisfying the condition
\be
K \leq L \leq n - K + 1,
\ee

We define the \( L \)-trajectory matrix \( \mathbf{H}_{L, n-L+1} \in \mathbb{C}^{L \times (n-L+1)} \) constructed from this time series by
\begin{equation}
    \mathbf{H}_{L, n-L+1} := \left(\hat{\mu}(\ell + m)\right)_{\ell=0,\ldots,L-1;\, m=0,\ldots,n-L}, \label{eq:trajectory_matrix}
\end{equation}
where \( L \in \{K, \ldots, n - K + 1\} \) is called the \emph{window length}. Note that \( \mathbf{H}_{L, n-L+1} \) is a rectangular Hankel matrix, with constant entries along each anti-diagonal.

Given a sequence of noisy samples \( \tilde{\mu}(\ell)=\hat{\mu}(\ell)+\epsilon_\ell \), one constructs a Hankel matrix \( \mathbf{H}_{L,n-L+1} \in \mathbb{C}^{L \times (n-L+1)} \). In the ideal (noiseless) case, the column space of this matrix coincides with the span of \( K \) exponential vectors \( \{ e_L(\varphi_j) \}_{j=1}^K \), where
\be
e_L(\varphi) = \left( e^{i\varphi \ell} \right)_{\ell = 0}^{L-1} \in \mathbb{C}^L.
\ee
This span defines the signal space \( \mathcal{S}_L \), while its orthogonal complement in \( \mathbb{C}^L \) is the noise space \( \mathcal{N}_L \).

Let \( Q_L \) be the orthogonal projector onto \( \mathcal{N}_L \). Since the true signal vectors \( e_L(\varphi_j) \in \mathcal{S}_L \), they are orthogonal to the noise space, and thus satisfy
\be
Q_L e_L(\varphi_j) = 0, \quad j = 1, \dots, K.
\ee
Hence, the parameters \( \varphi_j \) can be recovered as the zeros of the so-called noise-space correlation function
\be
\mathcal{N}_L(\varphi) := \frac{1}{\sqrt{L}} \| Q_L e_L(\varphi) \|_2
\ee
(see Figure \ref{music}(b)) or, equivalently, as the peaks of the MUSIC imaging function
\be
\mathcal{J}_L(\varphi) := \sqrt{L} \cdot \| Q_L e_L(\varphi) \|_2^{-1}.
\ee

In practice, \( Q_L \) is computed via the SVD of \( \mathbf{H}_{L,n-L+1} \):
\be
\mathbf{H}_{L,n-L+1} = U_L D_L W^H,
\ee
where \( U_L = [u_1, \dots, u_L] \in \mathbb{C}^{L \times L} \) contains the left singular vectors. 

Let \( U^{(2)}_{L,L-M} = [u_{M+1}, \dots, u_L] \in \mathbb{C}^{L \times (L-M)} \) be the orthonormal basis of the noise space. Then,
\be
Q_L = U^{(2)}_{L,L-M} (U^{(2)}_{L,L-M})^H,
\ee
and
\be
\mathcal{N}_L(\varphi) = \frac{1}{\sqrt{L}} \left\| (U^{(2)}_{L,L-M})^H e_L(\varphi) \right\|_2.
\ee


\subsection{ESPRIT method}

The Estimation of Signal Parameters via Rotational Invariance Techniques (ESPRIT) is a classical subspace-based method for estimating exponential parameters from sampled data. It builds on the structure of the Hankel matrix associated with the signal and leverages a matrix pencil framework to compute the exponential nodes \cite[Section 10.2.2]{plonka2018numerical}.

Assume that the observed data \( \tilde{\mu}(\ell)=\hat{\mu}(\ell)+\epsilon_\ell \) follow an exponential model of the form \eqref{eq:noisysamples}. As in MUSIC, we construct a rectangular Hankel matrix \( \mathbf{H}_{L,n-L+1} \in \mathbb{C}^{L \times (n-L+1)} \) and perform singular value decomposition (SVD):
\be
\mathbf{H}_{L,n-L+1} = \widetilde{U}_L \widetilde{D}_L \widetilde{W}^H_{n-L+1}.
\ee
The numerical rank \( K \) is determined by examining the decay of the singular values, e.g., by selecting the largest \( K \) such that \( \sigma_K \geq \tau \sigma_1 \), where \( \tau > 0 \) is a chosen tolerance (see Figure \ref{music}(a)).

Let \( \widetilde{W}_{n-L+1,K} \in \mathbb{C}^{(n-L+1) \times M} \) contain the first \( K \) right singular vectors. Removing the last and first rows respectively gives matrices \( \widetilde{W}_{n-L,K}(0) \) and \( \widetilde{W}_{n-L,K}(1) \). The key insight of ESPRIT is that the signal model implies a shift-invariance property between these two matrices, leading to the generalized eigenvalue problem:
\be
z \widetilde{W}_{n-L,K}(0)^H - \widetilde{W}_{n-L,K}(1)^H.
\ee

The ESPRIT nodes \( \{ z_j \} \) are computed as the eigenvalues of the matrix
\be
\widetilde{F}^{\text{SVD}}_K := \widetilde{W}_{n-L,K}(1)^H \left( \widetilde{W}_{n-L,K}(0)^H \right)^{\dagger},
\ee
where \( (\cdot)^{\dagger} \) denotes the Moore–Penrose pseudoinverse. The estimated frequencies \( \varphi_j \) are then recovered via
\be
\varphi_j := \log \left( \frac{z_j}{|z_j|} \right), \quad j = 1, \ldots, K,
\ee
taking the principal branch of the complex logarithm.

Finally, the amplitudes \( \{ c_j \} \) are determined by solving the least squares problem
\be
\min_{c \in \mathbb{C}^K} \left\| V_{n,K}(z) c - \begin{bmatrix} \hat{\mu}(0) & \cdots & \hat{\mu}(n-1) \end{bmatrix}^T \right\|_2,
\ee
where \( V_{n,K}(z) \) is the Vandermonde matrix associated with the recovered nodes \( z = \{ z_j \}_{j=1}^K \).


\subsection{Examples of Prony's, MUSIC, and ESPRIT methods}

Before illustrating the methods with an example, we briefly summarize their key differences. Prony's method is a classical algebraic approach that reconstructs the signal parameters by solving a system of linear equations derived from the data, and it performs best in low-noise settings. MUSIC is a subspace-based method that estimates the signal frequencies by analyzing the orthogonality between signal and noise subspaces via eigen-decomposition of a Hankel matrix. ESPRIT also uses subspace information, but exploits a shift-invariance structure in the signal model to compute frequencies. While MUSIC and ESPRIT generally perform better in noisy environments, they require singular value decompositions and are computationally more demanding.

To illustrate Prony's, MUSIC, and ESPRIT methods, we consider a simple example below:
\begin{equation}\label{eq:exp_3_points_0}
\hat{\mu}(\ell)=5\exp(i\ell)+30\exp(-2i\ell)+20\exp(-2.005i\ell) +\mbox{noise}, \qquad |\ell|<n;
\end{equation}
i.e., $K=3, A_1 = 5, A_2 = 30, A_3 = 20, \lambda_1=-1, \lambda_2=2, \lambda_3=2.005$.

\begin{table}[H]
\begin{center}
\resizebox{0.8\textwidth}{!}{%
\begin{tabular}{|c|c|c|c|c|c|c|c|c|}
\hline
SNR (dB) & Method & $n$ & $\lambda_1$ & $\lambda_2$ & $\lambda_3$ & $A_1$ & $A_2$ & $A_3$ \\
\hline
100 & Prony & 1024 & -1.000 & 2.000 & 2.005 & 5.000 & 29.999 & 20.000\\
100 & MUSIC & 1024 & -0.999 & 2.000 & 2.006 & 4.335 & 26.734 & 16.090\\
100 & ESPRIT & 1024 & -1.000 & 2.000 & 2.005 & 5.000 & 30.000 & 20.000\\
-5 & Prony & 1024 & -0.147 & 2.001 & -2.081 & 17.212 & 0.379 & 0.773\\
-5 & MUSIC & 1024 & -0.999 & 2.000 & 2.006 & 4.631 & 26.215 & 16.081\\
-5 & ESPRIT & 1024 & -1.000 & 2.000 & 2.005 & 4.909 & 29.389 & 19.830\\
\hline
\end{tabular}
}
\end{center}
 	\caption{The table above compares results between Prony, MUSIC, and ESPRIT methods on $n=1024$ (2048 number of samples) on 1 trial.} \label{tab:compare_table_prony}
\end{table}

From the Table \ref{tab:compare_table_prony}, we can see that Prony method performs well in 100 dB SNR case, while MUSIC and ESPRIT methods performs well in both 100 and -5 dB SNR cases. The more extensive results for MUSIC and ESPRIT methods will be discussed later in Table \ref{tab:compare_table_1d}.

\begin{figure}[H]
\begin{center}
\begin{minipage}{0.55\textwidth}
\subfloat[]{
\includegraphics[scale=.135]{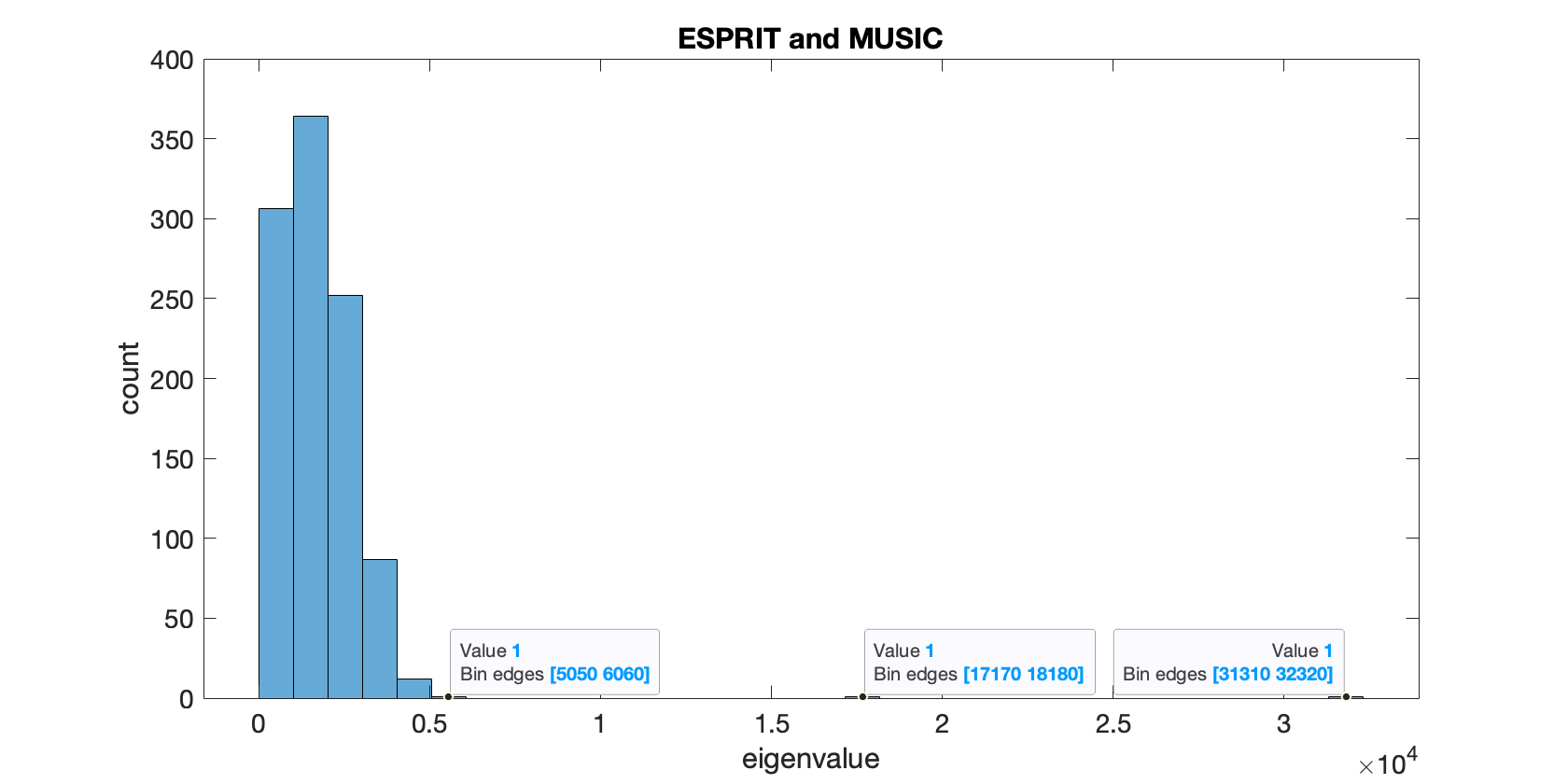}
}
\end{minipage}
\begin{minipage}{0.35\textwidth}
\subfloat[]{
\includegraphics[scale=.135]{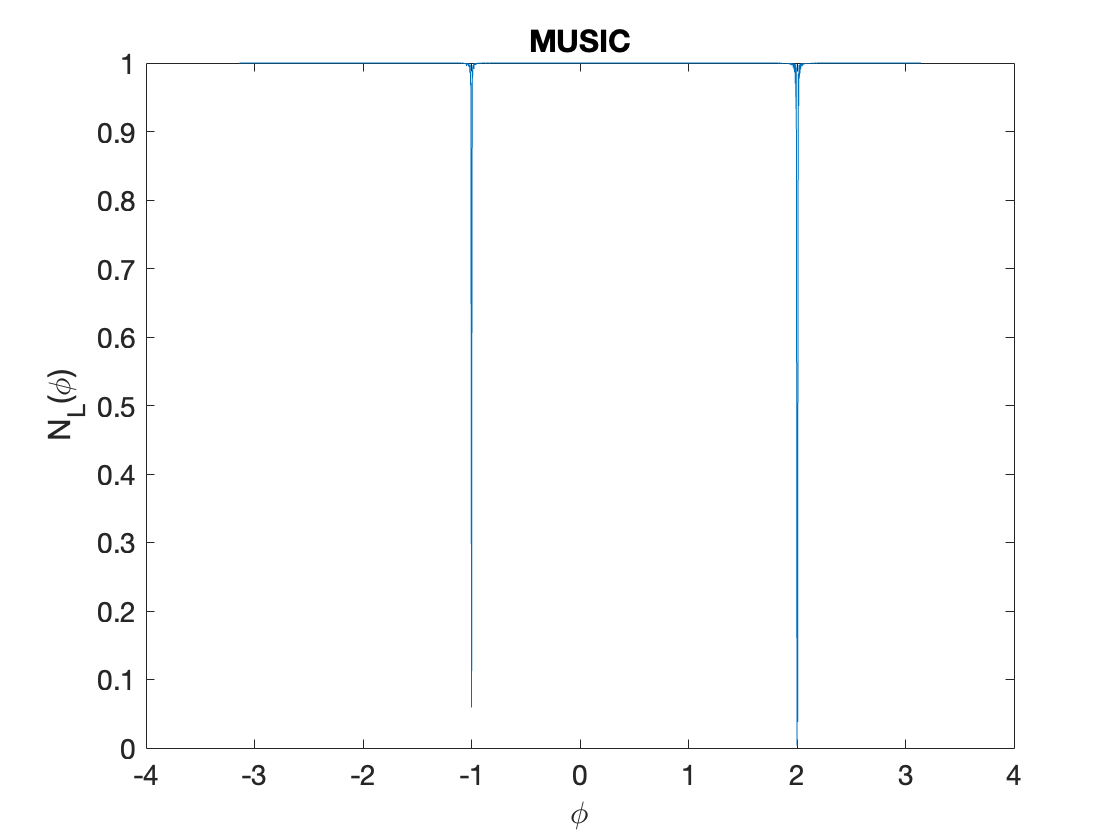}
}
\end{minipage}
\end{center}
\caption{(a) The histogram of eigenvalues of $\mathbf{H}_{L,n-L+1}$. (b) The plot of $\mathcal{N}_L(\varphi)$ on the y-axis. The value of $\mathcal{N}_L(\varphi) \approx 0$ indicates the approximation of $\lambda \approx \varphi$.}
\label{music}
\end{figure}

While Prony's, MUSIC, and ESPRIT methods form the foundation of parameter estimation in univariate exponential sums, their extensions to higher dimensions are not straightforward and their performance degrades in the presence of low SNR.

\vspace{2ex} \newpage

The rest of this dissertation is organized as follow:
\begin{itemize}
\item Chapters \ref{section:mainresult} introduces the use of localized trigonometric kernels for exponential parameter estimation in univariate settings. We establish theoretical guarantees for accuracy and robustness under sub-Gaussian noise.
\item Chapters \ref{1paper} extends the localized kernel method to the multidimensional exponential analysis problem. We use coordinate-wise projections, apply our algorithm, and demonstrate its effectiveness in synthetic imaging applications.
\item Chapters \ref{2paper} applies the localized kernel method to the separation of linear chirp components in noisy signals. We construct a signal separation operator (SSO) and design algorithms to estimate instantaneous frequencies without prior knowledge of the number of components.
\item Chapters \ref{summary} concludes with a summary of contributions and outlines several directions for future work, including nonlinear chirps, adaptive kernel design, and applications to classification tasks.
\end{itemize}


\chapter{Exponential Parameter Estimation}\label{section:mainresult}

The goal of this chapter is to develop a robust and tractable method for estimating the parameters of a finite sum of exponential signals from noisy observations. Specifically, we aim to recover the number of components, their frequencies, and amplitudes. Traditional methods, such as Prony's method or subspace-based techniques like MUSIC and ESPRIT, face limitations in the presence of noise. In contrast, our approach builds on the use of \emph{localized trigonometric kernels} to isolate and estimate signal parameters without requiring matrix decompositions, enabling accurate parameter recovery under sub-Gaussian noise.

This dissertation is based on localized trigonometric polynomial kernels developed in \cite{loctrigwave}. We introduce our localized trigonometric kernels in Section~\ref{bhag:localized_kernels} and reconstruction operator in Section~\ref{bhag:sso_0}. Our main theorem is stated in Section~\ref{bhag:theorem}, and proved in Section~\ref{bhag:proof}.
Section~\ref{bhag:1d} discusses the numerical implementation of this theorem.

\newpage

\bhag{Localized trigonometric kernels}\label{bhag:localized_kernels}

The main idea is to use a low pass filter and the corresponding localized kernel $\Phi_n$. 
A \textit{low pass filter} is an infinitely differentiable function $H:\RR\to [0,1]$ such that $H(t)=H(-t)$ for all $t\in\RR$, $H(t)=1$ if $|t|\le 1/2$, and $H(t)=0$ if $|t|\ge 1$.

With a low pass filter $H$, we now define
 \yadi{$H$}{smooth low pass filter Section~\ref{section:theosect}B} \yadi{$\Phi_n$}{Localized kernel of degree $n$} \yadi{$\sigma_n$}{reconstruction operator, Section~\ref{section:theosect}} \yadi{$\hbar_n$}{Normalizing constant, Eqn. \eqref{eq:lockerndef}} 
\begin{equation}\label{eq:lockerndef}
\hbar_n=\left\{\sum_{|\ell|<n}H\left(\frac{|\ell|}{n}\right)\right\}^{-1}, \quad \Phi_n(x)=\hbar_n\sum_{\ell\in\mathbb{Z}}H\left(\frac{|\ell|}{n}\right)e^{i\ell x}, \qquad
x\in\TT, \ n>0.
\end{equation}
We note that the normalizing factor $\hbar_n$ is chosen so that
\begin{equation}\label{eq:phimax}
\max_{x\in \TT}|\Phi_n(x)|=\Phi_n(0)=1.
\end{equation}
An important property of $\Phi_n$ is the following localization estimate given in \cite{singdet}. 

Let
\be
H_1 := \|H\|_1 = \int_{-1}^1 H(t) \, dt > 0.
\ee
Since \( H(t) \in [0, 1] \) for all \( t \in [-1, 1] \), it follows that \( H_1 \leq 2 \).

The following Theorem \ref{theo:fundatheorem} gives the upper bound of $|\Phi_n(x)|$.

\begin{theorem}\label{theo:fundatheorem}
Let $H$ be a smooth low pass filter, $p\ge 1$ and $S\ge 2$ be integers. 
We have\\
{\rm (a)} 
\begin{equation}
\label{eq:hsums}
\left|\sum_{k=-n}^{n} H\left( \frac{k}{n} \right)^p - n\int_{-1}^{1} H(x)^p ~dx\right|\le \frac{3p}{12} \cdot \frac{\| H'' \|_1}{n}.
\end{equation}
In particular, if 
\begin{equation}\label{eq:ncond}
n\ge\sqrt{\frac{p\|H''\|_1}{\|H^p\|_1}},
\end{equation}
then
\begin{equation}
\label{eq:hsums_asymp}
\left|\sum_{k=-n}^{n} H\left( \frac{k}{n} \right)^p - n\int_{-1}^{1} H(x)^p~dx\right|\le (n/4)\int_{-1}^{1} H(x)^p~dx.
\end{equation}
In particular, if $p=1$,
\begin{equation} \label{eq:order_of_mag}
\left|\frac{1}{n\hbar_n\|H\|_1} - 1\right|\le \frac{1}{4}.
\end{equation}
Hence, the lower and upper bounds of $n\hbar_n$ are
\begin{equation} \label{eq:upper_bound}
\frac{4}{5\|H\|_1} \leq n\hbar_n \le \frac{4}{3\|H\|_1}.
\end{equation}
{\rm (b)} For $x\in [-\pi,\pi]$, $x\not=0$, $n\ge 1$, we have
\begin{equation}
\label{eq:phi_awaybd}
|\hbar_n^{-1}\Phi_n(x)| \le \frac{\|H^{(S)}\|_1\sqrt{2\pi}n}{(n|x|)^S} \cdot \frac{7}{2}.
\end{equation}
In particular, if $n$ satisfies \eqref{eq:ncond} with $p=1$, then
\begin{equation}
\label{eq:locest}
|\Phi_n(x)|\le \frac{14\sqrt{2\pi}}{3} \cdot \frac{\|H^{(S)}\|_1}{\|H\|_1} \cdot \frac{1}{\max(1,(n|x|)^S)}.
\end{equation}
\end{theorem}
The following proposition summarizes some facts which we will need in the proof.
We recall that the Fourier transform (respectively, inverse Fourier transform) of $f$ is defined by
\begin{equation}
\label{eq:fourtrans}
\hat{f}(\omega)=\frac{1}{\sqrt{2\pi}}\int_{\mathbb{R}}f(t)\exp(it\omega)dt,\qquad \check{f}(t)= \frac{1}{\sqrt{2\pi}}\int_{\mathbb{R}}f(\omega)\exp(it\omega)d\omega.
\end{equation}
 \begin{prop}\label{prop:basicprop}
 {\rm (a)} If $f :\mathbb{R}\to \mathbb{R}$ is twice continuously differentiable and $n\ge 1$ is an integer, then 
\begin{equation}
\label{eq:euler}
\left|\sum_{|k|\le n}f(k)-\int_{-n}^n f(t)dt\right|\le \frac{3}{24}\int_{-n}^n |f''(t)|dt.
\end{equation}
{\rm (b)} If $\check{f}$ is integrable and has bounded total variation then
\begin{equation}
\sum_{k\in\mathbb{Z}}f(k)\exp(ikx)=\sqrt{2\pi}\sum_{\ell\in\mathbb{Z}}\check{f}(x+2\ell\pi).
\end{equation}
 \end{prop}
 \begin{proof}
 Part (a) is a consequence of the Euler-Maclaurin formula (cf. \cite[Formulas[Chapter 8, (2.01), (2.04)]{olverbook}.
 Part (b) is the Poisson summation formula \cite[Proposition~5.1.30]{butzer1971fourier}.
 \end{proof}

\noindent\textsc{Proof of Theorem~\ref{theo:fundatheorem}. }

In order to prove part (a), we use Proposition~\ref{prop:basicprop}(a) with $H(\cdot/n)$ in place of $f$ and recall that $H(1)=H(-1)=0$ to conclude that
\begin{equation}
\label{eq:pf0eqn1}
\left| \frac{1}{n} \sum_{k=-n}^{n} H\left( \frac{k}{n} \right)^p - \int_{-1}^1 H(x)^p \, dx \right|\le \frac{3}{24n^2}\int_{-1}^1 |(H^p)''(t)|dt.
\end{equation}
If $p=1$, this proves \eqref{eq:hsums}. 
Let $p\ge 2$.
 Using integration by parts, we see that
\be
(p-1)\int_{-1}^1 H^{p-2}(t)H'(t)^2dt=-\int_{-1}^1 H^{p-1}(t)H''(t)dt.
\ee
Since $(H^p)''(t)=pH^{p-1}(t)H''(t)=p(p-1)H^{p-2}(t)H'(t)^2$, we conclude that
\be
\int_{-1}^1 |(H^p)''(t)|dt\le 2p\|H''\|_1.
\ee
Together with \eqref{eq:pf0eqn1}, this leads to \eqref{eq:hsums} also in the case when $p\ge 2$.

In order to prove part (b), we use Proposition~\ref{prop:basicprop}(b) with $f(t)=H(t/n)$, so that $\check{f}(x)=n\check{H}(nx)$.
Hence,
\begin{equation}
\label{eq:pf0eqn2}
\hbar_n^{-1}\Phi_n(x)=\sum_{|k|<n}H(k/n)\exp(ikx)=\sqrt{2\pi}n\sum_{\ell\in\mathbb{Z}}\check{H}(n(x+2\ell\pi)).
\end{equation}
Since $h$ is $S$ times continuously differentiable, then an integration by parts shows that
\be
\check{H}(nx)=\frac{(-1)^S}{(nix)^S}\check{H^{(S)}}(nx),
\ee
and hence
\begin{equation}
\label{eq:pf0eqn3}
|\check{H}(n(x+2\ell\pi))|\le \frac{\|H^{(S)}\|_1}{n^S|x+2\ell\pi|^S}
\end{equation}
If $|x|<\pi$, and $\ell\not=0$, then $|x+2\ell\pi|\ge 2|\ell|\pi-|x|\ge (2|\ell|-1)\pi\ge (2|\ell|-1)|x|$.
So, 
\begin{equation}
\label{eq:pf0eqn4}
|\hbar_n^{-1}\Phi_n(x)| \le \frac{\sqrt{2\pi}n\|H^{(S)}\|_1}{(n|x|)^S}\left\{1+2\sum_{\ell=1}^\infty (2\ell-1)^{-S}\right\}.
\end{equation}
We have
\be
\sum_{\ell=1}^\infty (2\ell-1)^{-S}=\sum_{\ell=1}^\infty \ell^{-S}-(1/2^S)\sum_{\ell=1}^\infty \ell^{-S}=(1-2^{-S})\zeta(S),
\ee
where $\zeta(S)$ is the Riemann zeta function. Note that when \( S = 2 \), 
\be
\left( 1 + 2(1 - 2^{-S}) \zeta(S) \right) = \left( 1 + 2(1 - 1/4) \frac{\pi^2}{6} \right) \approx 3.4674 < \frac{7}{2}.
\ee
As $S \rightarrow \infty$, we have $\zeta(S) \rightarrow 1$ and $2^{-S} \rightarrow 0$. Thus,
\be
\lim_{S\rightarrow \infty}\left( 1 + 2(1 - 2^{-S}) \zeta(S) \right) = \left( 1 + 2(1 - 0) (1) \right) = 3.
\ee
Hence, it holds that
\be \label{eq72}
\left( 1 + 2(1 - 2^{-S}) \zeta(S) \right) < \frac{7}{2}.
\ee
Together with \eqref{eq:pf0eqn4}, this leads to
\be
|\hbar_n^{-1}\Phi_n(x)| \le \frac{\sqrt{2\pi}\|H^{(S)}\|_1n}{(n|x|)^S}(1+2(1-2^{-S}))\zeta(S) < \frac{\sqrt{2\pi}\|H^{(S)}\|_1n}{(n|x|)^S} \cdot \frac{7}{2}.
\ee
This proves \eqref{eq:phi_awaybd}.
Additionally,
\bea
H(t) &=& \int_{-1}^{t} H'(u) \, du \nonumber \\
&=& \int_{-1}^{t} \left( \int_{-1}^{u} H''(v) \, dv \right) du \qquad -1 \leq v \leq u \leq t \nonumber \\
&=& \int_{-1}^{t} (t-v)H''(v) \, dv \nonumber \\
&=& \cdots ~=~ \frac{1}{(S-1)!} \int_{-1}^{1} (t-v)_+^{S-1}H^{(S)}(v) \, dv. 
\eea
Thus,
\be
\int_{-1}^1 H(t) ~dt = \frac{1}{(S-1)!} \int_{-1}^{1} \left( \int_{-1}^1 (t-v)_+^{S-1} ~dt \right) H^{(S)}(v) \, dv.
\ee
Since \( (t - v)_+^{S-1} = 0 \) for \( t < v \), we simplify the inner integral:
\be
\int_v^1 (t - v)^{S-1} \, dt = \int_0^{1 - v} u^{S-1} \, du = \frac{(1 - v)^S}{S}
\ee
Hence,
\be
\int_{-1}^1 H(t) \, dt = \frac{1}{S!} \int_{-1}^1 (1 - v)^S H^{(S)}(v) \, dv
\ee
To estimate this, note that \( (1 - v)^S \le 2^S \) for \( v \in [-1, 1] \). Therefore,
\be
\int_{-1}^1 H(t) \, dt \le \frac{2^S}{S!} \int_{-1}^1 H^{(S)}(v) \, dv.
\ee
It follows that
\be \frac{\|H^{(S)}\|_1}{\|H\|_1 } \geq \frac{2^S}{S!} \geq 1. \ee
This to together with \eqref{eq:upper_bound} and \eqref{eq:phi_awaybd} prove \eqref{eq:locest}, and completes the proof. \qed

Clearly, the parameter $S$ controls the localization of the kernel, with higher smoothness of $H$ implying better localization. 
We refer to \cite{eugenenevai} for a more detailed discussion of the effect of $S$.
The estimate \eqref{eq:locest} implies that $\Phi_n(x)$ is an approximation to the Dirac delta $\delta_0$.
Figure~\ref{fig:kernloc} illustrates the localization property of the kernel.

There are, of course, many possible choices for a low pass filter. 
The actual choice is not so critical for our theory.
In the rest of this dissertation, we fix a low pass filter. 
All constants may depend upon this filter, and the filter will be omitted from the notation.

\begin{figure}[H]
\begin{center}
\begin{minipage}{0.49\textwidth}
\includegraphics[width=\textwidth]{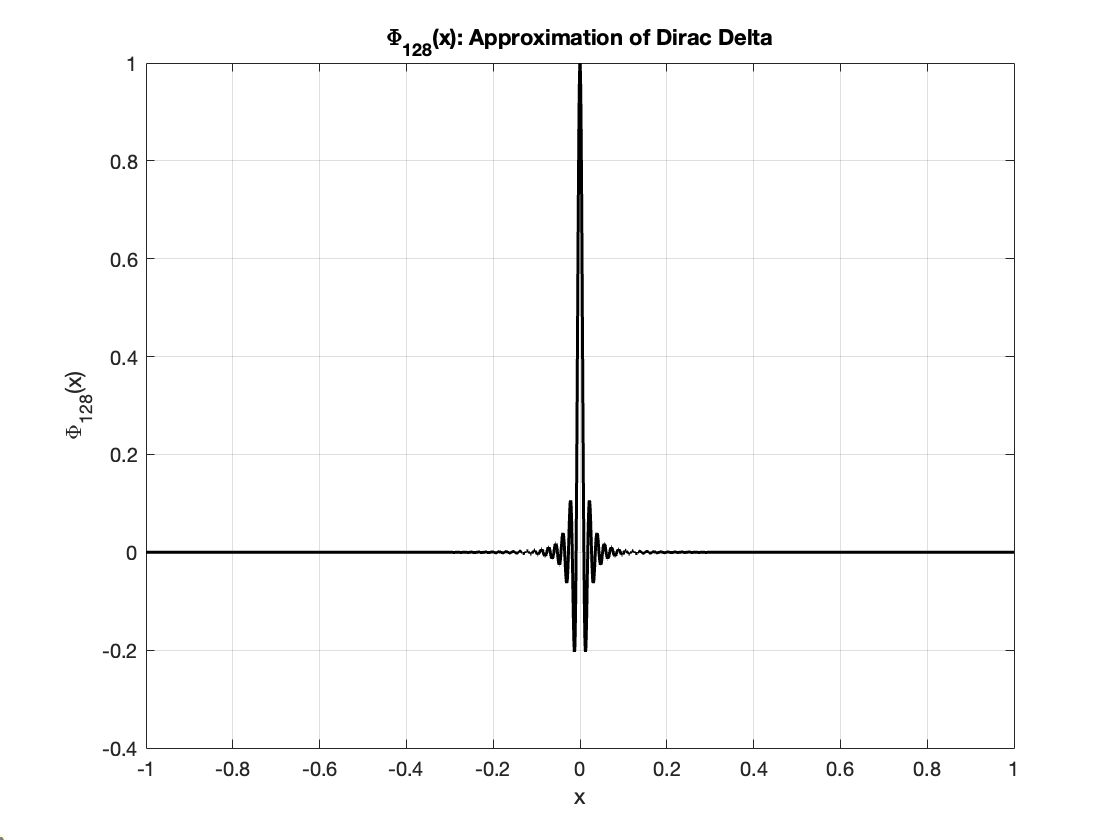} 
\end{minipage}
\begin{minipage}{0.49\textwidth}
\includegraphics[width=\textwidth]{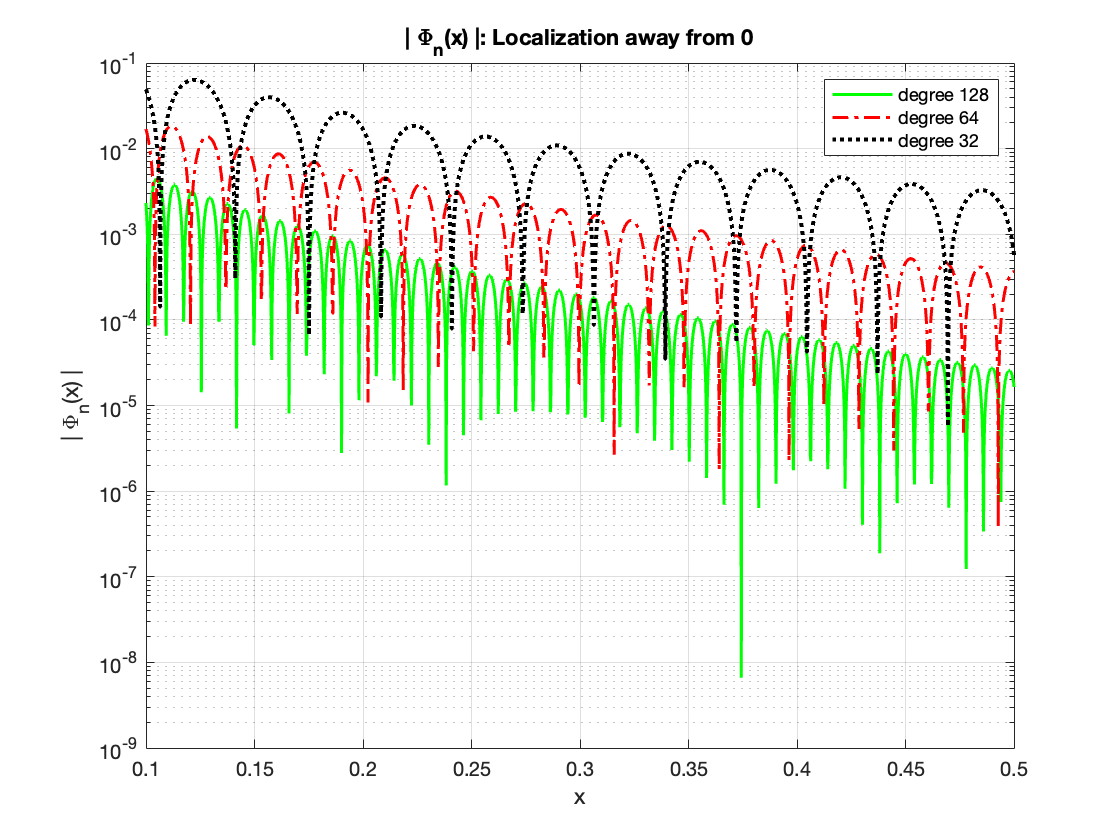} 
\end{minipage}
\end{center}
\caption{Illustration for the localization of the kernels $\Phi_n$. 
In both figures $x$-axis is normalized by a factor of $\pi$. 
Left: Graph of $\Phi_{128}(x)$, clearly demonstrating the approximation of Dirac delta at $0$, 
Right: Illustration of how rapidly $|\Phi_n(x)|$ decreases away from $0$ as $n$ increases.}
\label{fig:kernloc}
\end{figure}

\bhag{Reconstruction operator}\label{bhag:sso_0}

\noindent\textbf{Constant convention}

\vspace{2ex}

\textit{
The letters $c, c_1,\cdots$ will denote generic positive constants depending on $H$ and $S$ alone. Their values might be different at different occurrences within  single formula. 
The notation $A\lesssim B$ means $A\le c B$, $A\gtrsim B$ means $B\lesssim A$, and $A\sim B$ means $A\lesssim B\lesssim A$. 
Constants denoted by capital letters, such as $L$, $C$, etc. will retain their values.
}

The localization property \eqref{eq:locest} implies that 
\begin{equation}\label{eq:purespectrum}
\sigma_n(\hat{\mu})(x)=\hbar_n\sum_{|\ell|<n}H\left(\frac{|\ell|}{n}\right)\hat{\mu}(\ell)\exp(i\ell x) 
=\sum_{k=1}^K A_k\Phi_n(x-\lambda_k)\approx \mu(x).
\end{equation}

Thus, the ``prominent'' peaks of $|\sigma_n(\hat{\mu})(x)|$ will occur at (or close to) the points $\lambda_k$ and the value of $\sigma_n$ at these points lead to the corresponding complex amplitudes $A_k$. 
It turns out that the values of $A_k$ are actually determined fairly accurately simply by evaluating $\sigma_n(\mu)(x)$ at the estimated value of $\lambda_k$.

The main difficulty is to make this more precise, and quantify the approximation error in terms of $n$ and the properties of the noise.

\bhag{Main theorem}\label{bhag:theorem}
The purpose of this section is to make precise the idea as to how the use of the localized kernel facilitates the solution of the point source separation problem, giving precise error bounds.
We assume the notation in \eqref{eq:noisysamples}, and make further notation as follows.
Let
\begin{equation}\label{eq:notation}
M=\sum_{k=1}^K |A_k|, \ \ \mathfrak{m}=\min_{1\le k\le K}|A_k|, \ \ \eta=\min_{\ell\not=k} |\lambda_k-\lambda_\ell|,
\end{equation}
We further assume that each $\epsilon_\ell$ is a realization of a sub-Gaussian random variable in $\mathcal{G}(V)$.
\yadi{$M$}{sum of absolute values of coefficients} \yadi{$\mathfrak{m}$}{minimum absolute value of coefficients}
\yadi{$\eta$}{minimal separation among unitless frequencies} \yadi{$E_n$}{reconstruction operator used with noise alone Eqn. \eqref{eq:noisespectrum}}
We will write
\begin{equation}\label{eq:powerspectrum}
\sigma_n(x)=\sigma_n(\hat{\mu})(x) =\hbar_n\sum_{|\ell|<n}H\left(\frac{|\ell|}{n}\right)\tilde{\mu}(\ell)\exp(i\ell x), \quad x\in\TT.
\end{equation}
Writing
\begin{equation}\label{eq:noisespectrum}
E_n(x)=\sigma_n(\{\epsilon_\ell\})(x) = \hbar_n\sum_{|\ell|<n}H\left(\frac{|\ell|}{n}\right)\epsilon_\ell \exp(i\ell x), \quad x\in\TT,
\end{equation}
we observe that
\begin{equation}\label{eq:powerspectrum_bis}
\sigma_n(\tilde{\mu})(x)=\sigma_n(\hat{\mu})(x)+E_n(x) =\sum_{k=1}^K A_k\Phi_n(x-\lambda_k)+E_n(x).
\end{equation}
With this set up, our main theorem concerning the recuperation of point sources can be stated as follows. \yadi{$\mathbb{G}$}{Support of the thresholded power spectrum} \yadi{$\mathbb{G}_\ell$}{Cluster definded in Theorem~\ref{theo:main}} \yadi{$C$}{Eqn. \eqref{eq:thresholdCdef}}

\newpage
\begin{theorem}\label{theo:main}
Let $0<\delta<1$,
\begin{equation}\label{eq:levelsetdef}
\mathbb{G}=\{x\in [-\pi,\pi] : |\sigma_n(x)|\ge \mathfrak{m}/2\}.
\end{equation}
and (cf. \eqref{eq:locest}, \eqref{eq:notation})
\begin{equation}\label{eq:thresholdCdef}
C = \max\left(1,\left( \frac{16ML}{\mathfrak{m}} \right)^{1/S}\right).
\end{equation}
For sufficiently large $n$ (cf. \eqref{eq:pf2eqn3}), each of the following statements holds with probability exceeding $1-\delta$.
\begin{itemize}
\item (\textbf{Disjoint union condition}) \\
the set $\mathbb{G}$ is a disjoint union of exactly $K$ subsets $\mathbb{G}_\ell$,
\item (\textbf{Diameter condition}) \\
for each $\ell=1,\cdots, K$,   $\mathsf{diam}(\mathbb{G}_\ell) \le 2C/n$,
\item (\textbf{Separation}) \\
$\mathsf{dist}(\mathbb{G}_\ell, \mathbb{G}_k) \ge \eta/2$ for $\ell \neq k$,
\item (\textbf{Interval inclusion}) \\
For each $\ell=1,\cdots, K$,  $I_\ell=\{x\in\TT: |x-\lambda_\ell|\le 1/(4n)\}\subseteq \mathbb{G}_\ell$.
 
\end{itemize}
Moreover, if
\begin{equation}\label{eq:lambda_estimator}
\hat{\lambda}_\ell =\arg\max_{x\in \mathbb{G}_\ell}|\sigma_n(x)|,
\end{equation}

then

\begin{equation}\label{eq:lambdaerr}
|\hat{\lambda}_\ell-\lambda_\ell| \le 2C/n.
\end{equation}

\end{theorem}
\begin{rem} \label{rem:eta_and_m}
{\rm Equation~\eqref{eq:pf2eqn3} shows that the value of $n$ should be at least proportional to the reciprocal of the minimal separation, $\eta$, among the $\lambda_\ell$'s. There is no direct connection with the number $K$ of signals, except the fact that there can be at most $\O(\eta^{-1})$ values of $\lambda_\ell$ in $[-\pi,\pi]$. 
For example, our method would require a much larger value of $n$ if there are two signals which are very close by compared to the case when they are far apart. 
Equation~\eqref{eq:pf2eqn3} also demonstrates the connection between $n$, the noise variance $V$, the minimal amplitude $\mathfrak{m}$, and the desired level of certainty $1-\delta$. Since the value of $\mathfrak{m}$ is unknown, we estimate the threshold $\tau$ ($3\mathfrak{m}$ in \eqref{eq:levelsetdef}) by taking a percentile  of the histogram of $|\sigma_n(x)|$. 
As a result, any signal with amplitude less than $\tau$ will not be detected by our algorithm. \qed}
\end{rem}

\bhag{Illustration in a univariate case}\label{bhag:1d}

We implement Theorem~\ref{theo:main} in the univariate case using Algorithm~\ref{alg:univariate}.
 \begin{algorithm}[ht]
 \begin{algorithmic}[1]
 \item[{\rm a)}] \textbf{Input:} The signal $\{\tilde{\mu}(\ell)\}_{|\ell|<n}$, threshold $\tau$, and a guess $\eta$ for the minimal separation. 
 \item[{\rm b)}] \textbf{Output:} Estimation of $\hat{A}_k, \hat{\lambda}_k$, and $\hat{\theta}_k$ for $k = 1, \ldots, K$.
 \STATE $\hbar_n \gets \left\{\sum_{|\ell|<n}H\left(\frac{|\ell|}{n}\right)\right\}^{-1}$
 \STATE $\sigma_n(x) \gets \hbar_n \sum_{|\ell|<n} H\left(\frac{|\ell|}{n}\right)\tilde{\mu}(\ell)e^{i\ell x}$
 \STATE $\mathcal{G} \gets \{x\in [-\pi,\pi] : |\sigma_n(x)|\ge \tau\}$
 \STATE $\mathcal{G}_1,\ldots,\mathcal{G}_K \gets Partition(\mathcal{G}) \mbox{ with minimal separation } \eta/4$
  \FOR{ $k=1$ to $K$}
 \STATE $\hat{\lambda}_k \gets \arg\max_{x\in \mathcal{G}_k} (|\sigma_n(x)|)$ 
 \STATE $\hat{\theta}_k \gets Phase(\sigma_n(x))$ 
 \STATE $\hat{A}_k \gets |\sigma_n(\hat{\lambda}_k)|$ 
 \ENDFOR
 \item[] \textbf{Note:} step 3 - 8 can be computed by using \texttt{findpeaks} in MATLAB with parameters \texttt{MinPeakDistance} $\eta/4$ and \texttt{MinPeakHeight} $\tau$.
 \STATE \textbf{Return: } $\hat{A}_k, \hat{\theta}_k, \hat{\lambda}_k$
 \caption{Given a univariate signal \\$\hat{\mu}(\ell) =\sum_{k=1}^K A_k\exp(-i\ell\lambda_k),$ find $K$, $A_k$'s and $\lambda_k$'s.}
 \label{alg:univariate}
 \end{algorithmic}
 \end{algorithm}
 
To illustrate this algorithm and Theorem~\ref{theo:main}, we consider a simple example:
\begin{equation}\label{eq:exp_3_points}
\hat{\mu}(\ell)=5\exp(i\ell)+30\exp(-2i\ell)+20\exp(-2.005i\ell) +\mbox{noise}, \qquad |\ell|<n;
\end{equation}
i.e., $K=3, A_1 = 5, A_2 = 30, A_3 = 20, \lambda_1=-1, \lambda_2=2, \lambda_3=2.005$.
We note that the amplitude $A_1$ is substantially smaller than the other two amplitudes, and $\lambda_2$ is very close to $\lambda_3$. 
Figure~\ref{fig:signal_and_noise} shows the original signal, the signal corrupted with a 0dB noise, and the detection of the $\lambda$'s with $n=1024$. 
The estimated frequencies are  $-0.9999,
   1.9997, 
   2.0058$.
 
\vspace{2ex}

\noindent\textbf{Performance assessment:}

\vspace{2ex}

After finding the estimate $\hat{\lambda}_k$ for $\lambda_k$ for each $k$, the next challenge is to determine how many of these estimates represent the actual points. Throughout the experiments in Table \ref{tab:compare_table_1d}, we evaluated our approach using 16 trials. For each trial, the RMSE was calculated by the formula
\be\mbox{RMSE} = \sqrt{\frac{1}{K} \sum_{k=1}^K |\lambda_k-\hat{\lambda}_k|^2}.\ee
We then reported the number of points reconstructed, run-time, memory, accuracy, and RMSE by taking average over 16 trials. Finally, we computed the standard deviation of the RMSE.

\vspace{2ex}

\begin{figure}[H]
\begin{center}
\begin{minipage}{0.32\textwidth}
\begin{center}
\subfloat[]{
\includegraphics[width=0.95\textwidth]{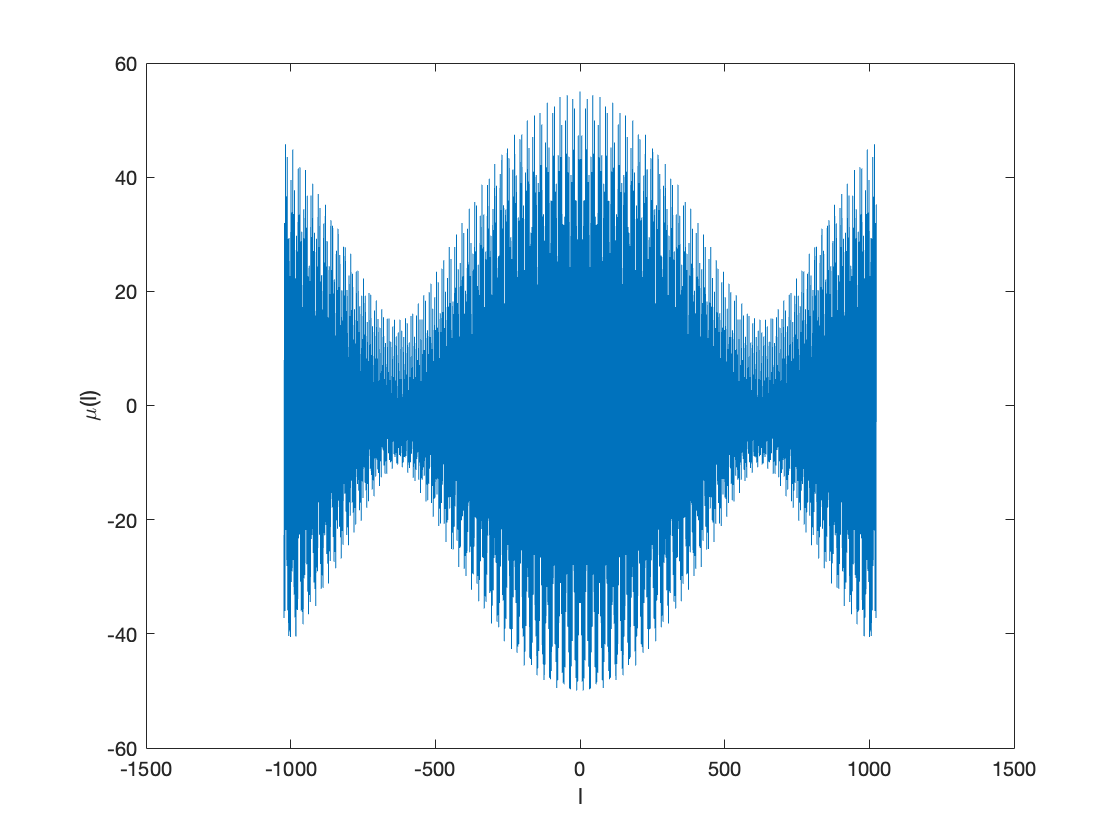}
}
\end{center}
\end{minipage}
\begin{minipage}{0.32\textwidth}
\begin{center}
\subfloat[]{
\includegraphics[width=0.95\textwidth]{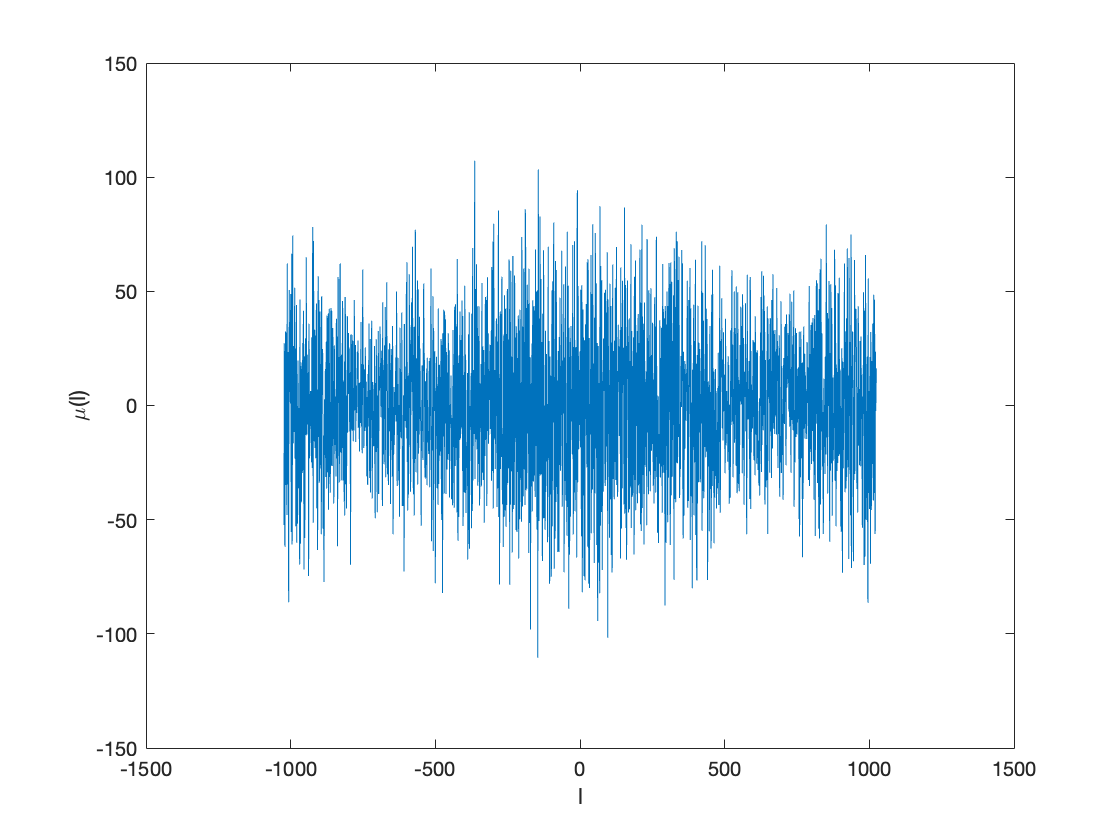}
}
\end{center}
\end{minipage}
\begin{minipage}{0.32\textwidth}
\begin{center}
\subfloat[]{
\includegraphics[width=0.95\textwidth]{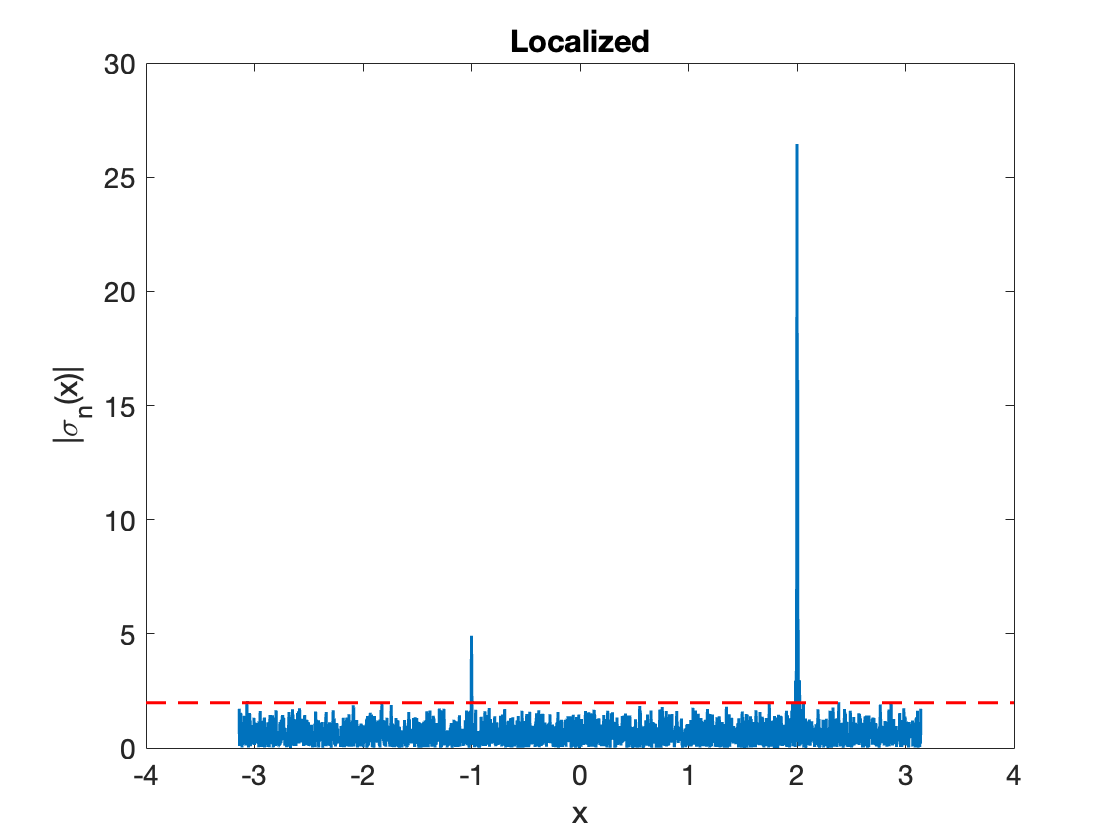}
}
\end{center}
\end{minipage}
\caption{(a) Signal $|\hat{\mu}(\ell)|$ as in \eqref{eq:exp_3_points} without noise $\hat{\mu}(\ell)$, $|\ell|<1024$. (b) Signal $|\tilde{\mu}(\ell)|$with 0 dB noise. (c) $|\sigma_{1024}(x)|$ with threshold $\tau=2.5$ at the red line. The detected values of $\lambda$ are  $-0.9999,
   1.9997, 
   2.0058$.}
\label{fig:signal_and_noise}
\end{center}
\end{figure}

Next, 
we compare the results of our localized kernel algorithm with ESPRIT and MUSIC algorithms (see \cite{plonka2018numerical}, pp. 576–586) on noisy variants of the one dimensional signal defined in \eqref{eq:exp_3_points} with different levels of noise and values of $n$.
We see that with $n=1024$, ESPRIT gives the best RMSE up to a noise level of -5dB, but our method gives better results when the SNR falls below -10dB. 
The runtime for both ESPRIT and MUSIC is much more than ours. 
In fact, they both require more than 5 minutes if $n\ge 16384$.
On the other hand, our method gives a better performance even at -15dB SNR if only we can increase $n$.

\begin{table*}[ht]
\begin{center}
\resizebox{\textwidth}{!}{%
\begin{tabular}{|c|c|c|c|c|c|c|c|c|}
\hline
SNR & Method & $n$ & Total & Recuperated & Run-time & Memory & RMSE & Standard \\
(dB) & &  & points & points & (seconds) & (MB) & & Deviation\\
\hline
-15 & Localized & 32768 & 3 & 3 & 1.52e-02 & 6.5043 & 2.28e-05 & 0\\
-15 & Localized & 16384 & 3 & 3 & 9.12e-03 & 3.2543 & 2.50e-01 & 2.50e-01\\
-10 & Localized & 16384 & 3 & 3 & 1.01e-02 & 3.2543 & 6.53e-05 & 1.40e-20\\
-5 & Localized & 16384 & 3 & 3 & 9.36e-03 & 3.2543 & 6.53e-05 & 1.40e-20\\
-10 & ESPRIT & 1024 & 3 & 3 & 5.13e-01 & 0.2075 & 8.10e-01 & 5.42e-01\\
-10 & MUSIC & 1024 & 3 & 3 & 6.23e-01 & 0.2075 & 8.12e-01 & 5.42e-01\\
-10 & Localized & 1024 & 3 & 3 & 1.29e-02 & 0.2075 & \textbf{5.91e-01} & 4.05e-01\\
-5 & ESPRIT & 1024 & 3 & 3 & 5.36e-01 & 0.2075 & 6.95e-02 & 6.89e-02\\
-5 & MUSIC & 1024 & 3 & 3 & 6.41e-01 & 0.2075 & \textbf{6.94e-02} & 6.86e-02\\
-5 & Localized & 1024 & 3 & 3 & 1.19e-02 & 0.2075 & 1.77e-01 & 1.77e-01\\
0 & ESPRIT & 1024 & 3 & 3 & 5.34e-01 & 0.2075 & \textbf{9.43e-05} & 4.27e-05\\
0 & MUSIC & 1024 & 3 & 3 & 6.31e-01 & 0.2075 & 6.77e-04 & 0\\
0 & Localized & 1024 & 3 & 3 & 1.10e-02 & 0.2075 & 4.09e-04 & 3.12e-05\\
5 & ESPRIT & 1024 & 3 & 3 & 5.43e-01 & 0.2075 & \textbf{5.10e-05} & 2.29e-05\\
5 & MUSIC & 1024 & 3 & 3 & 6.39e-01 & 0.2075 & 6.77e-04 & 0\\
5 & Localized & 1024 & 3 & 3 & 1.19e-02 & 0.2075 & 3.86e-04 & 1.93e-05\\
\hline
\end{tabular}
}
\end{center}
 	\caption{The table above compares results between our algorithm, MUSIC, and ESPRIT on $n=1024$ (2048 number of samples). Note that for $n=16384, 32768$, ESPRIT and MUSIC algorithms don't finish within 5 minutes since the complexity of the algorithms rely on eigenvalues finding algorithm.} \label{tab:compare_table_1d}
\end{table*}

In Figure~\ref{fig:fail_case}, we examine the dependence of our algorithm on the choice of the threshold $\tau$. 
It is clear from Figure~\ref{fig:fail_case}(a) that if $\tau$ is too high (e.g., $\tau=5$), we will not detect the signal with the smallest amplitude even in the absence of noise.
When noise is present, our method results in a reduction of noise as indicated by Lemma~\ref{lemma:noiselemma} below.
We may then set the threshold $\tau$ by examining the power spectrum $|E_n(x)|$ (cf. \eqref{eq:noisespectrum}).
In our example, we chose this to be $1.01\times\max_x|E_n(x)|$.
Figure~\ref{fig:fail_case}(b) shows that for a small value of $n$, this threshold is too high to detect the smallest amplitude signal.
The noise is suppressed even more for a larger value of $n$, which results in the detection of all the frequencies again, as shown in Figure~\ref{fig:fail_case}(c).

Finally, Figure~\ref{fig:fail_close_case} illustrates the dependence of $n$ on the minimal separation $\eta$, as well as the effect of choosing the threshold too small. 
For this purpose, we focus on the power spectrum $|\sigma_n(x)|$ in the noiseless case. 
Figure~\ref{fig:fail_close_case}(a) shows that the frequencies $2$ and 2.005 are not resolved correctly with $n=512$.
Figure~\ref{fig:fail_close_case}(b) shows that they are resolved if $n=1024$ if the threshold is set at the right level at 2.5, and $\eta=0.004$.
If the threshold is too low, then the minimal separation $\eta$ controls the number of signals detected.
If $\eta=0.004$, we detect three frequencies as reported earlier;  if $\eta=0.001$, we detect correctly the  frequency $-1$, but the sidelobes near $2$ are detected as 11 frequencies near $2$.
Figure~\ref{fig:fail_close_case}(c) illustrates the suppression of sidelobes to allow a greater leeway in setting these parameters.

\begin{figure}[H]
\begin{center}
\begin{minipage}{0.32\textwidth}
\begin{center}
\subfloat[]{
\includegraphics[width=0.95\textwidth]{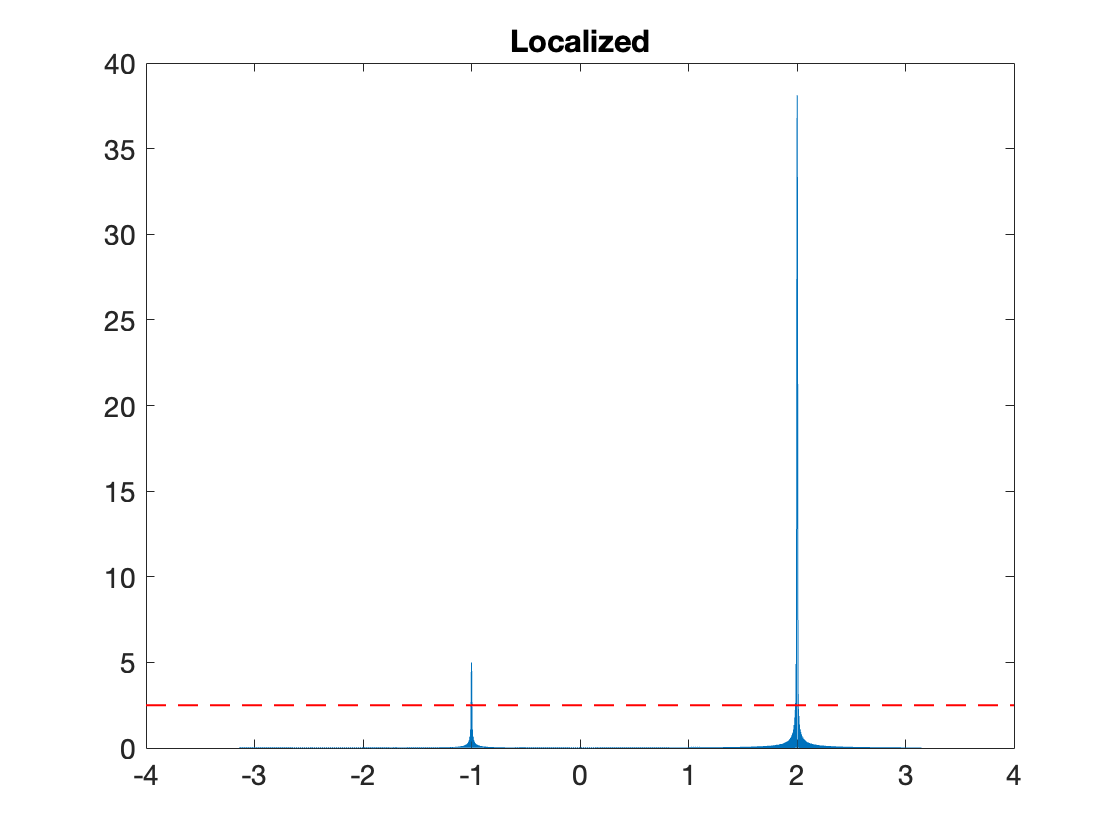}
}
\end{center}
\end{minipage}
\begin{minipage}{0.32\textwidth}
\begin{center}
\subfloat[]{
\includegraphics[width=0.95\textwidth]{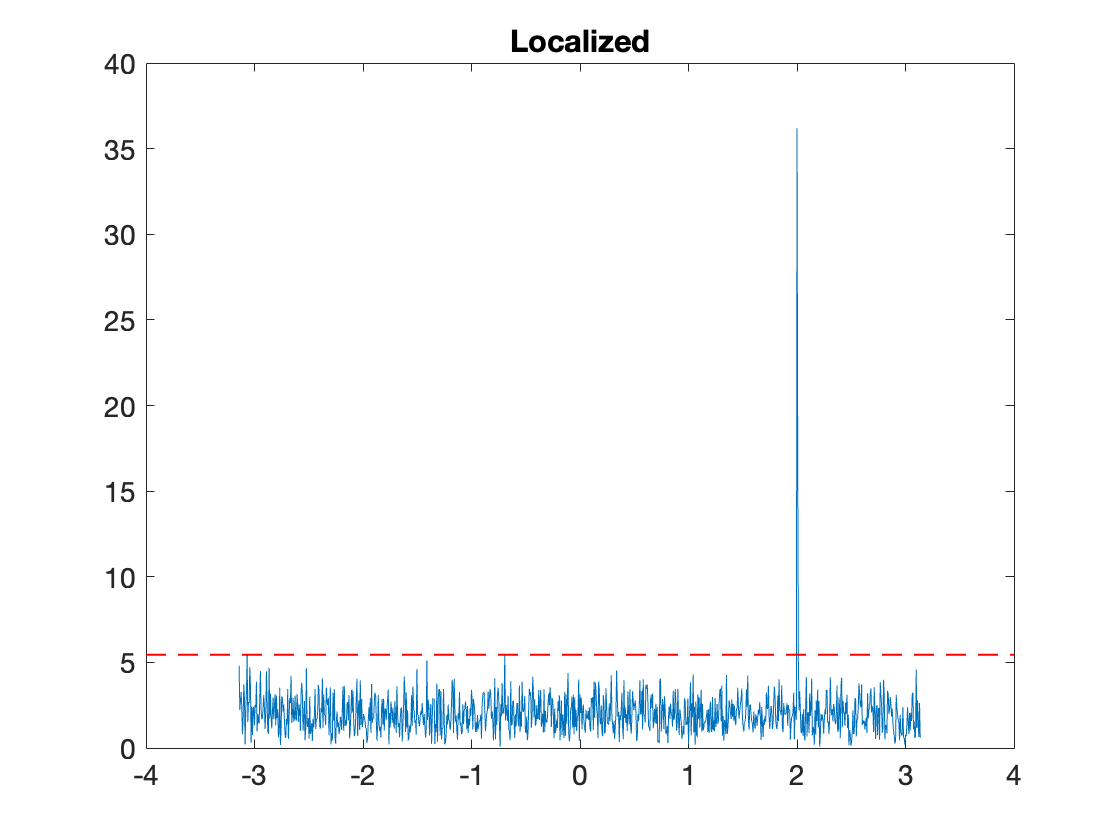}
}
\end{center}
\end{minipage}
\begin{minipage}{0.32\textwidth}
\begin{center}
\subfloat[]{
\includegraphics[width=0.95\textwidth]{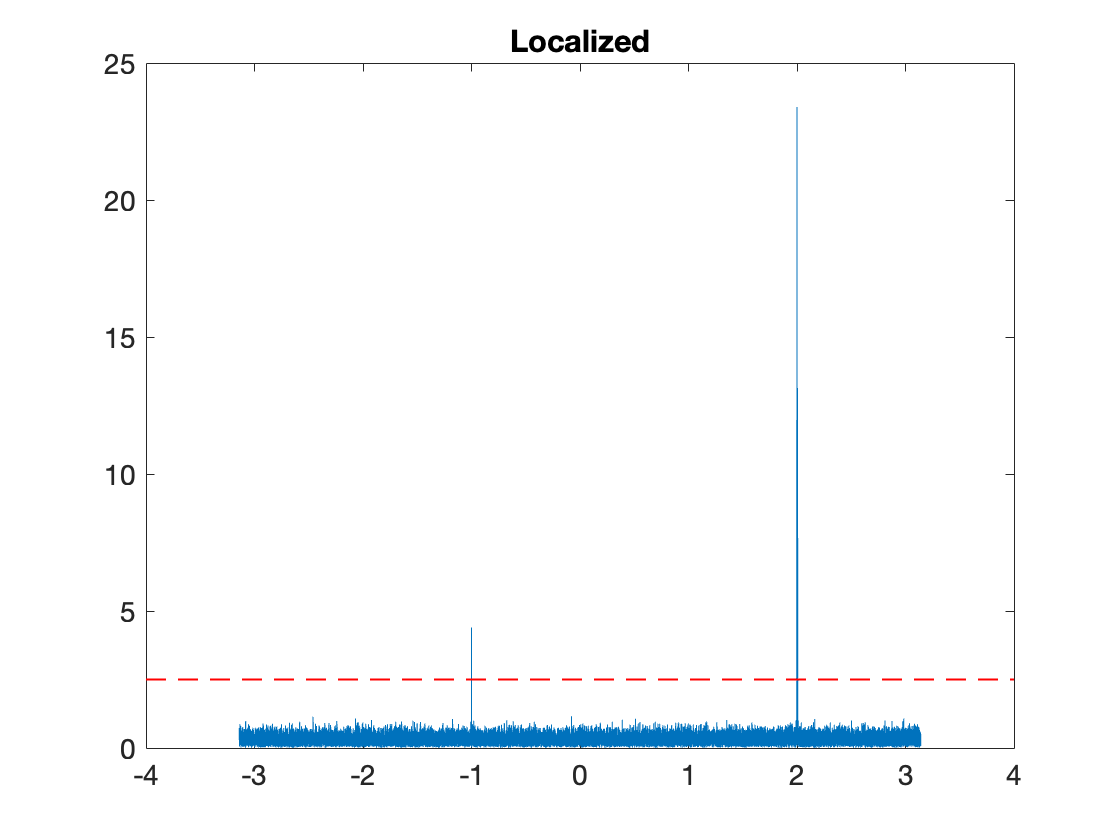}
}
\end{center}
\end{minipage}
\caption{(a) $|\sigma_{n}(x)|$ for $n=512$ without noise with threshold at the red line. All frequencies are detected with $\tau$ as shown in the red line. If $\tau=10$, then the smallest amplitude signal will not be detected. (b) $|\sigma_{n}(x)|$ for $n=1024$ with SNR -5 dB and threshold at the red line. (c) $|\sigma_{n}(x)|$ for $n=16384$ with SNR -5 dB and threshold at the red line.}
\label{fig:fail_case}
\end{center}
\end{figure}

Another challenge is when there are multiple $\lambda$'s that are too close (see Figure \ref{fig:fail_close_case}, $\lambda_2$ and $\lambda_3$ in this example). Our algorithm  needs more samples $n$ as indicated in \eqref{eq:pf2eqn3} and a proper minimum separation $\eta$ as discussed Remark \ref{rem:eta_and_m} to separate between close $\lambda$'s.

\begin{figure}[H]
\begin{center}
\begin{minipage}{0.32\textwidth}
\begin{center}
\subfloat[]{
\includegraphics[width=0.95\textwidth]{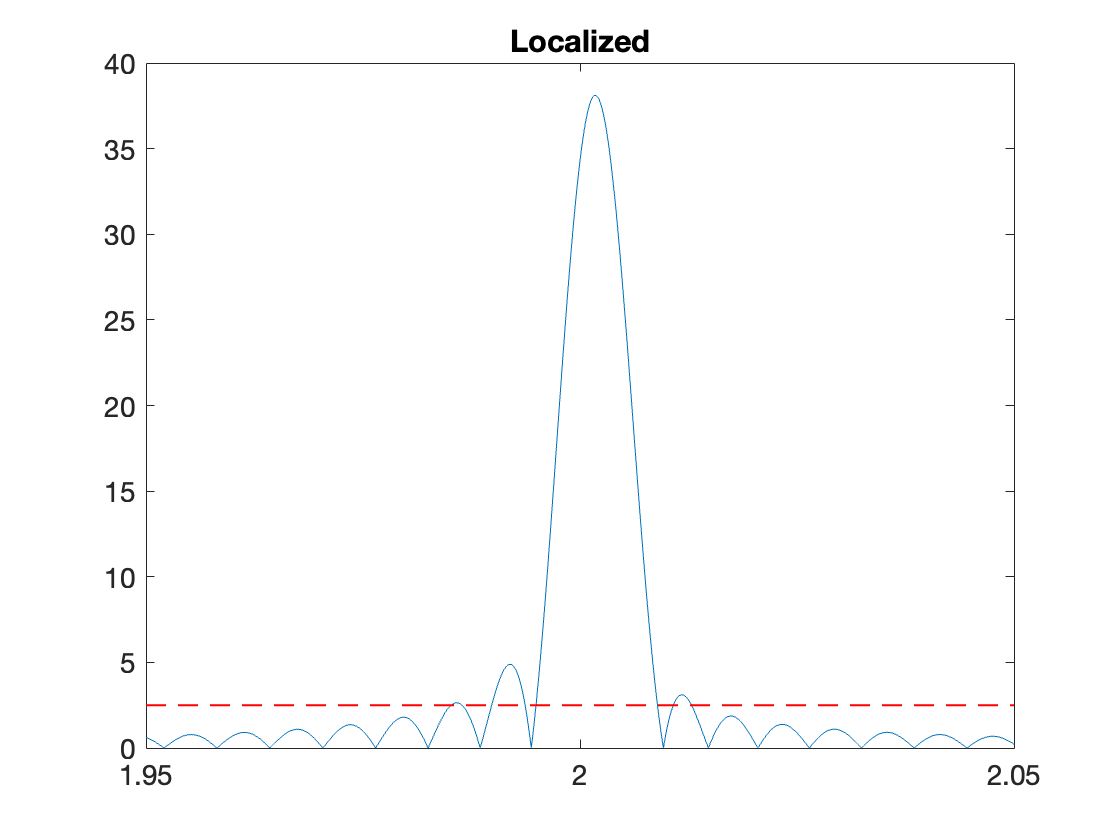}
}
\end{center}
\end{minipage}
\begin{minipage}{0.32\textwidth}
\begin{center}
\subfloat[]{
\includegraphics[width=0.95\textwidth]{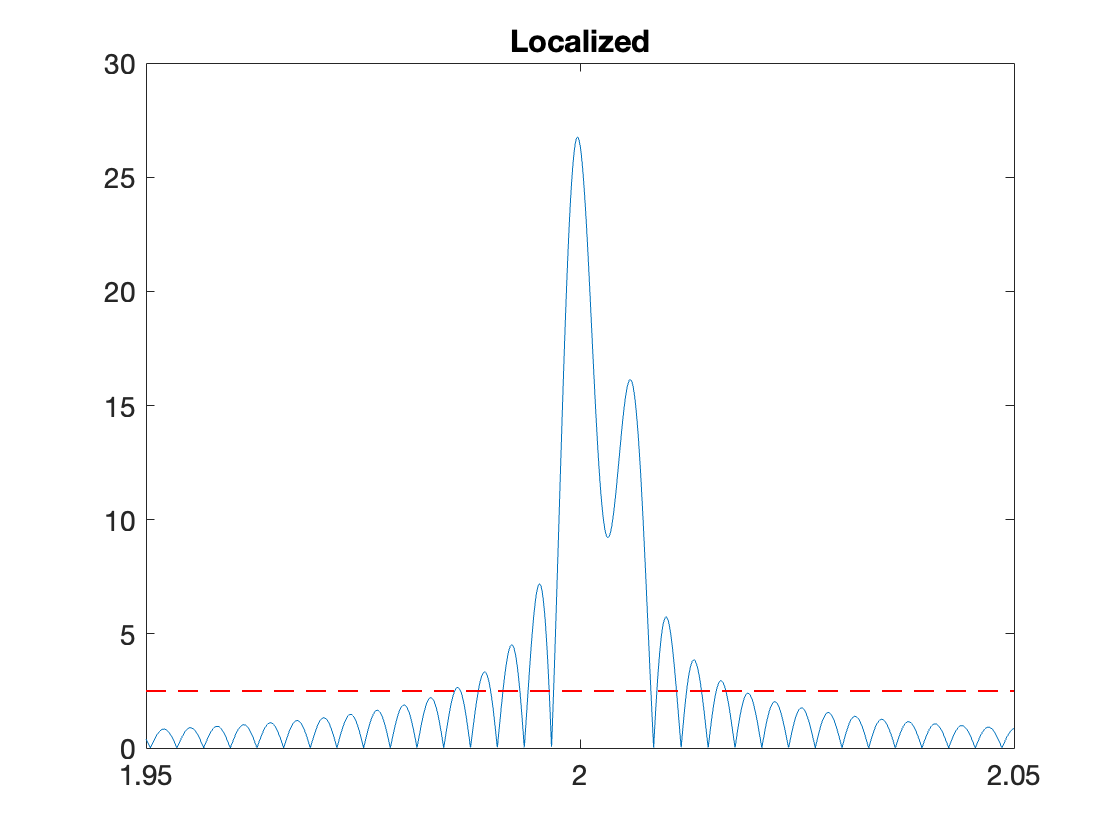}
}
\end{center}
\end{minipage}
\begin{minipage}{0.32\textwidth}
\begin{center}
\subfloat[]{
\includegraphics[width=0.95\textwidth]{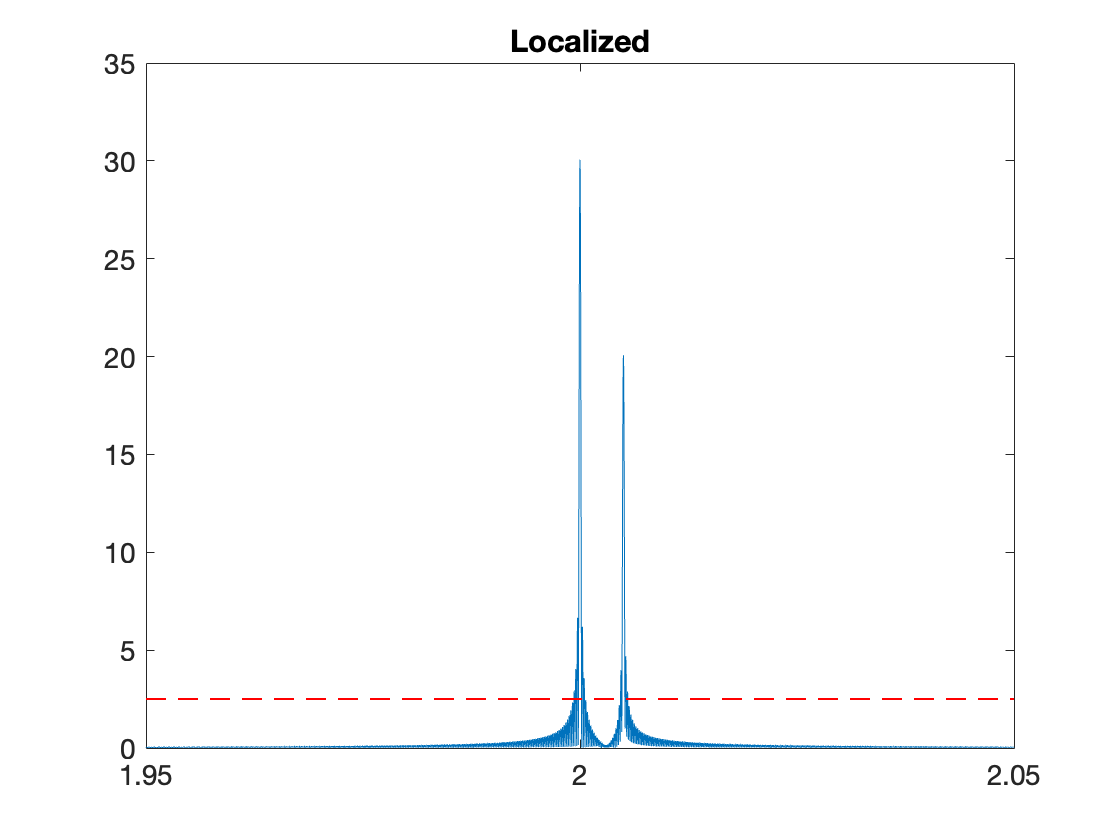}
}
\end{center}
\end{minipage}
\caption{(a) $|\sigma_{n}(x)|$ at $x$ near $2$ for $n=512$ without noise with threshold at the red line.  (b) $|\sigma_{n}(x)|$ at $x$ near $2$ for $n=1024$ without noise with threshold at the red line. (c) $|\sigma_{n}(x)|$ at $x$ near $2$ for $n=16384$ without noise with threshold at the red line.}
\label{fig:fail_close_case}
\end{center}
\end{figure}

\bhag{Proof of Theorem~\ref{theo:main}}\label{bhag:proof}
The following inequality, known as the \emph{Bernstein concentration inequality}, plays an important role in our proofs.
We recall that a trigonometric polynomial of order $<n$ is a function of the form $x\mapsto \sum_{|\ell|<n} b_\ell\exp(i\ell x)$.
\begin{prop}\label{prop:bernineq}
Let $n\ge 1$ be an integer, $T$ be any trigonometric polynomial of order $<n$. 
Then the derivative $T'$ of $T$ satisfies
\begin{equation}\label{eq:bernineq}
\max_{x\in\TT}|T'(x)|\le n\max_{x\in\TT}|T(x)|.
\end{equation}
In particular, if $N\ge 4\pi n$,  then
\begin{equation}\label{eq:meshnorm}
\frac{1}{2}\max_{x\in\TT}|T(x)|\le \max_{0\le k\le N}\left|T\left(\frac{2\pi k}{N}\right)\right|\le \max_{x\in\TT}|T(x)|.
\end{equation}
\end{prop}
For the proof of Theorem~\ref{theo:main}, we first estimate $E_n(x)$.
\begin{lemma}\label{lemma:noiselemma}
Let $\delta\in (0,1)$. 
There exist positive constants $C_1, C_2, C_3$, depending only on $H$ such that for $n\ge C_1(\ge 1)$, we have (cf. \eqref{eq:noisespectrum})
\begin{equation}\label{eq:noiseest}
\mathsf{Prob}\left(\max_{x\in\TT}|E_n(x)| \ge C_2V\sqrt{\frac{\log (C_3n/\delta)}{n}}\right)\le \delta.
\end{equation}
\end{lemma}

\begin{proof}
Let $x\in\TT$. 
We will use \eqref{eq:subgaussian_sum_tail} with  $\mathbf{a}=(a_\ell)_{\ell=-n+1}^{n-1}$, where $a_\ell=\hbar_n H(|\ell|/n)\exp(i\ell x)$, $|\ell|<n$.
From Theorem \ref{theo:fundatheorem}, it is not difficult to show that for $n\ge C_1$,
\begin{equation}\label{eq:pf1eqn0}
n\int_{-1}^1 H(t)^2dt\sim \sum_{|\ell|<n}H\left(\frac{|\ell|}{n}\right)^2, \qquad \hbar_n^{-1}\sim  n;
\end{equation}
i.e.,
\begin{equation}\label{eq:pf1eqn1}
|\mathbf{a}|_n^2\sim 1/n.
\end{equation}
Hence, \eqref{eq:subgaussian_sum_tail} shows that
\begin{equation}\label{eq:pf1eqn2}
\mathsf{Prob}\left(|E_n(x)|>t\right)\le  4\exp\left(-c\frac{nt^2}{V^2}\right).
\end{equation}
Applying this inequality for each $x=k/2n$, $k=0,\cdots,\lceil 4\pi n\rceil-1$, we see that
\begin{equation}\label{eq:pf1eqn3}
\mathsf{Prob}\left(\max_{0\le k\le 4n-1}|E_n(2\pi/(4\pi n))|>t\right)\le  c_1n\exp\left(-c\frac{nt^2}{V^2}\right).
\end{equation}
We observe that $E_n$ is a trigonometric polynomial of order $<n$. 
Hence, \eqref{eq:meshnorm} shows that
\begin{equation}\label{eq:pf1eqn4}
\mathsf{Prob}\left(\max_{x\in\TT}|E_n(x)|>2t\right)\le c_1n\exp\left(-c\frac{nt^2}{V^2}\right).
\end{equation}
We set the right hand side of the estimate \eqref{eq:pf1eqn4} equal to $\delta$, and solve for $t$ to obtain
\be
2t=C_2V\sqrt{\frac{\log(C_3n/\delta)}{n}}.
\ee
This leads to \eqref{eq:noiseest}.
\end{proof}

We are now in a position to prove Theorem~\ref{theo:main}.

\vspace{2ex}

\noindent\textsc{Proof of Theorem~\ref{theo:main}}

\vspace{2ex}

In this proof, we will denote 
\be
\varepsilon_n=\max_{x\in\TT}|E_n(x)|.
\ee
We choose $n$ so that Lemma~\ref{lemma:noiselemma} is applicable and yields with probability exceeding $1-\delta$:
\begin{equation}\label{eq:pf2eqn1}
 \varepsilon_n \le C_2V\sqrt{\frac{\log (C_3n/\delta)}{n}}\le \frac{\mathfrak{m}}{16}.
\end{equation}
All the statements in the rest of the proof assume a realization of the $\epsilon_j$'s so that \eqref{eq:pf2eqn1} holds; i.e., they all hold with probability exceeding $1-\delta$.

We observe next that if $J\subseteq \{1,\cdots, K\}$, $d\ge c_1n$, $x\in \TT$, and $|x-\lambda_\ell|\ge d$ for all $\ell\not\in J$, then
\begin{equation}\label{eq:pf2eqn2}
\left|\sigma_n(x)-\sum_{\ell\in J}A_\ell\Phi_n(x-\lambda_\ell)\right|\le \frac{ML}{(nd)^S}+\varepsilon_n \le \frac{ML}{(nd)^S}+\frac{\mathfrak{m}}{16}.
\end{equation}
We now choose $C$ as in \eqref{eq:thresholdCdef} and assume that
\begin{equation}\label{eq:pf2eqn3}
n\ge \max(4C/\eta, C_1), \mbox{ and } C_2V\sqrt{\frac{\log (C_3n/\delta)}{n}}\le \frac{\mathfrak{m}}{16}.
\end{equation}
Then \eqref{eq:pf2eqn2} implies that for every $x\in\TT$ for which $|x-\lambda_\ell|\ge C/n$ for all $\ell\not\in J$, we have
\begin{equation}\label{eq:pf2eqn6}
\left|\sigma_n(x)-\sum_{\ell\in J}A_\ell\Phi_n(x-\lambda_\ell)\right|\le \frac{ML}{(nd)^S}+\varepsilon_n \le \frac{\mathfrak{m}}{8}.
\end{equation}
Hence, if $|x-\lambda_\ell|\ge C/n$ for all $\lambda_\ell$, $\ell=1,\cdots, K$, \eqref{eq:pf2eqn2} applied with $J=\emptyset$ implies that
\begin{equation}\label{eq:pf2eqn4}
|\sigma_n(x)|\le \frac{\mathfrak{m} }{8}.
\end{equation}
Consequently, if $x\in\mathbb{G}$, then there is some $\lambda_\ell$ such that $|x-\lambda_\ell|<C/n\le \eta/4$.
Necessarily, there is only one $\lambda_\ell$ with this property.
We now define for $\ell=1,\cdots, K$,
\begin{equation}\label{eq:pf2eqn5}
\mathbb{G}_\ell=\{x\in\mathbb{G}: |x-\lambda_\ell|<C/n\}.
\end{equation}
Obviously, $\mathbb{G}=\bigcup_{\ell=1}^K \mathbb{G}_\ell$, and $\mathbb{G}_\ell$'s are all mutually disjoint.
This proves the disjoint union condition.
The diameter condition as well as the separation condition are obviously satisfied.

We prove next the interval inclusion property, which implies in particular, that none of the sets 
$\mathbb{G}_\ell$ is empty. 
In order to prove this, we observe that $1 =\max_{x\in\TT}|\Phi_n(x)|$.
Hence, for $x\in I_\ell$, the estimate \eqref{eq:pf2eqn6} applied with $J=\{\ell\}$ implies that
\begin{equation}\label{eq:pf2eqn7}
|\sigma_n(x)-A_\ell\Phi_n(x-\lambda_\ell)|\le \frac{\mathfrak{m} }{8}.
\end{equation}
Since $\Phi_n$ is a trigonometric polynomial of order $<n$, the Bernstein inequality (together with \eqref{eq:phimax}) implies that for $x\in I_\ell$ (i.e., $|x-\lambda_\ell|<1(4n)$),
\be
|\Phi_n(x-\lambda_\ell)-1 | \le n|x-\lambda_\ell|  \le (1/4) .
\ee
So, \eqref{eq:pf2eqn7} leads to
\be
|\sigma_n(x)|\ge (3/4)|A_\ell| -\frac{\mathfrak{m} }{8}\ge (5/8)\mathfrak{m} , \qquad x\in I_\ell.
\ee
This proves that $I_\ell \subseteq \mathbb{G}_\ell$ for all $\ell=1,\cdots, K$.
The estimate \eqref{eq:lambdaerr} is clear from the diameter condition and the interval inclusion property.
\qed


\chapter{Robust and Tractable Multidimensional Exponential Analysis}\label{1paper}

\noindent
The content presented in this chapter is adapted from our paper titled \textit{``Robust and Tractable Multidimensional Exponential Analysis''}~\cite{mhaskar2024robust}.

\bhag{Introduction}\label{section:intro}
In continuation of the point source separation problem discussed in Section \ref{section:theosect}, this chapter addresses its multidimensional extension, which appears in various applications such as tomographic imaging (including Computerized Tomography (CT), magnetic resonance imaging (MRI), radar and sonar imaging), wireless communication, antenna array processing, sensor networks, and automotive radar, among others. 
Mathematically, the problem can be formulated as follows.
Given a multidimensional signal of the form
\begin{equation}\label{eq:multi_exp_problem}
f(\mathbf{x})=\sum_{k=1}^K a_k\exp(-i\langle \mathbf{x}, \w_k\rangle), \quad \mathbf{x}, \w_1,\cdots,\w_k\in\mathbb{R}^q,
\end{equation}
find the number $K$ of components, and the parameters $a_k$ and $\w_k$'s \yadi{$\w_k$}{Multivariate frequencies/points} \yadi{$q$}{Dimension of the observations}. 
Of course, this is a problem of inverse Fourier transform if we could observe the function $f$ at \textbf{all} values of $\x$. In practice, however, one can observe (after some sampling and renaming of the variables) the values of $f$ at only \textbf{finitely many} multi-integer values of $\x$.
In this case, it is not possible to distinguish values of $\w_k$ which are equal modulo $2\pi$ in all variables. So, this is a special case of the ancient trigonometric moment problem \cite{shohat1950problem}, except that we do not have \textbf{all} the trigonometric moments (i.e., the samples $f(\boldsymbol\ell)$) for all values of $\boldsymbol\ell\in \ZZ^q$. Thus, the problem is the ill-posed problem known often as the super-resolution problem: knowing the information in a finite domain of the frequency space, we need to extend it to the entire frequency space. The important problem in this connection is to determine the relationship between the number of samples  $f(\boldsymbol\ell)$  needed to recuperate the desired quantities up to a given accuracy.

In the univariate case, there are many methods to solve the problem of parameter estimation in exponential sums, we refer to \cite{plonka2018numerical, diederichs2018sparse} for a good introduction.
 
Under the rubric of target estimation or localization, which is one of the fundamental problems in radar signal processing with many civilian and military applications including landmine detection and geolocations (cf. \cite{zhu2016super}), a broad set of techniques has emerged for solving exponential analysis problems. Most of these methods are categorized broadly as subspace methods \cite{krim}, and are based on statistical considerations, rather than the nature of the signal itself. In fact, quite a few papers, e.g., \cite{venkatasubramanian2022toward, raghavan2020generalized}, are interested in testing a statistical hypothesis on whether or not there exists a signal at all. There is a theoretical limit on how much SNR can be tolerated, depending upon the number of antenna elements and number of observations \cite{venkatasubramanian2022toward}, in spite of a huge computational cost. On the other hand, beamforming methods \cite{krim} take into account the nature of the signal but focus again on noise, and try to maximize the SNR. 

Contrary to the remarks in \cite{krim}, methods based on filtered inverse Fourier transform were developed  in \cite{loctrigwave, singdet, bspaper}, and shown to work very well and faster than subspace methods.

Over the past several decades, several approaches to multivariate exponential analysis have been investigated. 
Many of these are extensions of the classical Prony method, and variations of MUSIC and ESPRIT, which are designed to stabilize this classical method.
For example, \cite{quinquis2004some} discusses multivariate extensions of MUSIC and ESPRIT algorithms.
The number of samples required to recuperate the $\w_k$'s up to an accuracy of $\mathcal{O}(1/n)$ is typically on the order of $\O(n^q)$ \cite{diederichs2018sparse, potts2013parameter, sahnoun2017multidimensional, kunis2016multivariate, peter2015prony}.
Under some extra assumptions, it is shown in \cite{diederichs2023many} that this can be improved to $\O(n\log^{q-1}n)$, and to $(q+1)n^2\log^{2q-2}n$ \cite{sauer2018prony}.
In \cite{bspaper}, we have discussed an an analogue of the beam-forming methods in multivariate setting.
However, this requires $\O(n^q)$ samples.

In \cite{cuyt2018multivariate, cuyt2020sparse}, the authors proposed a method to solve the problem using a combination of Prony's method and a based method, so that number of samples required is $\O(qn)$. 
Methods based on Pad\'e approximation and orthogonal polynomials on the unit circle are explored, especially in the univariate case \cite{singdet,derevianko2022exact}. 
The methods described in the paper \cite{cuyt2018multivariate} are also connected with Pad\'e approximation, and the paper \cite{briani2017vexpa} develops this method further to obtain accurate solution to the multivariate problem in the presence of a moderate Gaussian noise with a sub-Nyquist sampling rate.

In this chapter, we propose a novel method based on localized trigonometric polynomial kernels developed in \cite{loctrigwave}.
Our method utilizes $\O(qn)$ samples as well, but is faster and far more robust under noise. 
In contrast to the subspace based methods, our method takes into account the nature of the signal, resulting in a significant noise reduction with only a small number of observations per signal, and yields accurate results with theoretical guarantees.

The rest of this chapter is organized as follows. Section~\ref{section:sysmodel} introduces a system model for tomographic imaging which illustrates how the problem of multidimensional exponential analysis arises in signal processing. The algorithms to implement these theorems are given in Section~\ref{section:algsect}, and demonstrated in the case of the three examples explored in \cite{cuyt2018multivariate}.

\bhag{System model}\label{section:sysmodel}
As mentioned in Section \ref{section:intro}, the multi-dimensional exponential model in \eqref{eq:multi_exp_problem} arises in many applications in science and engineering. 
In this section, we illustrate the details of one such application, namely,  tomographic imaging.

In tomographic imaging,  an object of interest being imaged is probed by a sequence of monochromatic tones swept through a  frequency range $[\Omega_{init}, \Omega_{fin}] \in \RR$ (in units of Hertz).
The sensor transmits a signal onto a scene with respect to various angles $\{\bs\theta_m=(\theta_m, \phi_m)\}_m$, where $\theta_m \in [\Theta_{init}, \Theta_{fin}] \subseteq [0, 2\pi]$ and $\phi_m \in [\Phi_{init}, \Phi_{fin}] \subseteq [0, \pi]$, are the azimuth and elevation angles, with respect to the sensor (such as a radar \cite{jakowatz1996spotlight}), respectively. 
The scene reflectivity is a complex-valued function over the spatial coordinates, $\w  \in  \RR^3$ in 3D imaging (units of meters).
In this chapter, this is modeled as a distribution, $\mu=\sum_{k=1}^K a_k\delta_{\w_k}$, where $\delta$ denotes the Dirac delta.
For the sake of simplicity, we will use $\delta$ regardless of whether it is applied to vectors in different dimensions or scalars.
With rescaling and shifting, we may assume that the domain of $\mu$ is a subset of $[-\pi,\pi]^3$.

We define the center reference point (CRP) to be the center of mass of the scene to be reconstructed, and the line of sight (\textsf{los}) as the unit vector, $\mathbf{i}_{\mbox{los}}$, that points from the transmitter to the CRP of the scene. 
The distance along the \textsf{los} from the transmitter to the voxel location (placing the origin at the transmitter),  $\mathbf{r}_0$, of interest is called the `downrange' $r_0 =|\mathbf{r}_0\cdot \mathbf{i}_{\mbox{los}}|$ (where $\mathbf{r}_0$ is the position vector of a scatterer in the scene, and $\cdot$ denotes the inner product operation). Given this, it can be shown that the backscattered signal at downrange $r_0$ from the sensor, when viewed at angle $\bs\theta$, is given by \textcolor{black}{\cite{jakowatz1996spotlight, raj2016hierarchical, idrisstaes21}}
\begin{equation}
\Upsilon(t)|_{\bs{\theta},\mathbf{r}_0} = [R_{\bs{\theta}}\left\{\mu\right\}(r_0)] \chi\left(t-\frac{2r_{0}}{\nu_{p}}\right)+\check{\epsilon}(t),
\end{equation}
where $2r_{0}/\nu_{p}$ denotes the two-way time delay, $\nu_{p}$ represents the speed of wave propagation, $\chi$ is the transmitted waveform, $\check{\epsilon}(t)$ is the measurement noise, and $R_{\bs\theta}$ is the Radon transform of the scene, $\mu$, with respect to angle $\bs\theta$, and evaluated at the \textcolor{black}{downrange location $\mathbf{r}_0$}.
This corresponds to the integral across sensor returns from all points along a hyperplane perpendicular to the downrange location $\mathbf{r}_0$,  commonly referred to as an `iso-range contour'.

Therefore the complete response at time $t$ from all ranges, along the line formed by intersecting the scene at $\bs\theta$, is rewritten as
\begin{equation}
   \Upsilon_{\bs\theta}(t)=\int [R_{\mathbf{\bs\theta}}\left\{\mu\right\}(r)]\chi\left(t-\frac{2r}{\nu_{p}}\right)\,dr+\check{\epsilon}(t).
\end{equation}
This can be reformulated in convolution form (after a choice of units, without loss of generality, so that $\nu_p=2$) as:
\begin{equation}\label{eq:sys1}
    \Upsilon_{\bs\theta}(t)=\left(R_{\bs\theta}(\mu)\ast \chi\right)(t)+\check{\epsilon}(t),
\end{equation}
where  $\ast$ denotes the convolution operation. 
Equation \eqref{eq:sys1} can be interpreted as the response to a linear time invariant (LTI) system \cite{oppenheim1996signals} with an input signal $\chi(t)$ and an distributional impulse response $\mu_{\bs\theta}(t)$ which characterizes the interaction between the transmit waveform and the scene with respect to sensing angle $\bs\theta=(\theta,\phi)$ (where $\theta$ and $\phi$ are the azimuth and elevation angles respectively).
The received signal $\Upsilon_{\bs\theta}(t)$, is the output of this LTI system and is subsequently sampled at the receiver.
Given $\Upsilon_{\bs\theta}(t)$ for a finite grid $\bs\theta \in \{\bs\theta_m\}$ as described earlier, our problem is to estimate the underlying scene reflectivity  $\mu$ i.e. to form the image.

Taking the Fourier transform, we obtain
\begin{equation}\label{eq:sys3}
\frac{\widehat{\Upsilon_{\bs\theta}}(u)}{\widehat{\chi}(u)}=\widehat{R_{\bs\theta}\{\mu\}}(u)+\epsilon(u),
\end{equation}
where $\epsilon$ is derived from the Fourier transform of $\check{\epsilon}$ in an obvious way.
In this chapter, we will treat $\epsilon$ itself as the noise in the obervations.
The division on the left hand side of \eqref{eq:sys3} might amplify the noise in the original observations, but we will relate the noise level and the number of observations etc. in Theorem~\ref{theo:main}.

When $\bs\theta=(\theta,\phi)$, the Fourier projection slice theorem implies that the one dimensional Fourier transform  $\widehat{R_{\bs\theta}(\mu)}(u)$ is given by $\mathfrak{F}({\mu})(u\cos(\theta)\cos(\phi), u\sin(\theta)\cos(\phi), u\sin(\phi))$, where $\mathfrak{F}$ denotes three dimensional Fourier transform.
Writing \be\x=(u\cos(\theta)\cos(\phi), u\sin(\theta)\cos(\phi), u\sin(\phi)),\ee equation \eqref{eq:sys3} becomes
\begin{equation}\label{eq:sys4}
\frac{\widehat{\Upsilon_{\bs\theta}}(u)}{\widehat{\chi}(u)}=\mathfrak{F}({\mu})(\x)+\epsilon(\x).
\end{equation}
Although the vector $\x$ is in spherical coordinates here, one interpolates the actual data to obtain (a noisy version of) $\mathfrak{F}({\mu})(\x)$ at $\x$ in a Cartesian grid, so as to facilitate the use of fast Fourier transform \textcolor{black}{\cite{jakowatz1996spotlight}}. In particular, the problem of finding the reflexivity function $\mu$ reduces to the problem mentioned in Section~\ref{section:intro}.

 In this chapter we present a novel approach for solving the multidimensional analysis problem by resampling the Fourier space $\mathfrak{F}$ in a computationally efficient manner, while showing substantial improvements over state-of-the-art techniques including the MUSIC and ESPRIT approaches.

It is important to note that this chapter offers a general efficient tool for solving multidimensional exponential problems with applications beyond tomographic imaging such as signal source separation and direction-of-arrival estimation in multi-channel radar systems.

\bhag{Multivariate illustrations}\label{section:algsect}

We have described our univariate Theorem~\ref{theo:main} in Section~\ref{bhag:theorem} and implemented Algorithm~\ref{alg:univariate} for the univariate case in Section~\ref{bhag:1d}.
In this section, we will extend the algorithms to higher dimensional cases, following the ideas in \cite{cuyt2018multivariate, cuyt2020sparse}.
As in this chapter, we will use Algorithm~\ref{alg:univariate} on a set of univariate problems to components of $\w_k$ along different lines, defined by a basis $(\Delta_k)$ of $\RR^q$.
Ideally, the choice of this basis should be such that the minimal separation among $\langle \Delta_j, \w_k\rangle$ should be maximal for each $j$.
At this point, we don't know how to ensure this without knowing the ground truth.
We will use the same bases as in \cite{cuyt2018multivariate}.

There are two main issues to solve here, registration and assessment.

\vspace{2ex}

\noindent\textbf{Registration}:

\vspace{2ex}

The problem of registration is to figure out which solution based on the data on one line corresponds to which solution based on the data on another line.
Our solution is to use data of the form $f(\Delta_1+\ell\Delta_2)$ (cf. \eqref{eq:multi_exp_problem}) to obtain an accurate estimate on $\langle \w_k, \Delta_1\rangle$. 
The corresponding amplitude is then $A_k\exp(-i\langle \w_k,\Delta_1\rangle)$, which yields an approximation to $\langle \w_k,\Delta_1\rangle$. 
By reversing the roles of $\Delta_1, \Delta_2$, we obtain an accurate estimate on $\langle \w_k,\Delta_1\rangle$ and an approximation to $\langle \w_k,\Delta_2\rangle$.
The nearest neighbor search gives an accurate value for both the components. 
This procedure is described in Algorithm~\ref{alg:estimation}.

We note that \cite{cuyt2018multivariate} has a different approach for this problem.
In that paper, one finds first the amplitudes by solving a system of linear equations, and the registration is done by keeping track of the solutions.
In the case of noisy signal, they use more samples of the form $\ell\langle \w_k, \Delta_j\rangle +\ell' \langle \w_k, \Delta_{j'}\rangle$ for different values of $j$, $j'$.

\vspace{2ex}

\noindent\textbf{Performance assessment:}

\vspace{2ex}

After finding the estimate $\hat{\w}_k$ for $\w_k$ for each $k$, the next challenge is to determine how many of these estimates represent the actual points.
For this reason, we fix a radius $r$ and declare $\hat{\w}_k$ to be an accurate estimator of $\w_k$ if $|\w_k-\hat{\w}_k|<r$. 
We then count how many points were estimated accurately within this error margin.

Throughout the experiments in Table \ref{tab:compare_table} and Table \ref{tab:result_table}, we evaluated our approach using 16 trials. For each trial, the RMSE was calculated by the formula
\be\mbox{RMSE} = \sqrt{\frac{1}{K} \sum_{k=1}^K |\w_k-\hat{\w}_k|^2}.\ee
We then reported the number of points reconstructed, run-time, memory, accuracy, and RMSE by taking average over 16 trials. Finally, we computed the standard deviation of the RMSE.

In Section~\ref{section:algortihms}, we will describe our algorithm in the multivariate case.
In Section~\ref{section:twod12pts}, we illustrate the various steps in the case of a two dimensional dataset, and discuss the results. 
Section~\ref{section:algsect2} discusses the adaptation of our algorithm in the three and higher dimensional case. 
The corresponding results are discussed in Section~\ref{section:3dresults}.

\subsection{Multivariate algorithms}\label{section:algortihms}

We extend Algorithm~\ref{alg:univariate} to a two dimensional problem by using $\tilde{\mu}(\Delta_2+\ell \Delta_1)$, resulting in an accurate estimation of one component of the $\w_k$'s and an approximate estimation of the other component.

Algorithm~\ref{alg:estimation} uses Algorithm~\ref{alg:univariate} successively with pairs of components to obtain the final accurate estimation of all the components of $\w_k$'s.

\begin{algorithm}[ht]
\begin{algorithmic}[1]
\item[{\rm a)}] \textbf{Input:} Basis $\{\Delta_d\}_{d=1}^q$ of $\RR^q$, $\tau$, $\eta$, and signal $\tilde{\mu}(\Delta_{d_2}+\ell \Delta_{d_1})$, $1\le d_1<d_2\le q$ (a total of $(2n-1)q$ samples).
\item[{\rm b)}] \textbf{Output:} Estimation of $\hat{A}_k$ and $\hat{\w}_k$ for $k = 1, \ldots, K$.
\FOR{ $d_1=1$ to $q-1$}
\FOR{ $d_2=d_1+1$ to $q$}
\STATE Run Algorithm \ref{alg:univariate} with parameters $\tau$, $\eta$, and signal $\tilde{\mu}(\Delta_{d_2}+\ell \Delta_{d_1})$
\item[] \textbf{Note:} From the above step,  we will obtain $A_k$ and highly accurate result of $\langle \Delta_{d_1}, \hat{\w}_k\rangle$ together with corresponding less accurate result of $\langle \Delta_{d_2}, \hat{\w}_k\rangle$ in the output parameter for phase.
\STATE Run Algorithm \ref{alg:univariate} with parameters $\tau$, $\eta$, and signal $\tilde{\mu}(\Delta_{d_1}+\ell \Delta_{d_2})$
\item[] \textbf{Note:} From the above step,  we will obtain $A_k$ and highly accurate result of $\langle \Delta_{d_2}, \hat{\w}_k\rangle$ together with corresponding less accurate result of $\langle \Delta_{d_1}, \hat{\w}_k\rangle$ in the output parameter for phase.
\STATE Use nearest neighbor algorithm to obtain the highly accurate pair for both $\langle \Delta_{d_1}, \hat{\w}_k\rangle$ and $\langle \Delta_{d_2}, \hat{\w}_k\rangle$
\ENDFOR
\ENDFOR
\STATE We can write the result as $ \Delta \hat{\w}$, where \\
$\Delta =[\Delta_1,\ldots,\Delta_q]^T$ and $\hat{\w} = [\hat{\w}_1, \ldots, \hat{\w}_k ]$.
\STATE \textbf{Return: } We then obtain $\hat{\w}_k$ by computing \\
$\hat{\w} =  \Delta^{-1} \Delta \hat{\w}$.
\STATE \textbf{Return: } $\hat{A}_k$ and $\hat{\w}_k$ for $k = 1,\ldots, K$.
\caption{Parameter estimation in a multidimensional signal.}
\label{alg:estimation}
\end{algorithmic}
\end{algorithm}

\vspace{2ex}

\subsection{Illustration in a two dimensional case}\label{section:twod12pts}
We will illustrate our algorithm by using an example of 2-d image which we obtain from \cite{cuyt2020sparse}, as shown in Figure~\ref{fig:12ptexample}. Following this chapter, we take $\Delta_1 = (1.38, 4.14)$, $\Delta_2 = (-7.56, 5.67)$.
\yadi{$\Delta_j$}{Independent vectors in $\RR^q$}
\begin{figure}[ht]
\begin{minipage}{0.49\textwidth}
\begin{center}
\subfloat[]{
\begin{tabular}{ |c|c|c| } 
 \hline
 $k$ & $\omega_k$& $a_k$ \\ 
 \hline
 1 & $(-1.2566,0.6283)$ & 50 \\
 2 & $(-0.7540,0.3142)$ & 50 \\
 3 & $(-0.2513,1.2566)$ & 50 \\
 4 & $(-0.2513,0.6283)$ & 50 \\
 5 & $(-0.2513,0)$ & 50 \\
 6 & $(0,-0.6283)$ & 50 \\
 7 & $(0,-1.2566)$ & 50 \\
 8 & $(0.2513,1.2566)$ & 50 \\
 9 & $(0.2513,0.6283)$ & 50 \\
 10 & $(0.2513,0)$ & 50 \\
 11 & $(0.7540,0.3142)$ & 50 \\
 12 & $(1.2566,0.6283)$ & 50 \\
 \hline
\end{tabular}
\label{fig_table}}
\end{center}
\end{minipage}
\begin{minipage}{0.49\textwidth}
\begin{center}
\subfloat[]{
\includegraphics[width=.8\textwidth]{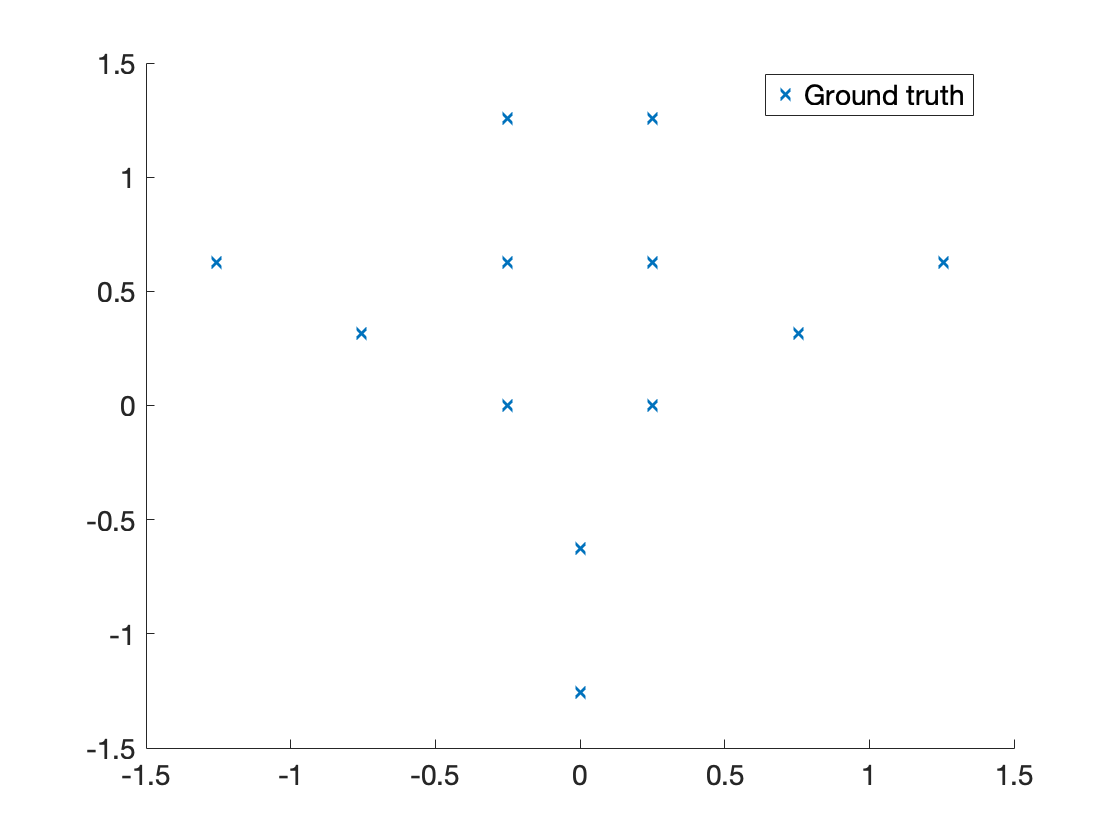}
\label{fig_pic}}
\end{center}
\end{minipage}
\caption{The two dimensional data comprising 12 points \cite{cuyt2020sparse}. 
(a) The actual points and amplitudes. (b) A graphic representation.}
\label{fig:12ptexample}
\end{figure}

Let $\w_1,\ldots,\w_{12} \in \mathbb{R}^2$ and $\{\Delta_1,\Delta_2\}$ be a basis for $\mathbb{R}^2$. Here, we have for $\ell\in\ZZ$, $|\ell|<n$,
\bea
\hat{\mu}(\Delta_2+\ell\Delta_1)&=\sum_{k=1}^{12} A_k \exp(-i\langle \Delta_2, \w_k\rangle) \exp(-i\ell\langle \Delta_1, \w_k\rangle) \\
    \hat{\mu}(\Delta_1+\ell\Delta_2)&=\sum_{k=1}^{12} A_k \exp(-i\langle \Delta_1, \w_k\rangle) \exp(-i\ell\langle \Delta_2, \w_k\rangle)
\eea
    The number of samples required is $4n-2$, where $n$ is the degree of the localized kernel. We then apply our low pass filter and obtain
\bea
\sigma_{n,1}(x) &=\hbar_n\sum_{k=1}^{12} A_k \exp(-i\langle \Delta_2, \w_k\rangle)\Phi_n(x-\langle \Delta_1, \w_k\rangle) \\
\sigma_{n,2}(x) &=\hbar_n\sum_{k=1}^{12} A_k \exp(-i\langle \Delta_1, \w_k\rangle)\Phi_n(x-\langle \Delta_2, \w_k\rangle)
\eea
From the Theorem \ref{theo:main}, $\langle \Delta_1, \w_k\rangle$ will be $x$ where the peaks occurs in $|\sigma_n(x)|$,  $\langle \Delta_2, \w_k\rangle \approx Phase \left( \sigma_n(x)\right)$,  and $A_k \approx |\sigma_n(x)|$. Now, we can obtain the accurate estimation of $\langle \Delta_1, \w_k\rangle$ corresponding to less accurate estimation of $\langle \Delta_2, \w_k\rangle$ (Figure~\ref{fig:12ptsalg_step1}).
\begin{figure*}[ht]
\begin{center}
\subfloat[]{\includegraphics[width=0.3\textwidth]{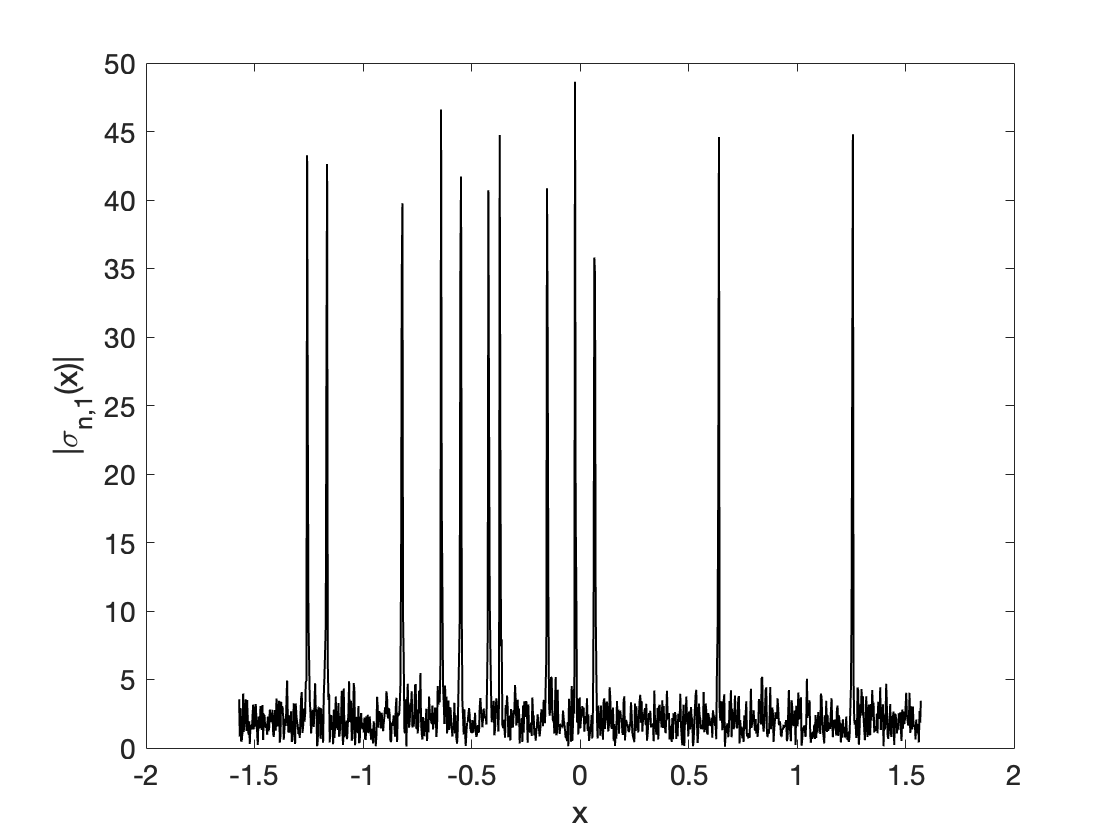}
\label{fig_first_case}}
\hfil
\subfloat[]{\includegraphics[width=0.3\textwidth]{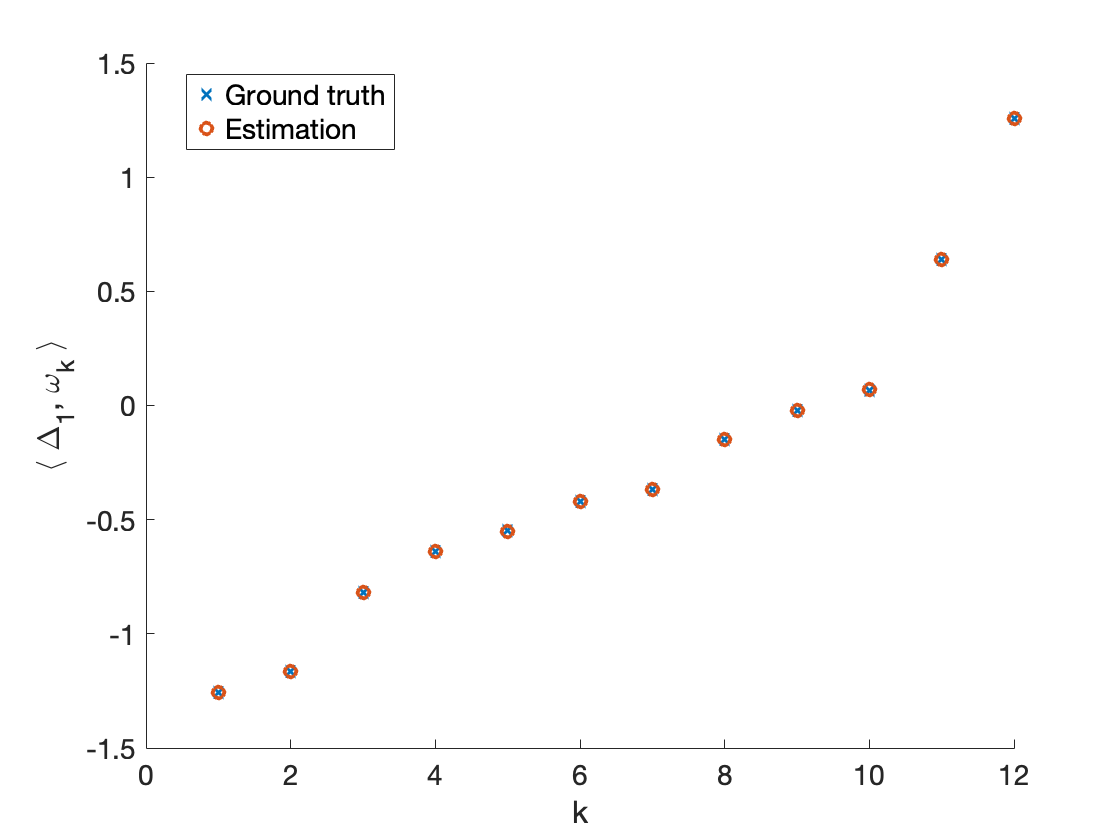}
\label{fig_second_case}}
\hfil
\subfloat[]{\includegraphics[width=0.3\textwidth]{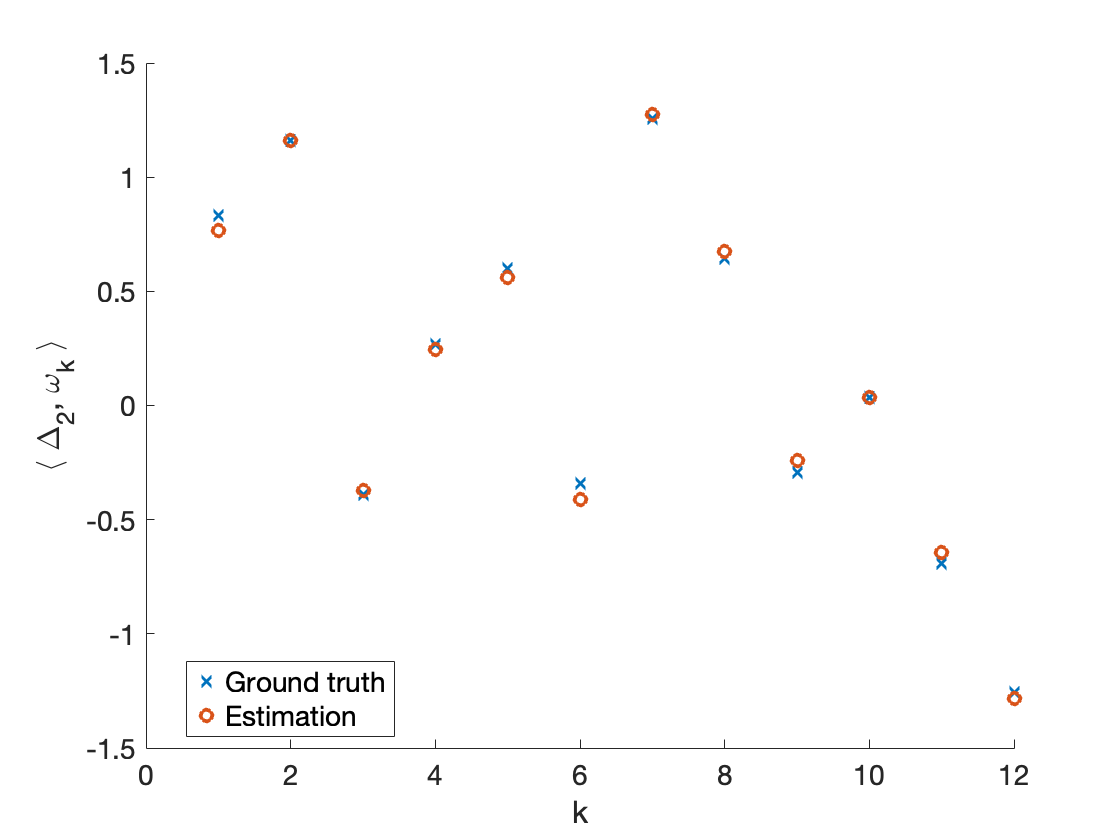}
\label{fig_third_case}}
\caption{(a) $|\sigma_{n,1}(x)|$. (b) Accurate determination of $\langle \Delta_1, \w_k\rangle$. (c) Approximate estimation of $\langle \Delta_2, \w_k\rangle$.}
\label{fig:12ptsalg_step1}
\end{center}
\end{figure*}
Then, we apply the same method in $\Delta_2$ direction to obtain the accurate estimation of $\langle \Delta_2, \w_k\rangle$ corresponding to less accurate estimation of $\langle \Delta_1, \w_k\rangle$ (Figure~\ref{fig:12ptsfinalstep} left and middle).

Finally, we can use nearest neighbor to obtain accurate estimation for both $\langle \Delta_1, \w_k\rangle$ corresponding to $\langle \Delta_2, \w_k\rangle$ and compute $A_k$. The final result is showed as the rightmost part of Figure~\ref{fig:12ptsfinalstep}.

\begin{figure*}[ht]
\subfloat[]{
\includegraphics[width=.3\textwidth]{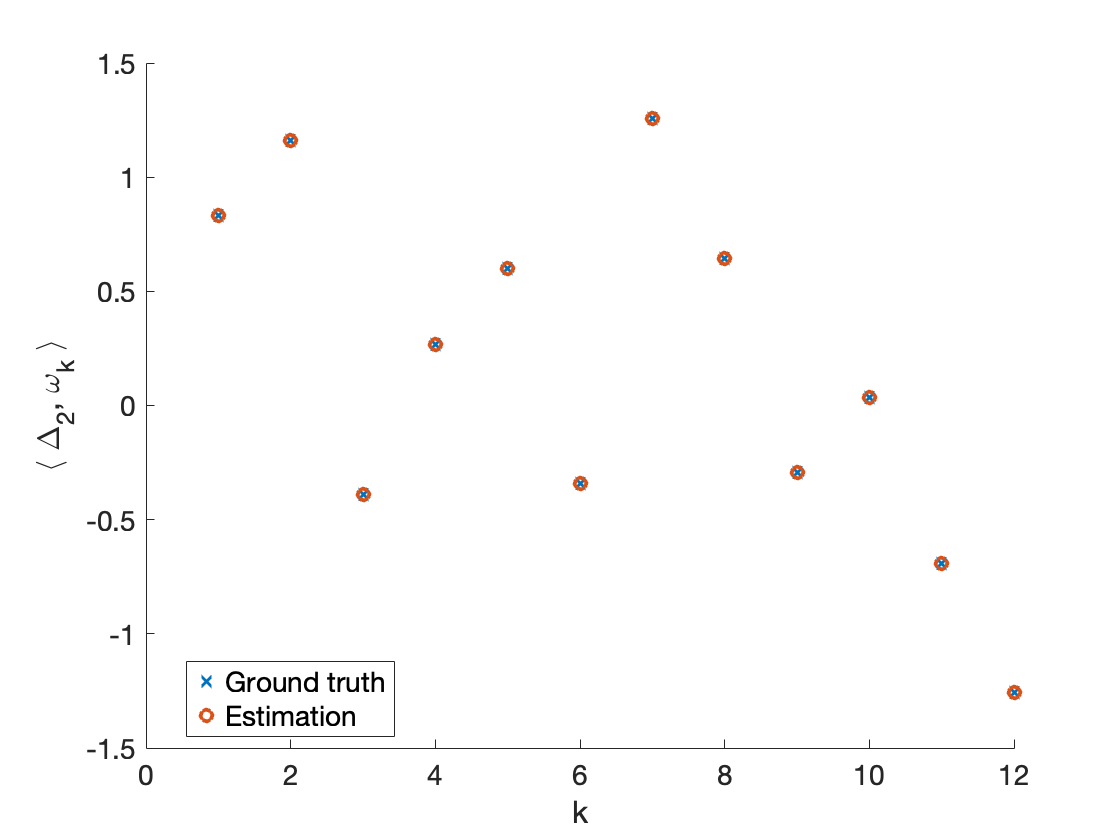}
\label{fig_first}}
\hfil
\subfloat[]{
\includegraphics[width=.3\textwidth]{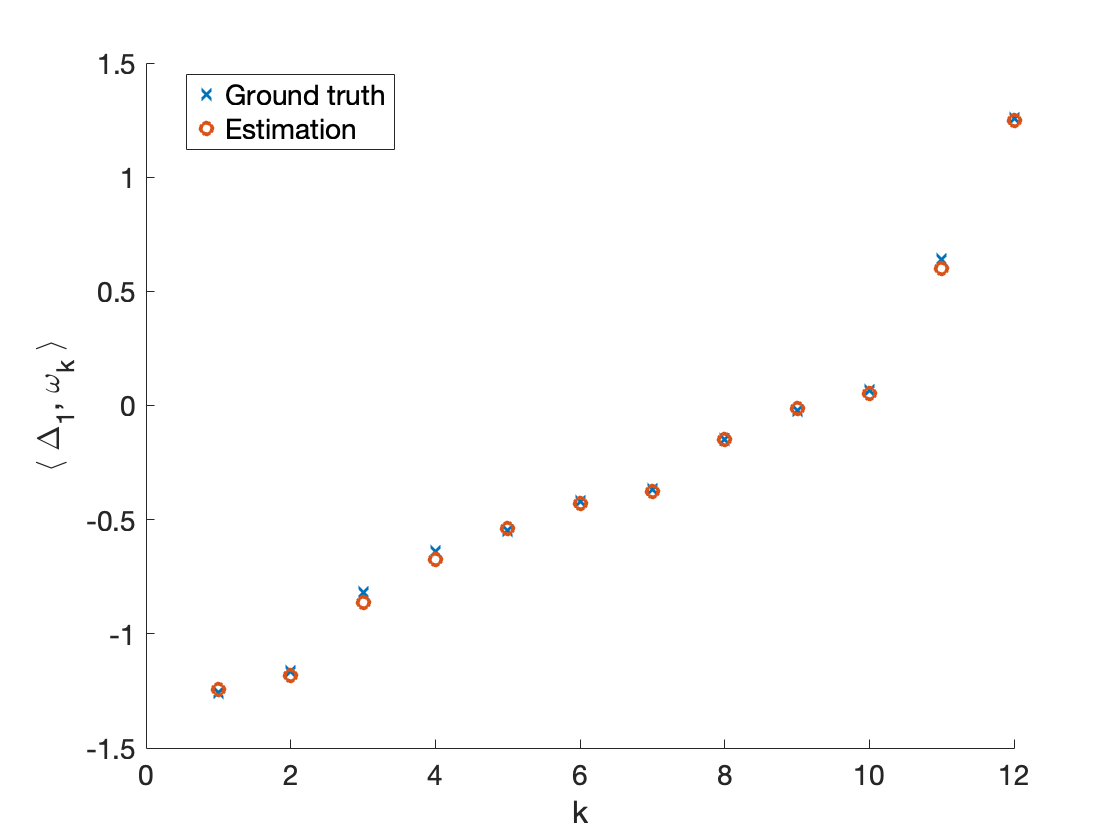}
\label{fig_second}}
\hfil
\subfloat[]{
\includegraphics[width=.3\textwidth]{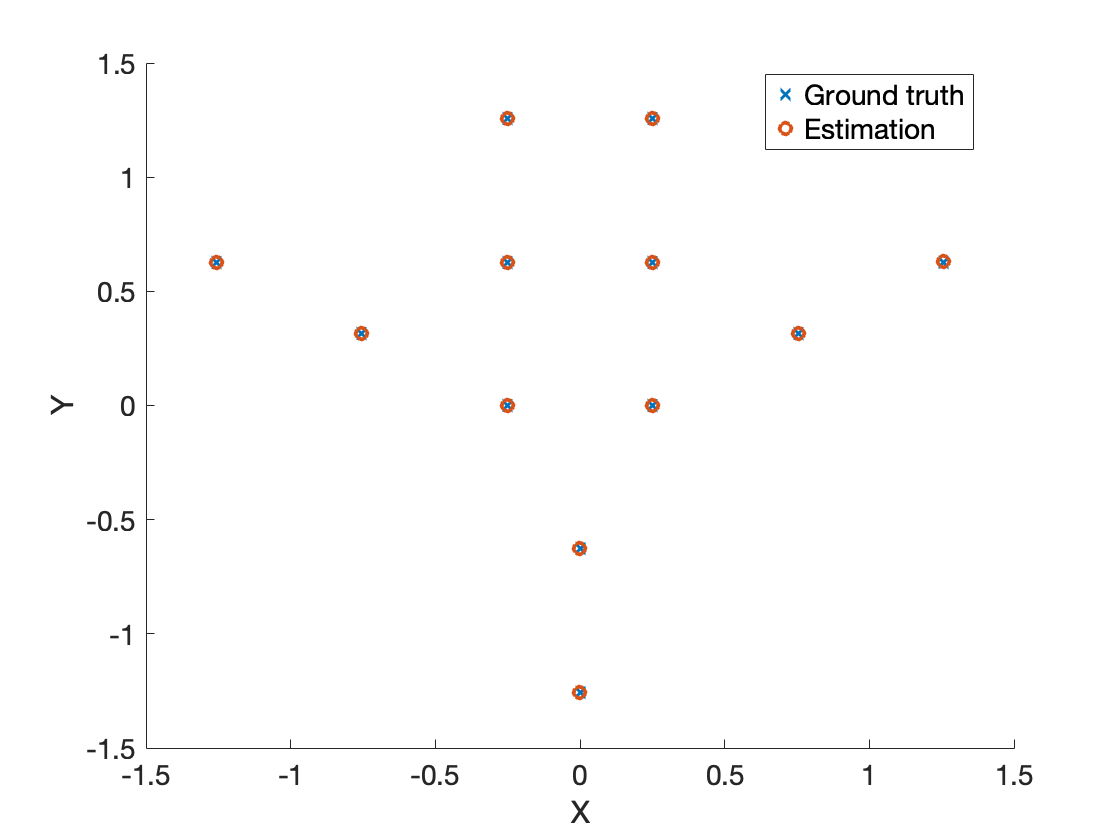}
\label{fig_third}}
\caption{(a) Accurate estimation of $\langle \Delta_2, \w_k\rangle.$ (b) Approximate estimation of $\langle \Delta_1, \w_k\rangle.$ (c) Final reconstruction.}
\label{fig:12ptsfinalstep}
\end{figure*}

In comparison, we show the results obtained by MUSIC and ESPRIT algorithms in a graphic manner in Figure~\ref{fig:music_esprit}, and as a table in Table~\ref{tab:compare_table}. 
It is obvious that both of MUSIC and ESPRIT algorithms fail at -10 dB SNR.

\begin{figure*}[ht]
	\begin{center}
	\begin{minipage}{0.44\textwidth}
	\subfloat[]{
	\includegraphics[width=0.95\textwidth]{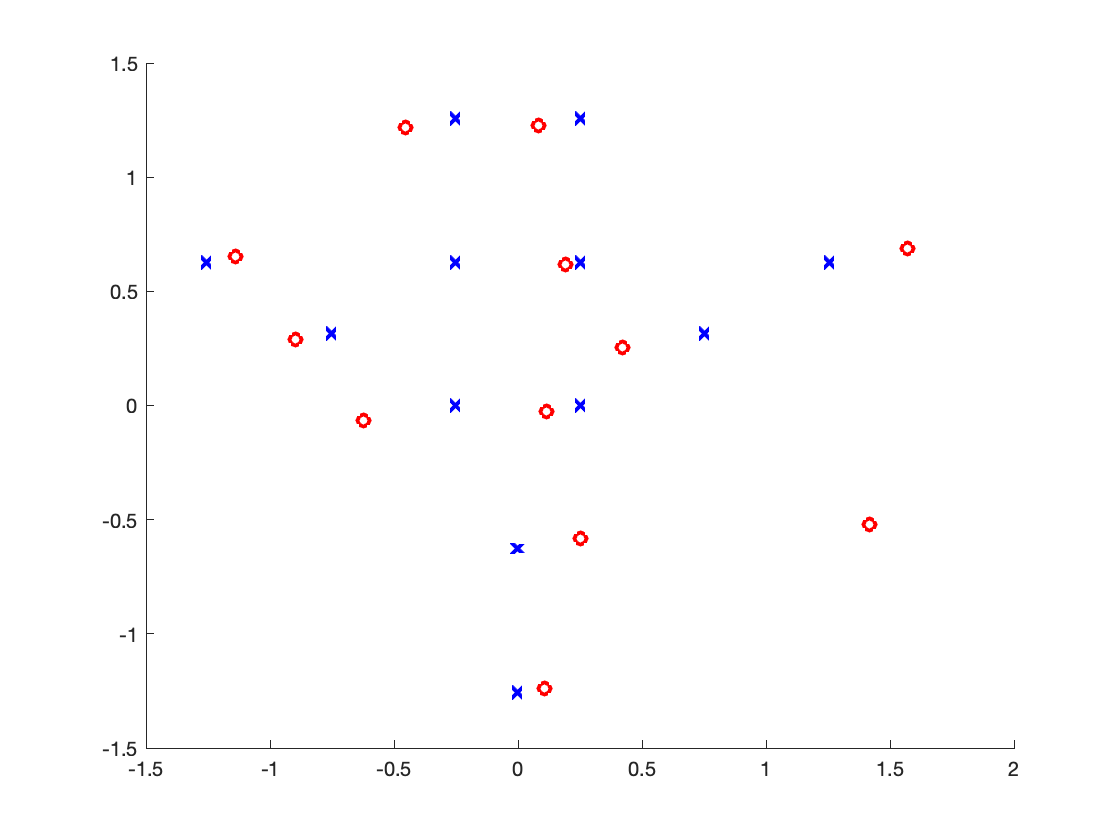}
	}
	\end{minipage}
	\begin{minipage}{0.44\textwidth}
	\hfill
	\subfloat[]{
	\includegraphics[width=0.95\textwidth]{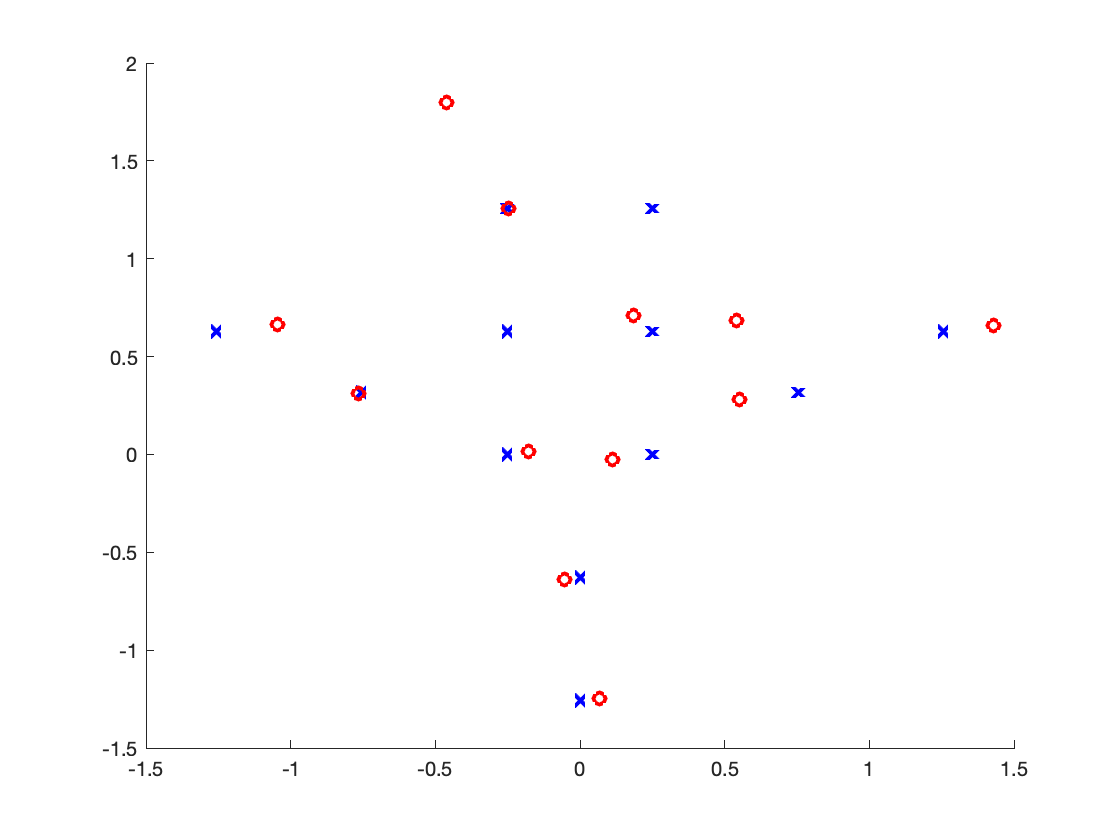}
	}
	\end{minipage}
		
	\end{center}
 	\caption{The performance of (a) MUSIC and (b) ESPRIT algorithms on the 12 point 2 dimensional data set at SNR level of -10 dB with 2048 samples.} \label{fig:music_esprit}
\end{figure*}

\begin{table*}
\begin{center}
\resizebox{\textwidth}{!}{%
\begin{tabular}{|c|c|c|c|c|c|c|c|c|}
\hline
SNR & Method & Length of & Total & Recuperated & Run-time & Memory &  RMSE & Standard \\
(dB) & & samples & points & points & (seconds) & (MB) & & Deviation\\
\hline
-10 & ESPRIT & 2048 & 12 & 12 & 9.09e-01 & 0.9769 & 2.46e-01 & 2.87e-02\\
-10 & MUSIC & 2048 & 12 & 12 & 9.56e-01 & 0.9765 & 2.69e-01 & 2.72e-02\\
-10 & Localized & 2048 & 12 & 12 & 2.14e-03 & 0.4305 & 6.23e-04 & 0\\
-5 & ESPRIT & 2048 & 12 & 12 & 8.84e-01 & 0.9769 & 2.01e-04 & 6.95e-05\\
-5 & MUSIC & 2048 & 12 & 12 & 9.16e-01 & 0.9765 & 9.64e-04 & 0\\
-5 & Localized & 2048 & 12 & 12 & 2.06e-04 & 0.4305 & 6.23e-04 & 0\\
0 & ESPRIT & 2048 & 12 & 12 & 8.93e-01 & 0.9769 & 1.11e-04 & 3.26e-05\\
0 & MUSIC & 2048 & 12 & 12 & 9.53e-01 & 0.9765 & 9.64e-04 & 0\\
0 & Localized & 2048 & 12 & 12 & 2.18e-03 & 0.4305 & 6.23e-04 & 0\\
5 & ESPRIT & 2048 & 12 & 12 & 9.17e-01 & 0.9769 & 6.26e-05 & 1.17e-05\\
5 & MUSIC & 2048 & 12 & 12 & 9.84e-01 & 0.9765 & 9.64e-04 & 0\\
5 & Localized & 2048 & 12 & 12 & 2.43e-03 & 0.4305 & 6.23e-04 & 0\\
\hline
\end{tabular}
}
\end{center}
 	\caption{The table above compares results between our algorithm, MUSIC, and ESPRIT on the 12 point 2 dimensional data set with 2048 samples.
 	Clearly, our method is faster and gives more accurate results in the presence of a low SNR.} \label{tab:compare_table}
\end{table*}

\subsection{Algorithm adaptation in the three dimensional case}\label{section:algsect2}

The details of our algorithm for the three dimensional case are as follows.

Let $\w_1,\ldots,\w_{K} \in \mathbb{R}^3$ and $\{\Delta_1,\Delta_2,\Delta_3\}$ be a basis for $\mathbb{R}^3$. Here, we have for $\ell\in\ZZ$, $|\ell|<n$, a total of $6n-3$ samples given by
\bea
\hat{\mu}(\Delta_2+\ell\Delta_1)&=\sum_{k=1}^{K} A_k \exp(-i\langle \Delta_2, \w_k\rangle) \exp(-i\ell\langle \Delta_1, \w_k\rangle) \\
    \hat{\mu}(\Delta_1+\ell\Delta_2)&=\sum_{k=1}^{K} A_k \exp(-i\langle \Delta_1, \w_k\rangle) \exp(-i\ell\langle \Delta_2, \w_k\rangle) \\
    \hat{\mu}(\Delta_1+\ell\Delta_3)&=\sum_{k=1}^{K} A_k \exp(-i\langle \Delta_1, \w_k\rangle) \exp(-i\ell\langle \Delta_3, \w_k\rangle).
\eea
By applying the low pass filter to the signals, we will get
\bea
\sigma_{n,1}(x) &= \hbar_n\sum_{k=1}^{K} A_k \exp(-i\langle \Delta_2, \w_k\rangle)\Phi_n(x-\langle \Delta_1, \w_k\rangle) \\
\sigma_{n,2}(x) &= \hbar_n\sum_{k=1}^{K} A_k \exp(-i\langle \Delta_1, \w_k\rangle)\Phi_n(x-\langle \Delta_2, \w_k\rangle) \\
\sigma_{n,3}(x) &= \hbar_n\sum_{k=1}^{K} A_k \exp(-i\langle \Delta_1, \w_k\rangle)\Phi_n(x-\langle \Delta_3, \w_k\rangle)
\eea

From the Theorem \ref{theo:main}, $\langle \Delta_1, \w_k\rangle$ will be $x$ where the peaks occurs in $|\sigma_n(x)|$,  $\langle \Delta_2, \w_k\rangle \approx Phase \left( \sigma_n(x)\right)$,  and $A_k \approx |\sigma_n(x)|$. Now, we can obtain the accurate estimation of $\langle \Delta_1, \w_k\rangle$ corresponding to less accurate estimation of $\langle \Delta_2, \w_k\rangle$.

Then, we apply the same method in $\Delta_2$ and $\Delta_3$ directions to obtain the accurate estimation of $\langle \Delta_2, \w_k\rangle$ and $\langle \Delta_3, \w_k\rangle$ corresponding to less accurate estimation of $\langle \Delta_1, \w_k\rangle$ respectively.

Finally, we can use nearest neighbor to obtain accurate estimation for all $\langle \Delta_1, \w_k\rangle$, $\langle \Delta_2, \w_k\rangle$, $\langle \Delta_3, \w_k\rangle$ and compute $A_k=|\sigma_n(x)|$.

\subsection{Results in three dimensional experiments}\label{section:3dresults}
We tested our method described in Section~\ref{section:algsect2} on two three dimensional data sets as described in \cite{cuyt2020sparse}, one comprising 29 points and the other comprising 1000 points. 
Figure~\ref{fig:3drecuperation} shows the results in a special case in a graphical manner.

\begin{figure*}[ht]
\begin{center}
\begin{minipage}{0.45\textwidth}
\subfloat[]{
\includegraphics[width=.95\textwidth]{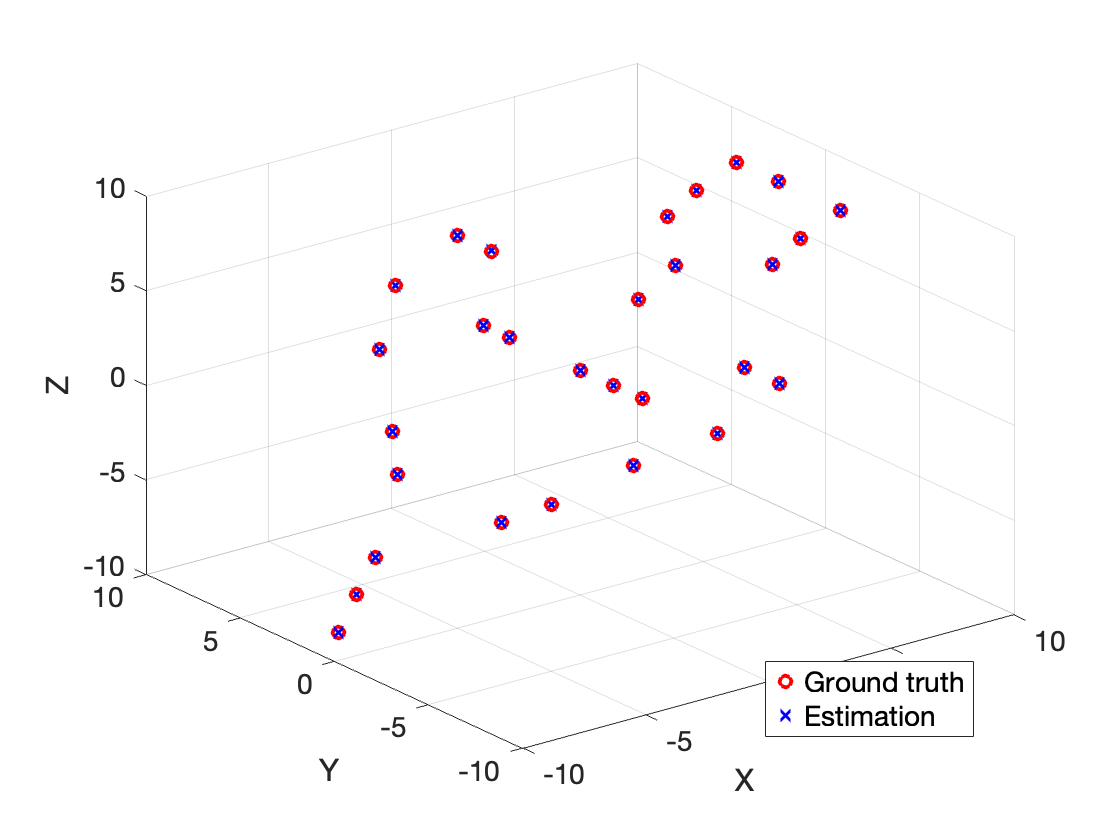}
}
\end{minipage}
\begin{minipage}{0.45\textwidth}
\hfill
\subfloat[]{
\includegraphics[width=.95\textwidth]{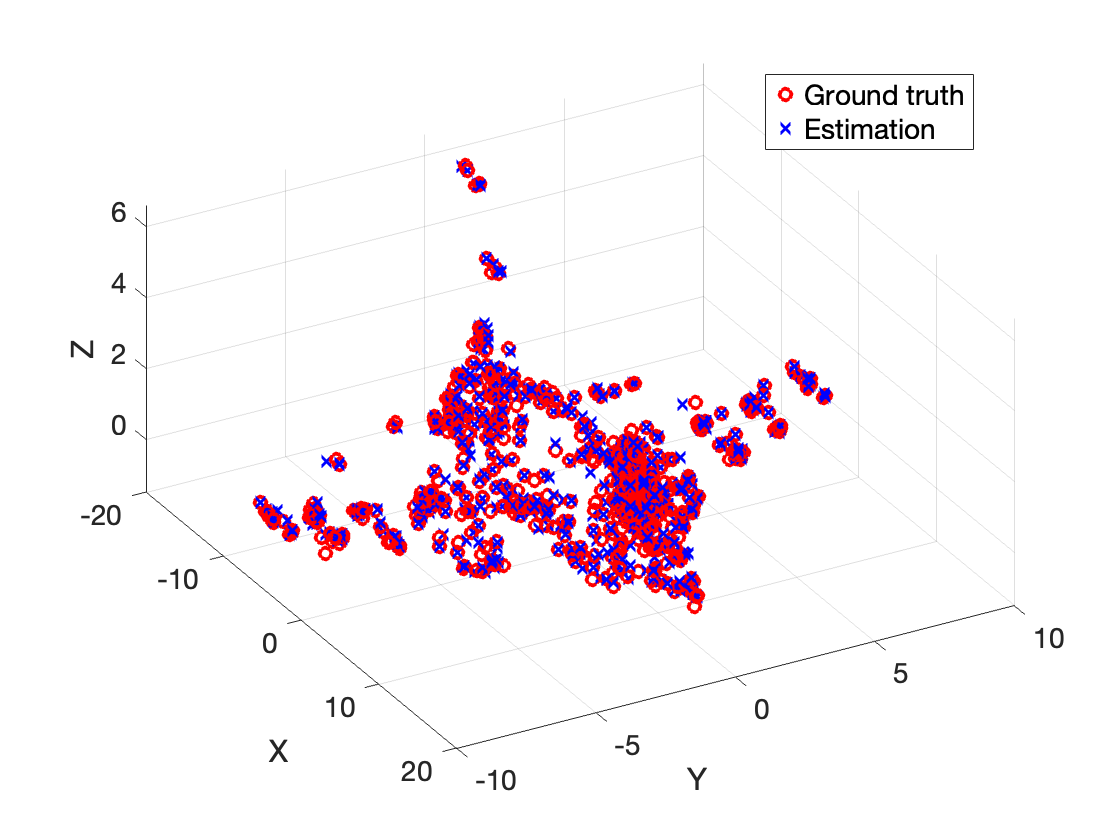}
}
\end{minipage}
\end{center}
\caption{(a) Recuperation of 29 points with 65536 samples at -10dB noise level. (b) Recuperation of 1000 points with 65536 samples at 5dB noise level.}
\label{fig:3drecuperation}
\end{figure*}

Our algorithm is able to reconstruct all 29 points of 3-d tomographic image accurately as shown in figure \ref{fig:3drecuperation}(a). We use $\Delta_1=(-0.73, -0.16, -0.66),$ $\Delta_2=(0.11, -0.98, 0.11),$ and $\Delta_3=(-2.10,1.20,3.29)$.
Our method  uses 2500 samples with 10 dB noise added to the data. 
As comparison to the original paper \cite{cuyt2020sparse}, the author required 2565 samples to reconstruct all 29 points with noise levels varying from 40 dB SNR to 5 dB SNR.

In the 1000 points of 3-d fighter jet image, our algorithm cannot separate the signals that are really close together due to high density data points on a cluster. We can reconstruct 934 data points of out 1000 points with $90000$ samples with an RMS accuracy of 15.58 cms at SNR level of 20 dB.
 Compared to our baseline results from \cite{cuyt2020sparse} at the same SNR level of 20 dB, the authors experimented with 72000, 90000, and 180000 samples, achieving the following reconstruction accuracies: for 72000 samples, 71\% of the scatterers were reconstructed within an error of at most 10 cm and 93\% within 30 cm; for 90000 samples, 81\% within 10 cm and 95\% within 30 cm; and for 180000 samples, 94\% within 10 cm and 98\% within 30 cm.
 
Figure~\ref{fig:1000pts} shows the dependence on the SNR of the accuracy and the number of points recuperated out of the 1000 points using our method.

\begin{figure*}[ht]
\begin{center}
\begin{minipage}{0.45\textwidth}
\subfloat[]{
\includegraphics[width=.95\textwidth]{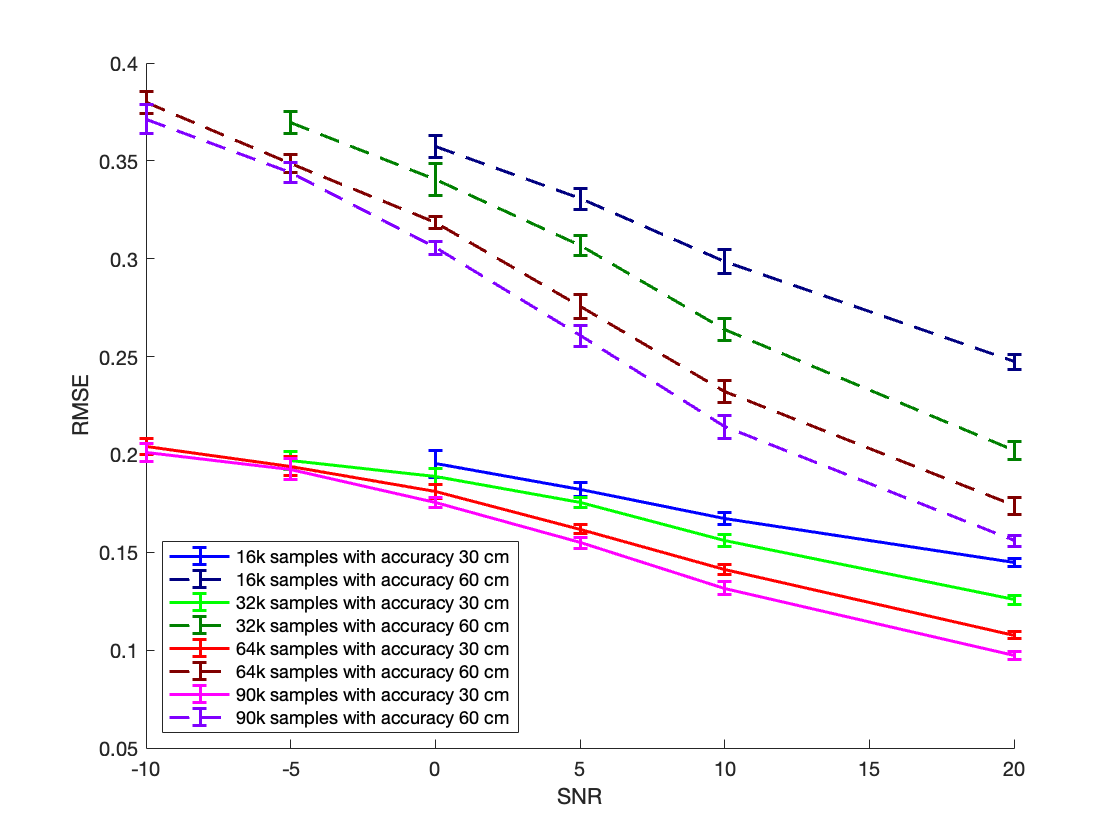}
}
\end{minipage}
\begin{minipage}{0.45\textwidth}
\hfill
\subfloat[]{
\includegraphics[width=.95\textwidth]{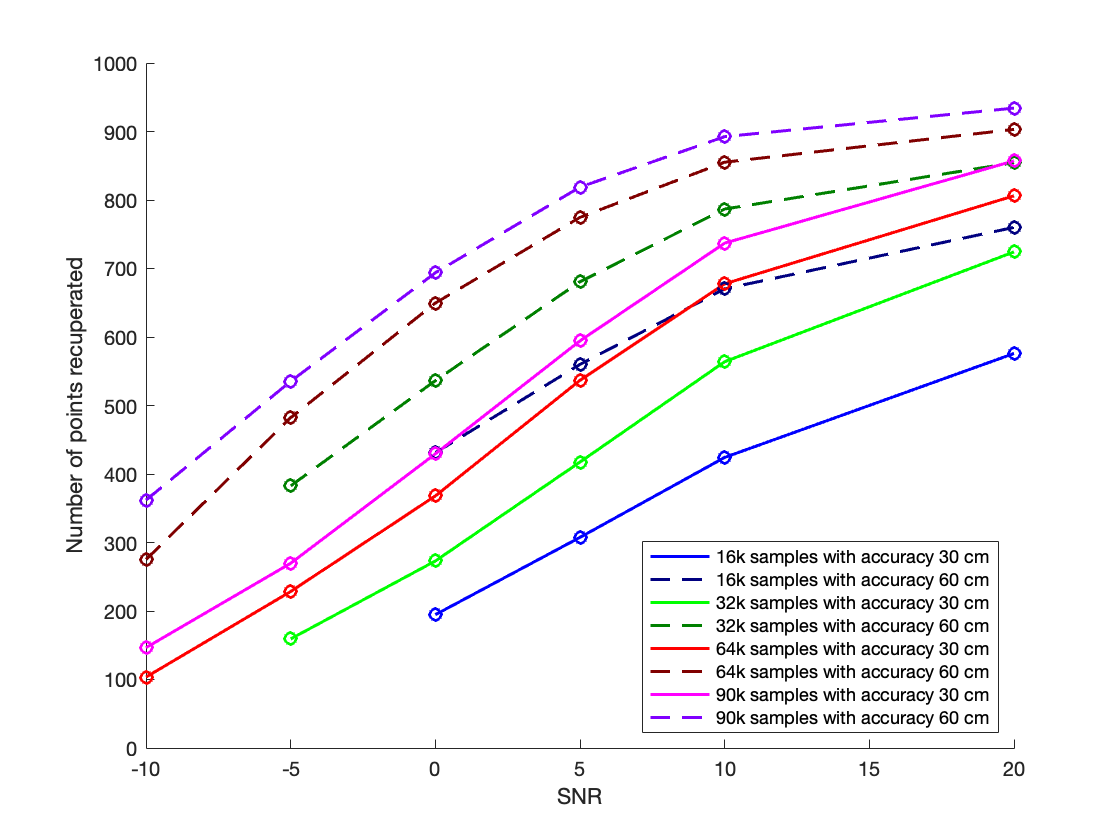}
}
\end{minipage}
\end{center}
\caption{For the 1000 point data set, dependence on SNR for (a) accuracy and (b) number of points recuperated.}
\label{fig:1000pts}

\end{figure*}

Finally, the details of all the results are summarized in Table~\ref{tab:result_table}.

\begin{table*}[ht]
\begin{center}
\resizebox{!}{0.475\textheight}{%
\begin{tabular}{ |c|c|c|c|c|c|c| } 
 \hline
 SNR & Number of & Total & Number of points & Accuracy & RMSE & Standard \\
 (dB) & samples & Points & reconstructed & (meters) & (meters) & Deviation \\
 \hline
 10 & 2500 & 29 & 29 & 0.03 & 0.0209 & 0.0062\\
 -10 & 65536 & 29 & 29 & 0.3 & 0.0007 & $0.0003$ \\
 20 & 16384 & 1000 & 575 & 0.3 & 0.1445 & $0.0020$ \\
 10 & 16384 & 1000 & 424 & 0.3 & 0.1673 & $0.0029$ \\
 5 & 16384 & 1000 & 307 & 0.3 & 0.1819 & $0.0035$ \\
 0 & 16384 & 1000 & 194 & 0.3 & 0.1952 & $0.0068$ \\
 20 & 16384 & 1000 & 760 & 0.6 & 0.2473 & $0.0037$ \\
 10 & 16384 & 1000 & 670 & 0.6 & 0.2984 & $0.0060$ \\
 5 & 16384 & 1000 & 559 & 0.6 & 0.3305 & $0.0054$ \\
 0 & 16384 & 1000 & 431 & 0.6 & 0.3572 & $0.0055$ \\
 20 & 32768 & 1000 & 724 & 0.3 & 0.1256 & $0.0023$ \\
 10 & 32768 & 1000 & 563 & 0.3 & 0.1561 & $0.0030$ \\
 5 & 32768 & 1000 & 417 & 0.3 & 0.1752 & $0.0026$ \\
 0 & 32768 & 1000 & 273 & 0.3 & 0.1886 & $0.0039$ \\
 -5 & 32768 & 1000 & 159 & 0.3 & 0.1969 & $0.0045$ \\
 20 & 32768 & 1000 & 854 & 0.6 & 0.2020 & $0.0045$ \\
 10 & 32768 & 1000 & 786 & 0.6 & 0.2638 & $0.0057$ \\
 5 & 32768 & 1000 & 680 & 0.6 & 0.3067 & $0.0052$ \\
 0 & 32768 & 1000 & 536 & 0.6 & 0.3404 & $0.0080$ \\
 -5 & 32768 & 1000 & 382 & 0.6 & 0.3695 & $0.0057$ \\
 20 & 65536 & 1000 & 806 & 0.3 & 0.1077 & $0.0018$ \\
 10 & 65536 & 1000 & 677 & 0.3 & 0.1413 & $0.0025$ \\
 5 & 65536 & 1000 & 536 & 0.3 & 0.1618 & $0.0025$ \\
 0 & 65536 & 1000 & 367 & 0.3 & 0.1810 & $0.0037$ \\
 -5 & 65536 & 1000 & 228 & 0.3 & 0.1939 & $0.0048$ \\
 -10 & 65536 & 1000 & 103 & 0.3 & 0.2040 & $0.0040$ \\
 20 & 65536 & 1000 & 903 & 0.6 & 0.1736 & $0.0043$ \\
 10 & 65536 & 1000 & 855 & 0.6 & 0.2320 & $0.0057$ \\
 5 & 65536 & 1000 & 774 & 0.6 & 0.2755 & $0.0060$ \\
 0 & 65536 & 1000 & 649 & 0.6 & 0.3185 & $0.0030$ \\
 -5 & 65536 & 1000 & 482 & 0.6 & 0.3485 & $0.0047$ \\
 -10 & 65536 & 1000 & 275 & 0.6 & 0.3798 & $0.0058$ \\
 20 & 90000 & 1000 & 857 & 0.3 & 0.0971 & 0.0019\\
10 & 90000 & 1000 & 737 & 0.3 & 0.1316 & 0.0032\\
5 & 90000 & 1000 & 594 & 0.3 & 0.1547 & 0.0026\\
0 & 90000 & 1000 & 429 & 0.3 & 0.1753 & 0.0025\\
-5 & 90000 & 1000 & 269 & 0.3 & 0.1925 & 0.0052\\
-10 & 90000 & 1000 & 146 & 0.3 & 0.2008 & 0.0047\\
20 & 90000 & 1000 & 934 & 0.6 & 0.1558 & 0.0027\\
10 & 90000 & 1000 & 892 & 0.6 & 0.2141 & 0.0058\\
5 & 90000 & 1000 & 818 & 0.6 & 0.2605 & 0.0054\\
0 & 90000 & 1000 & 693 & 0.6 & 0.3055 & 0.0033\\
-5 & 90000 & 1000 & 535 & 0.6 & 0.3441 & 0.0050\\
-10 & 90000 & 1000 & 361 & 0.6 & 0.3712 & 0.0074\\
 \hline
\end{tabular}
}
\end{center}
 	\caption{The table above shows full performances of our algorithm.} \label{tab:result_table}
\end{table*}

\bhag{Conclusions}
The problem of multidimensional exponential analysis arises in many areas of applications including tomographic imaging, ISAR imaging, antenna array processing, etc. 
The problem is essentially to develop an efficient algorithm to obtain the inverse Fourier transform of a multidimensional signal, based on finitely many equidistant samples of the signal.
We have given a very simple algorithm based on localized trigonometric kernel, which reduces the problem to a series of one dimensional problems. 
Our algorithm works with a tractable number of samples, gives high accuracy, and is very robust even in the presence of noise as high as -10dB.
We have proved theoretical guarantees under the assumption that the noise is sub-Gaussian, a significantly weaker assumption than that of white noise common in the literature.


\chapter{Localized Kernel Method for Separation of Linear Chirps}\label{2paper}

\noindent
The content presented in this chapter is adapted from our paper titled \textit{``Localized Kernel Method for Separation of Linear Chirps''}~\cite{mason2025localized}.

\bhag{Introduction}\label{bhag:intro}

The task of separating a superposition of signals into its individual components is a common challenge encountered in various signal processing applications, especially in domains such as audio \cite{makino2018,Nugraha2016,yuan2025} and radar \cite{Joshil2020,Filice2013,Dutt2019,Wang2022,Nuhoglu2023,Yanchao2018}. Due to the extensive set of applications and the frequent occurrence of signal source separation challenges, methodologies are often tailored to the specific field in question, as well as the level of prior knowledge of both the component signals and the mixing mechanism.  
We focus on a specific example that occurs in radar signal processing, specifically in the context of electronic warfare (EW) and electronic intelligence (ELINT).  

In these applications, the RF system makes measurements of the electromagnetic spectrum with a large instantaneous bandwidth, typically above $1$ GHz and continuing to increase with continual improvements in hardware capabilities.
The resulting data may consist of tens to hundreds of signals mixed together after just a few seconds to minutes of recording, many overlapping in time and/or frequency and measured at low signal-to-noise (SNR).
Most radars transmit a signal that consists of a sequence of linear frequency modulated (LFM) pulses, commonly referred to as Chirp signals.

A linear chirp is a signal (pulse) of the form $A(t)\exp(\hati\phi(t))$, where $\hati=\sqrt{-1}$ and $\phi$ is a quadratic polynomial, so that the instantaneous frequency $\phi'(t)$ is a linear function.
If $\gamma$ is the starting time of the pulse, $d$ is its duration, $\theta$ is the initial frequency, $B$ is the bandwidth, then the phase of the pulse is given by
\be\label{eq:pulsephasedef}
\phi(t)=\phi(\gamma,  \theta, B, d; t)=\theta(t-\gamma)+(B/d)(t-\gamma)^2.
\ee
\yadi{$\theta$, $\gamma$, $B$, $d$}{Chirp parameters} The pulse is assumed to repeat periodically in $m$ bursts with a pulse repetition interval PRI being the time between the end of one burst and the beginning of the next.
Thus, for a signal $s$ starting at $t=t_0^*$, its $n$-th burst is given by
\be\label{eq:pulsebursts}
\resizebox{0.9\hsize}{!}{$
s_n(t)=s_n(A,  \theta, B, t_0^*, d, m, \mbox{PRI}; t)= \begin{cases}
\begin{aligned}
&\disp A(t)\exp\left(\hati (\theta(t-n*\mbox{PRI})+(B/d)(t-n*\mbox{PRI})^2)\right),&\\
&\mbox{ if $t_0^*+n*\mbox{PRI}\le t\le t_0^*+n*\mbox{PRI}+d$, \ $n=0,\cdots, m-1$,}\\
&0, &\mbox{otherwise}.
\end{aligned}
\end{cases}
$}
\ee
\yadi{PRI}{pulse repetition interval}
\yadi{$\hati$}{$\sqrt{-1}$}
 Naturally, the chirp signal has the form $s=\sum_n s_n$.

The first step in processing the electromagnetic spectrum is the detection and parameterization step, where the time and frequency locations of individual pulses are determined, and their parameters estimated. 
The second step in the EW receiver is deinterleaving, where groups of pulses are associated with each other to obtain the different radar signals, from which further insights are extracted.  
Here, we focus on the first problem of the detection, separation, and parameter estimation of individual radar pulses.

Thus, the problem which we intend to solve is the following. 
We observe a signal of the form
\be\label{eq:observations}
F(t)=f(t)+\epsilon(t),
\ee
where $\epsilon$ is a sub-Gaussian random noise, and $f$ is the ground truth signal
\be\label{eq:groundtruth}
f(t)=\sum_{j=1}^K s(A_j, \theta_j, B_j, t_{0,j}^*, d_j, m_j, \mbox{PRI}_j;t), \qquad t\in [0,T].
\ee
From equidistant observations $F(\ell\delta)$, $\ell=0,\cdots, N$, we wish to determine the number of pulses at any time $t$ as well as the parameters $A_j, \theta_j, B_j, t_{0,j}^*, d_j, m_j, \mbox{PRI}_j$ of the pulses.
In this chapter, we will assume that all $A_j$'s are $\equiv 1$. 

The instantaneous frequency of the chirp $s_n$ in \eqref{eq:pulsebursts} at time $t$ is defined by $\theta+(2B/d)(t-n*\mbox{PRI})$.
Our approach is to assume (following \cite{ingrid2011, bspaper}) that on a sufficiently small interval of the form $[t-\Delta,t+\Delta]$, this frequency is approximated by a constant frequency at the center of the interval. 
We will then use specially constructed localized kernels to solve the problem for each of this interval as described in \cite{loctrigwave, bspaper}.
A more elaborate subdivision is necessary when the instantaneous frequencies of different pulses cross each other.
We will demonstrate experimentally that our method is very stable, capable of working with a very small SNR if the sampling rate is sufficiently high.
Thus, the only limitation is the hardware capability, not a theoretical limitation.
We will analyze theoretically the connection between $\Delta$ and the various parameters of the pulses and the accuracy.

The main contributions of this chapter are the following.
\begin{enumerate}
\item This chapter is the first work of its kind that illustrates the use of SSO in the presence of a very high noise, crossover frequencies, and various components of the signal having different start and end points.
\item In particular, we describe how the SSO algorithm needs to be modified to detect the start and stop points of the various signal components and crossover frequencies.
\item We examine the theory behind the SSO to highlight the connection between the noise level, the sampling frequency, and the length of the snippet required for SSO to work properly.
\item Our algorithms are fully automated, assuming only the  minimal separation among the instantaneous frequencies. The minimal separation can be treated as a tunable parameter, but its choice can affect the performance of the algorithm (cf. \cite{mhaskar2024robust}).
\item Our algorithm works based on FFT, and hence, is very fast.
\end{enumerate}

We will describe some related work in Section~\ref{bhag:related}. 
The mathematical foundation behind our work is explained in Section~\ref{bhag:foundations}, including the intuition and definition of the Signal Separation Operator (SSO) introduced in \cite{bspaper}. 
The algorithms and methodology to use SSO for the chirp separation problem are described in Section~\ref{bhag:algorithm} with step-by-step illustrations on a few examples.
In particular, the examples demonstrate the connection between the sampling rate, noise level, minimal separation among the instantaneous frequencies.
A full report on the entire dataset which we worked on is presented in Section~\ref{bhag:numerical}.

\bhag{Related works}\label{bhag:related}

There are several general ways to approach the problem of radar signal detection and separation.  

In the case when the frequencies are constant, the problem is known as parameter estimation in exponential sums as we introduced in Chapter \ref{section:mainresult}.
If the radar pulses have sufficiently high SNR, then they can be detected using an energy-based envelope detector operating on the complex valued time-series data, known as In-phase and Quadrature (IQ) data \cite{lehtomaki2005,Juliano2020,Mello2018}.
It is actually more common for modern solid state radars to transmit low probability of intercept (LPI) waveforms that utilize LFM pulses at low and negative SNR \cite{Xie2025, Pace2008}; in this case, the envelope detector will no longer be suitable.
In the low power setting, there are two broad classes of methods, the first operates on the IQ data directly, while the second utilizes time-frequency transforms approaching detection and separation as a 2D signal processing problem. 
In the time domain, there are analytic methods that operate on the IQ data directly, such as variations of independent components algorithm (ICA) \cite{Joshil2020,Lin2020,Jiang2015,Shao2020}, with extensions to nonlinear chirps \cite{PENG2024,zhou2018}, and incorporation of simple neural networks \cite{pang2022}.
A more common approach is to utilize time-frequency transforms and process the result as an image instead of operating on the IQ data directly \cite{Jin2022,Swiercz2020,Hou2020,Gongming2019,Liu2015,Su2022}.

A general framework that includes both the problems of chirp separation and parameter estimation in exponential sums is the following.
One observes a signal of the form
 \be\label{eq:nonstationary}
f(t)=\sum_{j=1}^K f_j(t)=\sum_{j=1}^K A_j(t)\exp(\hati \phi_j(t)), \qquad 0\le t\le T,
\ee
where the instantaneous amplitudes (IA) $A_j$'s are complex valued, and $\phi_j$'s are differentiable functions of time.
\yadi{$A_j(t)$, $\phi_j(t)$}{Instantaneous complex amplitudes and phases}

The quantity $\phi_j'(t)$ is known as the instantaneous frequency (IF) of the component $f_j$.
The problem is to estimate the IFs and the IAs.
This is also a very important and old problem, starting with a paper by a Nobel Laureate, Gabor \cite{gabor1946}.
One of the first approaches to solve the problem is the Hilbert-transform based empirical mode decomposition (EMD) by Huang \cite{huang1998}. 

\paragraph{EMD.} The following introduction of the EMD method is provided in~\cite{bspaper}. To accommodate multiple frequency components as in the constant parameter problem with \(K > 1\), Huang et al.~\cite{huang1998, huang2008} introduced a technique called \emph{empirical mode decomposition} (EMD). This method decomposes a signal \(G(t)\) into a finite sum of \emph{intrinsic mode functions} (IMFs), along with a slowly varying residual trend. Each IMF is then transformed into an AM–FM signal using the Hilbert transform, allowing for a time-frequency representation of \(G(t)\).

In our setting, we consider the empirical mode decomposition of \(G(t)\) as follows:
\begin{equation}
    G(t) = f(t) + T(t), \label{eq:emd_form}
\end{equation}
where \(f(t)\) denotes the sum of IMFs and \(T(t)\) is the residual term. More specifically, we write:
\begin{equation}
\left\{
\begin{aligned}
    f(t) &= \sum_{j=1}^{K} f_j(t), \\
    f_j(t) &= A_j(t) \cos\left(2\pi \phi_j(t)\right), \quad j = 1, \dots, K,
\end{aligned}
\right. \label{eq:imf_decomp}
\end{equation}
where each \(f_j(t)\) represents an AM–FM component with instantaneous amplitude \(A_j(t)\) and instantaneous phase \(\phi_j(t)\). 

The EMD scheme relies on the presence of sufficiently many local extrema in the given real-valued signal \( G(t) \). However, in practical settings where the trend \( T(t) \) is rapidly varying, it is common to encounter segments of the signal where local maxima or minima are sparse or entirely absent. To handle such cases, the standard EMD procedure must be refined. This leads to various improved EMD strategies proposed in \cite{huang1998, huang2008}.

A typical modified EMD scheme proceeds as follows. Let the initial signal be denoted by \( g_{1,0} := G \). The algorithm constructs two cubic spline interpolants—one through the local maxima and the other through the local minima of \( g_{1,0} \). These interpolants define the upper and lower envelopes, respectively. Their average is denoted by \( m_{1,1} \), and the first intermediate signal is defined as
\be
g_{1,1} := g_{1,0} - m_{1,1}.
\ee

This procedure is repeated iteratively: for each \( \ell \geq 1 \),
\be
g_{1,\ell} := g_{1,\ell-1} - m_{1,\ell},
\ee
where \( m_{1,\ell} \) is the mean of the upper and lower envelopes of \( g_{1,\ell-1} \). This iterative refinement is called the \emph{sifting process}, and it terminates when the resulting function satisfies two empirical criteria: (i) the number of zero crossings and local extrema differ by at most one, and (ii) the envelopes are approximately symmetric with respect to the time axis.

The output of this procedure is the first intrinsic mode function (IMF), denoted \( f_1 \). The residue \( g_{2,0} := g_{1,0} - f_1 \) is then subjected to the same sifting process to extract \( f_2 \), and so on. Formally, the \( j \)-th IMF is computed from
\be
g_{j,0} := g_{1,0} - f_1 - \cdots - f_{j-1},
\ee
followed by iterative sifting until convergence.

This iterative process yields a decomposition of the signal \( G(t) \) into \( K \) IMFs and a residual trend term \( T_K(t) \), i.e.,
\begin{equation}
    G(t) = f_1(t) + \cdots + f_K(t) + T_K(t), \label{eq:emd_decomp}
\end{equation}
where \( T_K(t) \) is a slowly varying function.

\vspace{1em}

However, there is no mathematical theory behind this approach.
A mathematically rigorous approach, synchrosqueezing transform (SST), was given by Daubeschies and her collaborators \cite{ingrid2011, daubechies1996nonlinear}.
These seminal papers have given rise to a lot of variations and improvements.
We refer to \cite{li2022chirplet} for  more references on this topic.
In this chapter, we will use the algorithm given in \cite{thakur2013synchrosqueezing} for comparing our method with SST.

\paragraph{SST.} To formalize the types of signals SST can handle, Daubechies et al.~\cite{ingrid2011} define a class of functions called intrinsic mode type (IMT) components.

\begin{definition}[Intrinsic Mode Type (IMT) Function] \label{SST1}
A function \be f(t) = A(t) e^{i \phi(t)}\ee is said to be an IMT function with accuracy $\epsilon > 0$ if:
\begin{itemize}
    \item $A \in C^1(\mathbb{R}) \cap L^\infty(\mathbb{R})$, and $\phi \in C^2(\mathbb{R})$,
    \item $\inf_{t \in \mathbb{R}} \phi'(t) > 0$ and $\sup_{t \in \mathbb{R}} \phi'(t) < \infty$,
    \item $\left| A'(t) \right| \leq \epsilon \phi'(t)$ and $\left| \phi''(t) \right| \leq \epsilon \phi'(t)$ for all $t \in \mathbb{R}$.
\end{itemize}
\end{definition}

\begin{definition}[Well-Separated Superposition of IMT Components] \label{SST2}
A function $f(t)$ is a superposition of $K$ well-separated IMT components up to accuracy $\epsilon$ and separation $d > 0$ if:
\be
f(t) = \sum_{k=1}^{K} f_k(t) = \sum_{k=1}^{K} A_k(t) e^{i \phi_k(t)},
\ee
where each $f_k$ is an IMT function, and their respective instantaneous frequencies $\phi_k'(t)$ satisfy
\be
\phi_k'(t) > \phi_{k-1}'(t), \quad \text{and} \quad \frac{\left| \phi_k'(t) - \phi_{k-1}'(t) \right|}{\phi_k'(t) + \phi_{k-1}'(t)} \geq d, \quad \forall t \in \mathbb{R}.
\ee
\end{definition}

We denote the set of such functions by $\mathcal{A}_{\epsilon, d}$.

The SST algorithm first computes the continuous wavelet transform $W_f(a, b)$ of a signal $f \in \mathcal{A}_{\epsilon, d}$, and then sharpens the representation by reallocating the time-scale coefficients to the time-frequency plane using an estimated instantaneous frequency.

The central result justifying the method is the following:

\begin{theorem}[Daubechies et al.~{\cite[Theorem 3.3]{ingrid2011}}]
Let $f \in \mathcal{A}_{\epsilon, d}$ and define $\tilde{\epsilon} := \epsilon^{1/3}$. Let $\psi$ be a Schwartz-class wavelet whose Fourier transform $\hat{\psi}$ is supported in $[1 - \Delta, 1 + \Delta]$, where $\Delta < \frac{d}{1 + d}$.

Define the synchrosqueezed transform $S^{\delta}_{f, \tilde{\epsilon}}(b, \omega)$ by
\be
S^{\delta}_{f, \tilde{\epsilon}}(b, \omega) := \int_{A_{\tilde{\epsilon}, f}(b)} W_f(a, b) \cdot \frac{1}{\delta} h\left( \frac{\omega - \omega_f(a, b)}{\delta} \right) a^{-3/2} \, da,
\ee
where $A_{\tilde{\epsilon}, f}(b) = \{ a > 0 : |W_f(a, b)| > \tilde{\epsilon} \}$, and $\omega_f(a, b)$ is the instantaneous frequency estimate.

Then, for sufficiently small $\epsilon$, the following holds:
\begin{enumerate}
    \item $|W_f(a, b)| > \tilde{\epsilon}$ only when $(a, b)$ satisfies $|a \phi_k'(b) - 1| < \Delta$.
    \item For such $(a, b)$, the estimated frequency satisfies $|\omega_f(a, b) - \phi_k'(b)| \leq \tilde{\epsilon}$.
    \item For each $k \in \{1, \ldots, K\}$ and $b \in \mathbb{R}$,
    \be
    \left| \int_{|\omega - \phi_k'(b)| < \tilde{\epsilon}} S^{\delta}_{f, \tilde{\epsilon}}(b, \omega) \, d\omega - A_k(b) e^{i \phi_k(b)} \right| \leq C \tilde{\epsilon},
    \ee
    for some constant $C$.
\end{enumerate}
\end{theorem}

This theorem shows that SST achieves a sharp concentration in the time-frequency plane and allows for the accurate recovery of each component’s instantaneous frequency and amplitude. Consequently, SST serves as a rigorous and adaptive alternative to empirical mode decomposition (EMD), especially for signals composed of well-separated AM-FM components.

\vspace{1em}

In \cite{bspaper}, we have applied the technique based on localized trigonometric kernels to define a signal separation operator (SSO) to solve the constant parameter problem under conditions weaker than those in \cite{ingrid2011}.
Our conditions do assume that all the components are present in the interval $[0,T]$; i.e., the number of signals $f_j$ is the same throughout the interval, and there is no discontinuity in any of these.
Our method finds the IFs by looking at the peaks of a power spectrum. The IAs are estimated by a simple substitution compared to the limiting process required in SST. 
Finally, our method can be implemented very fast using FFT.

The paper \cite{li2022chirplet} deals specifically with the separation of chirplets, based on a modification of a filter described in \cite{bspaper}.
It is not clear when the various conditions in this chapter are satisfied.
Also, the method requires a three dimensional search.
In this chapter, we will demonstrate how the methods developed in \cite{bspaper} can be adapted for the chirp separation problem, including the case of crossover frequencies.

\bhag{Mathematical foundations}\label{bhag:foundations}

In Section~\ref{bhag:reduction}, we explain  how the general problem of blind source signal separation can be reduced to the constant parameter case at each point under certain assumptions as in \cite{ingrid2011, bspaper}. 
A theoretical analysis of the  performance of SSO in the presence of noise is given in Section~\ref{bhag:sso}.

\subsection{Reduction to the constant parameter case}\label{bhag:reduction}

In this section, we describe certain conditions under which the data of the form \eqref{eq:nonstationary} can be reduced to the constant parameter data as in Chapter~\ref{section:mainresult}.
Thus, we are interested in finding the instantaneous frequencies $\phi_j'(t^*)$ and amplitudes $A_j(t^*)$ in the signal
\be\label{eq:nonstationarybis}
f(t)=\sum_{j=1}^K f_j(t)=\sum_{j=1}^K A_j(t)\exp(\hati\phi_j(t)),
\ee
for some point $t^*\in\RR$, where $A_j$'s are continuous and $\phi_j$'s are continuously differentiable functions on an interval of the form $[t^*-\Delta, t^*+\Delta]$.
We introduce the notation (abusing the notation introduced in \eqref{eq:notation})
\be\label{eq:notation2}
\mathbf{M}(t^*)=\sum_{j=1}^K |A_j(t^*)|, \qquad B(t^*)=\max_{1\le j\le K}|\phi_j'(t^*)|.
\ee
The following theorem is a reformulation of \cite[Theorem~4.2]{bspaper}. 
We will reproduce the proof for the sake of completeness.
\yadi{$\Delta$}{half-length of time interval for snippet}
\yadi{$t_k$}{center of the time interval for snippet}
\yadi{$I_k$}{time interval for snippet $[t_k-\Delta,t_k+\Delta]$}
\yadi{$R$}{sampling frequency}

Note that our Theorem \ref{theo:reduction} does not require the conditions in Definition \ref{SST1} and Definition \ref{SST2}, unlike the SST method.

\begin{theorem}\label{theo:reduction}
Let $t^*\in\RR$, $\alpha, \Delta>0$, and for $|u|\le \Delta$,
\be\label{eq:slowvary}
|A_j(t^*+u)-A_j(t^*)| \le \alpha |u||A_j(t^*)|, \qquad |\phi_j'(t^*+u)-\phi_j'(t^*)| \le \alpha |u||\phi_j'(t^*)|.
\ee 
Then we have
\be\label{eq:reduction}
\left|f(t^*+u)-\sum_{j=1}^K f_j(t^*)\exp\left(\hati \phi_j'(t^*)u\right)\right| \le \alpha \mathbf{M}(t^*)(B(t^*)\Delta +1)\Delta, \qquad |u|\le |\Delta|.
\ee
\end{theorem}

\begin{proof}\ 
In this proof, we write $\mathbf{M}=\mathbf{M}(t^*)$ and $B=B(t^*)$. Let $|u|\le \Delta$. 
We observe that
\be\label{eq:pf1eqn1_navy}
\begin{aligned}
\left|\sum_{j=1}^K A_j(t^*+u)\right.&\left.\exp(\hati \phi_j(t^*+u))-\sum_{j=1}^K A_j(t^*)\exp(\hati \phi_j(t^*))\exp(\hati u \phi_j'(t^*))\right|\\
&\le \left|\sum_{j=1}^K \left(A_j(t^*+u)-A_j(t^*)\right)\exp(\hati \phi_j(t^*+u))\right|\\
&\qquad +\left|\sum_{j=1}^K A_j(t^*)\left(\exp(\hati \phi_j(t^*+u))-\exp(\hati \phi_j(t^*))\exp(\hati u \phi_j'(t^*))\right)\right|
\end{aligned}
\ee
In view of the fact that
\be
|e^{ix}-e^{iy}|=2|\sin((x-y)/2)| \le |x-y|,
\ee
and \eqref{eq:slowvary}, we deduce that
 for each $k$,
\be\label{eq:pf1eqn2_navy}
\begin{aligned}
\left|\exp(\hati \phi_j(t^*+u))\right.&\left.-\exp(\hati \phi_j(t^*))\exp(\hati u \phi_j'(t^*))\right|\le \left|\phi_j(t^*+u)-\phi_j(t^*)-u \phi_j'(t^*)\right|\\
& \le \left|\int_0^u \left(\phi_j'(t^*+v)-\phi_j'(t^*)\right)dv\right|\le \alpha B\left|\int_0^u |v|dv\right|\le \alpha B|u|^2\le \alpha B\Delta^2.
\end{aligned}
\ee
Using \eqref{eq:slowvary} again, we see that
\be
\left|\sum_{j=1}^K \left(A_j(t^*+u)-A_j(t^*)\right)\exp(\hati \phi_j(t^*+u))\right|\le \mathbf{M}\Delta.
\ee
Together with \eqref{eq:pf1eqn2_navy} and \eqref{eq:pf1eqn1_navy}, this leads to \eqref{eq:reduction}.
\end{proof}

\subsection{Signal separation operator}\label{bhag:sso}
In this section, we fix $t^*\in\RR$.
The quantities denoted below by $\alpha$, $\Delta$, $R$ are independent of $t^*$, and so are the constants involved in $\ls$, $\gs$, and $\sim$, but the other quantities will depend upon $t^*$ without its mention in the notation.

Let $R >0$ denote the sampling frequency.  We write
\be\label{eq:reducedfreq}
\lambda_j=\phi_j'(t^*)/R, \quad n=\lfloor R\Delta\rfloor, \quad \mu=\sum_{j=1}^K f_j(t^*)\delta_{\lambda_j}, \quad \hat{\mu}(\ell)=\sum_{j=1}^K f_j(t^*)\exp(-\hati \ell\lambda_j).
\ee
Then using Theorem~\ref{theo:reduction}, it is easy to deduce that (cf. \eqref{eq:powerspectrum})
\be\label{eq:ssointro}
\left|\hbar_n\sum_{|\ell|<n} H\left(\frac{|\ell|}{n}\right)f(t^*-\ell/R)e^{\hati \ell x}-\sigma_n(\hat{\mu})(x)\right|\le \alpha \mathbf{M}(t^*)(B(t^*)\Delta +1)\Delta.
\ee
This leads to the following definition (cf. \cite[Definition~2.3]{bspaper})
\begin{definition}\label{def:ssodef}
Let $F :\RR\to \CC$, $t^*\in\RR$, $R>0$.
 We define the \textbf{\textit{Signal Separation Operator}} SSO by
\be\label{eq:ssodef}
\mathcal{T}_{n,R}(F)(t^*;x)=\hbar_n\sum_{|\ell|<n} H\left(\frac{|\ell|}{n}\right)f(t-\ell/R)e^{\hati \ell x}, \qquad x\in\RR.
\ee
\yadi{$\mathcal{T}_{n,R}$}{Signal separation operator}
\end{definition}
We will use the operator to separate the components $f_j(t^*)$ and the instantaneous frequencies $\phi_j'(t^*)$ based on the noisy observations
\be\label{eq:nonstationary_noisy}
F(t)=\sum_{j=1}^K f_j(t) +\epsilon(t)=\sum_{j=1}^K A_j(t)\exp(i\phi_j(t))+\epsilon(t),
\ee
where for each $t$, $\epsilon(t)$ is a realization of a complex sub-Gaussian random variable in $\mathcal{G}(V)$.
Theorem~\ref{theo:main} can be translated into Theorem~\ref{theo:main_navy} below, where we abuse the notation to make the comparison easier. 
We note that  Theorem~\ref{theo:main_navy} itself is a restatement of \cite[Theorems~2.4,4.1]{bspaper} in a somewhat modified form taking the noise into account.
In order to state this theorem, we write (in addition to \eqref{eq:notation2})
\begin{equation}\label{eq:notation3}
 \ \mathfrak{m}(t^*)=\min_{1\le \ell\le K}|f_\ell(t^*)|, \ \ \eta(t^*)=\min_{\ell\not=k} |\phi_\ell'(t^*)-\phi_k(t^*)|, \ \eta^*=\eta(t^*)/R.
\end{equation}
For brevity, we will omit the mention of $t^*$ from our notation, unless we feel that this might cause some confusion.
\begin{theorem}\label{theo:main_navy}
Let $0<\delta<1$, $t^*\in\RR$, $R, \Delta>0$, and the assumptions of Theorem~\ref{theo:reduction} be satisfied. 
We assume that  $n=\lfloor R\Delta\rfloor$ satisfies the condition \eqref{eq:noiseest},  and in addition, $n\ge 4C/\eta^*$. 
We assume further that
\be\label{eq:Deltacond}
\alpha \mathbf{M}(B\Delta +1)\Delta\le \mathfrak{m}/4.
\ee
We define
\be\label{eq:levelsetbis}
\mathbb{G}(t^*)=\{x\in\TT : |\mathcal{T}_{n,R}(F)(t^*;x)|\ge 3\mathfrak{m}/4\}.
\ee
Then
each of the following statements holds with probability exceeding $1-\delta$.
\begin{itemize}
\item (\textbf{Disjoint union condition}) \\
The set $\mathbb{G}(t^*)$ is a disjoint union of exactly $K$ subsets $\mathbb{G}_\ell(t^*)$,
\item (\textbf{Diameter condition}) \\
For each $\ell=1,\cdots, K$,   $\mathsf{diam}(\mathbb{G}_\ell(t^*)) \le 2C/n$,
\item (\textbf{Separation}) \\
$\mathsf{dist}(\mathbb{G}_\ell(t^*), \mathbb{G}_j(t^*)) \ge \eta/2$ for $\ell \neq k$,
\item (\textbf{Interval inclusion}) \\
For each $\ell=1,\cdots, K$,  $I_\ell=\{x\in\TT: |x-\lambda_\ell|\le 1/(4n)\}\subseteq \mathbb{G}_\ell(t^*)$.
\end{itemize}
Moreover, if
\begin{equation}\label{eq:lambda_estimator_nonstationary}
\hat{\lambda}_\ell =\arg\max_{x\in \mathbb{G}_\ell(t^*)}|\mathcal{T}_{n,R}(F)(t^*;x)|,\quad \widehat{\phi_\ell'(t^*)}=R\hat{\lambda}(\ell),
\end{equation}
then
\begin{equation}\label{eq:lambdaerr_nonstationary}
|\hat{\lambda}_\ell-\lambda_\ell| \le 2C/n.
\end{equation}
\end{theorem}

\begin{remark}\label{rem:ncond}
{\rm
For the convenience of the reader, we summarize the conditions on the sampling frequency $R$ below.
We assume  that for every $t^*$ of interest,
\be\label{eq:rcond1}
R\Delta \ge \max\left(C_1, 4C/\eta^*\right)
\ee
and $R$ is large enough so that with probability $>1-\delta$,
\be\label{eq:recond2}
|E_n(x)|\le C_2V\sqrt{\frac{(\log(C_2R\Delta)/\delta)}{R\Delta}}
\ee
We note that the condition \eqref{eq:rcond1} puts a lower bound on $\Delta$, while the condition \eqref{eq:Deltacond} is an upper bound on the Lipschitz constant $\alpha$. 
In particular, the choice of $\Delta$ is a theoretically delicate one.
In our experiments, this was done by experimentation.
\qed}
\end{remark}
\begin{remark}\label{rem:chirps}
{\rm
In the case of a chirp
\be
\phi'(t)=\begin{cases}
\omega+\frac{B}{d}(t-\gamma), &\mbox{ if $t_0\le t\le t_0+d$}\\
0, &\mbox{otherwise},
\end{cases}
\ee
 we have for $t^*\in [t_0, t_0+d]$
\be
|\phi'(t^*+u)-\phi'(t^*)|=\left|\frac{B}{d}\right||u|, \qquad |\phi'(t^*)|=\left|\omega+\frac{B}{d}(t^*-\gamma)\right|
\ee
So, we may take 
\be
\alpha=\frac{|B/d|}{\min_{t\in [t_0, t_0+d]}|\omega+(B/d)(t-\gamma)|}.
\ee

\qed}
\end{remark}

\bhag{Algorithms}\label{bhag:algorithm}
In this section, we describe how the SSO can be used for signals of the form \eqref{eq:groundtruth}.
The basic idea behind our algorithms is to apply Theorem~\ref{theo:main_navy} to different snippets of the signal. 
Thus, we choose a small $\Delta$ and divide the signal duration into $D$ overlapping time  intervals $I_k=[t_k-\Delta, t_k+\Delta]$, $k=1,\cdots, D$.
The $t_k$'s are chosen such that $I_k \cap I_{k+1} \neq \emptyset$. This is to ensure that all the sampling data are considered in this experiment.
Theorem~\ref{theo:main_navy} gives guidelines on how to select $\Delta$. 
In practice, one has to treat this as a tunable parameter.
In our experiments reported here, we chose $\Delta$ by experimentation so that it is large enough to eliminate the noise effect, but small enough to give an accurate estimation. 
A large value of $\Delta$ implies that there are more data points to average out the noise, but this has higher chance that the signal will start or end within the interval which leads to inaccurate estimation.

  The signal restricted to $I_k$ has the form (cf. \eqref{eq:groundtruth}):
\be\label{eq:ground_truth_reformed}
F_k(t)=\sum_{j=1}^{J_k}f_{j,k}(t)+\epsilon(t), \qquad f_{j,k}(t)=A_{j,k}(t)\exp(\hati\phi_{j,k}(t)), 
\ee
 where $\epsilon(t)$ represents the noise and
 \be\label{eq:phase_bis}
  \phi_{j,k}(t)=\omega_{j,k}(t-\gamma_{j,k})+a_{j,k}(t-\gamma_{j,k})^2/2.
\ee
\yadi{$\omega_{j,k}$, $\gamma_{j,k}$, $a_{j,k}$}{Chirp parameters for chirp $j$ in $I_k$}
\yadi{$\Lambda_{j,k}$}{estimated IF for chirp $j$ at $t_k$}
\yadi{$\tau$}{threshold for SSO, estimate for $3\mathfrak{m}/4$}
The SSO will be used with each $F_k$ to obtain the parameter values for each $f_{j,k}$. 
In our experiments, we have focused on IFs, assuming all the amplitudes to be $1$.
Taken together, the set $(t_k,\widehat{\phi_{j,k}'(t_k)})$ is called the \textbf{SSO diagram}.
Once we determine the IFs, the amplitudes can be determined by simply substituting the estimated IFs in SSO and more accurately by solving an appropriate least square problem. 
We postpone a more detailed investigation of this question to a future work.

 Our methodology requires a prior knowledge of 
\begin{enumerate}
\item a minimum separation $\eta$ among all the instantaneous frequencies
\item the receiver bandwidth 
\be
B_{\mbox{\scriptsize{rec}}} \geq \max_{j,k,t}\phi_{j,k}'(t)-\min_{j,k,t}\phi_{j,k}'(t).
\ee
\end{enumerate}
Out of these, $B_{\mbox{\scriptsize{rec}}}$ is a system parameter and is always known a priori.  
The condition implies that all the signals are fully captured in the measurements. 
\yadi{$B_{\mbox{\scriptsize{rec}}}$}{Receiver bandwidth}
The minimal separation $\eta$ can be found experimentally as well, and its choice will affect the performance of our algorithms.
We have demonstrated this in detail in Chapter \ref{section:mainresult}.

The clusters defined in Theorem~\ref{theo:main_navy} are found, using a clustering algorithm.  For example, in this chapter we used DBSCAN \cite{ester1996density}.
The notation we use in our algorithm is 
\be\label{eq:dbscabdef}
\resizebox{0.9\hsize}{!}{$
([\mbox{scatter data}], [\mbox{label}]) \leftarrow\mbox{DBSCAN}([\mbox{scatter data}], [\mbox{neighborhood radius}], [\mbox{minimum number of neighbors}]).
$}
\ee
Here, $[\mbox{scatter data}]$ is a data matrix,  and the algorithm attaches a label $1$ for the cluster of interest (and $-1$ otherwise), based on the number of points neighboring the so called core points in a ball of radius denoted above by $[\mbox{neighborhood radius}]$.

Our algorithm needs a few more tunable parameters as summarized in Table \ref{tab:tunable_parameters}.

\begin{table}[H]
\begin{center}
\resizebox{0.9\hsize}{!}{
\begin{tabular}{ |c|l| } 
 \hline
 Parameters & Description \\
 \hline
 $\Delta$ & Time interval width centered at $t_k$, which is chosen based on Theorem \ref{theo:main_navy}\\
 $D$ & Number of intervals centered at $t_k$, which is chosen so that interval $I_k$'s are overlapping\\
 $D_1$ & Minimum number of neighbors, which can be tuned based on \\
 & the number of points in the SSO diagram in the radius given by $B_{\mbox{\scriptsize{rec}}}$.\\
 $D_2$ & Minimum number of neighbors, which can be tuned based on \\
 & the number of points in the smallest signal cluster\\
$M$ & Number of partitions when signal crossover detected, which can be tuned so that \\
& some partitions contain sufficiently long linear chirps \\
 \hline
\end{tabular}
}
\end{center}
 	\caption{The table above shows a list of tunable parameters.} \label{tab:tunable_parameters}
\end{table}

Algorithm~\ref{alg:univariate} is designed to find the IFs. 
A challenge here is to determine the correct threshold related to the unknown minimal amplitude $\mathfrak{m}$ in the presence of noise.
A value which is too small would indicate superfluous frequencies, a value which is too large would miss some frequencies. This effect is discussed in Figure \ref{fig:fail_case}.
There are other details such as detecting the clusters and find the right peaks which are described below.
Part 2 will estimate the parameters of the instantaneous frequencies obtained by part 1 using least squares.
The challenge here is to figure out where the different signal components begin and end, and which part of the signal to use in order to set up the least square problem in order to obtain reliable solutions.
SSO does not work directly when there are crossover frequencies.
In fact, the set up itself does not make sense in general in such cases; e.g., whether the crossover of the form $X$ represents two linear frequencies crossing or two $V$-shaped frequencies touching each other.
 Part 3 is the refinement process where we use the facts that SSO is inherently unable to resolve crossovers and that we are dealing with linear chirps
  to detect the presence of such frequencies.
 If sufficiently long linear chirps are present on both sides of the crossover, then we can detect those using the previous two algorithms with a smaller value of $\Delta$ and interpolate to determine where the crossover happens and which chirps cross each other.

 Figure~\ref{fig:flowchart1} and  \ref{fig:flowchart2} give flowcharts of how the algorithms are related to each other.
 The whole process will be illustrated step-by-step with one example in this section, and analyzed further in Section~\ref{bhag:numerical} with a few other examples.

\begin{table}[ht]
\begin{center}
\begin{tabular}{ |c|c|c|c|c|c|c|c| } 
 \hline
 $j$ & $A_j$ & $\omega_j$ & $B_j$ & $d_j$ & $t^*_{0,j}$ & $\mbox{PRI}_j$ & Number of pulses \\
 \hline
 1 & 1 & 1080000000 & 15000000 & 3.0e-05 & 1.0e-05 & 5.0e-05 & 2 \\
 2 & 1 & 1360000000 & 5000000 & 1.0e-05 & 1.0e-05 & 1.5e-05 & 5 \\
 3 & 1 & 1540000000 & -20000000 & 3.0e-05 & 1.0e-05 & 0 & 1 \\
 4 & 1 & 1510000000 & 50000000 & 7.0e-05 & 1.5e-05 & 0 & 1 \\
 5 & 1 & 1480000000 & -15000000 & 5.0e-05 & 3.0e-05 & 0 & 1 \\
 6 & 1 & 1040000000 & -15000000 & 3.0e-05 & 1.5e-05 & 4.0e-05 & 2 \\
 \hline
\end{tabular}
\end{center}
 	\caption{The table above shows an example of signal ground truth parameters.} \label{tab:signal_table}
\end{table}

We illustrate our method using the signal with parameters as described in Table~\ref{tab:signal_table}. 
In this example, we use $\epsilon(t)$ to be a Gaussian noise at -10 dB SNR.
 The visualizations of  the raw signal $f(t)$ and its instantaneous frequency $\phi'(t)$ corresponding to the parameters in Table~\ref{tab:signal_table} are shown in Figure ~\ref{fig:signal_ground_truth}.

\begin{figure}[ht]
\begin{center}
\includegraphics[scale=.135]{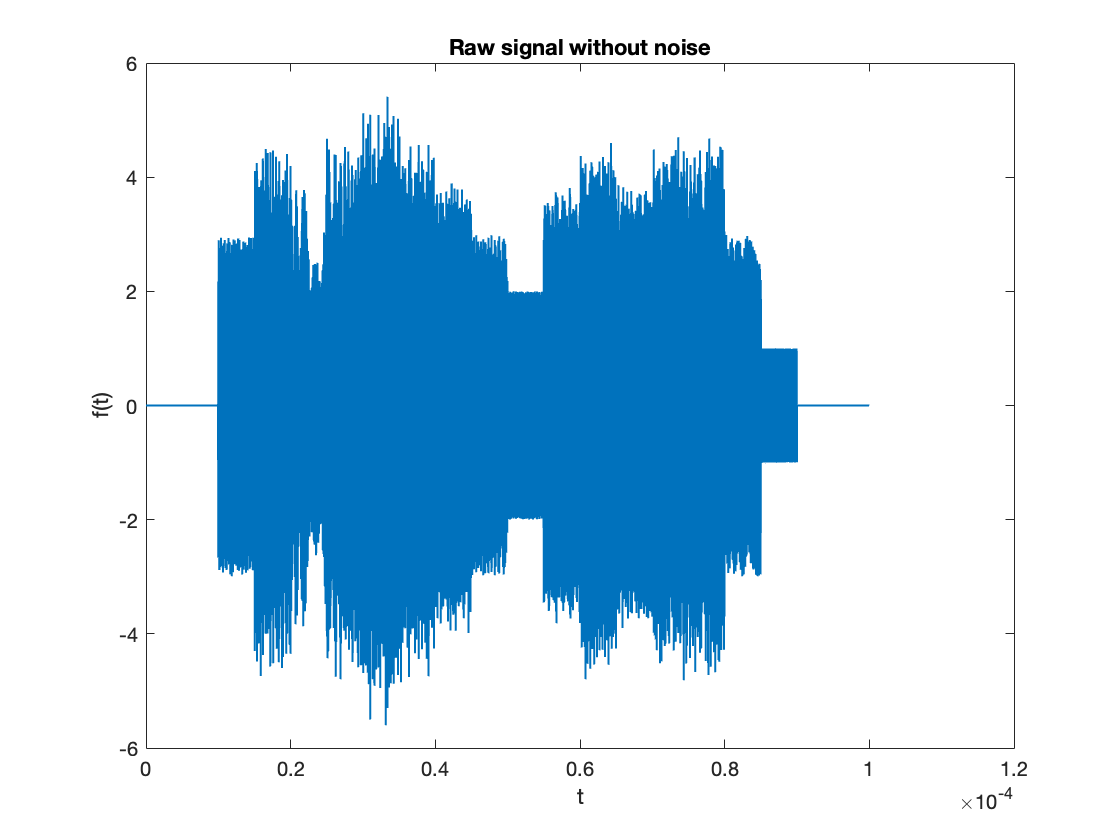}
\includegraphics[scale=.135]{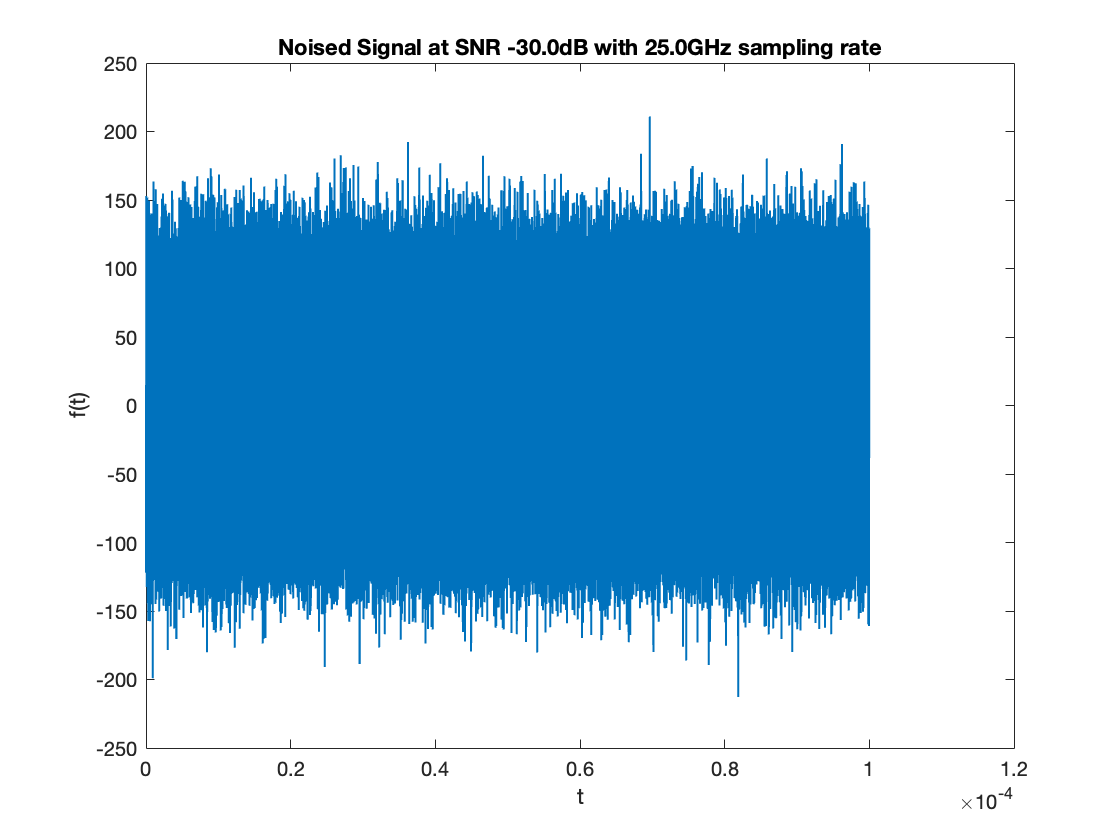}
\includegraphics[scale=.135]{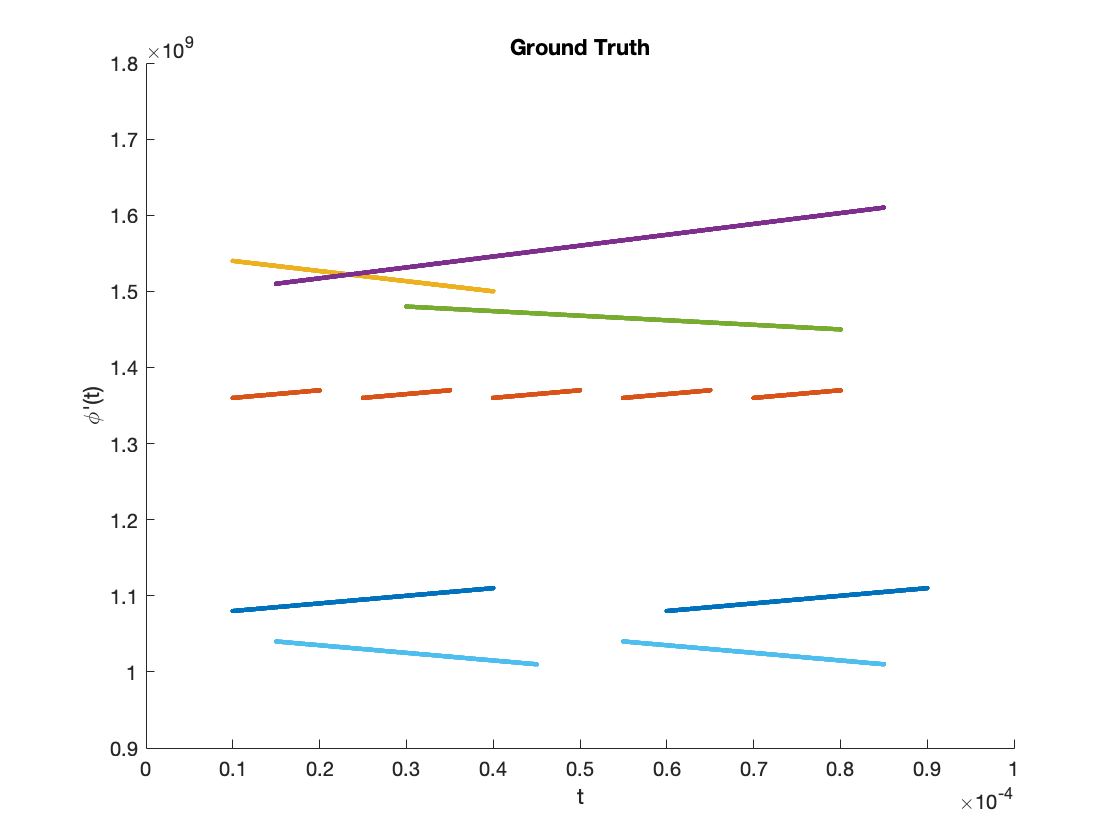}
\end{center}
\caption{(Left) The raw signal $f(t)$ as given in \eqref{eq:nonstationary} at sampling rate $0.5$ GHz. (Middle) The noised signal $F(t)=f(t)+\epsilon(t)$ at -10 dB SNR. (Right) The ground truth of $\phi'(t)$ from \eqref{eq:pulsephasedef}.}
\label{fig:signal_ground_truth}
\end{figure}

\begin{figure}[H]
\begin{center}
\includegraphics[width=.5\textwidth]{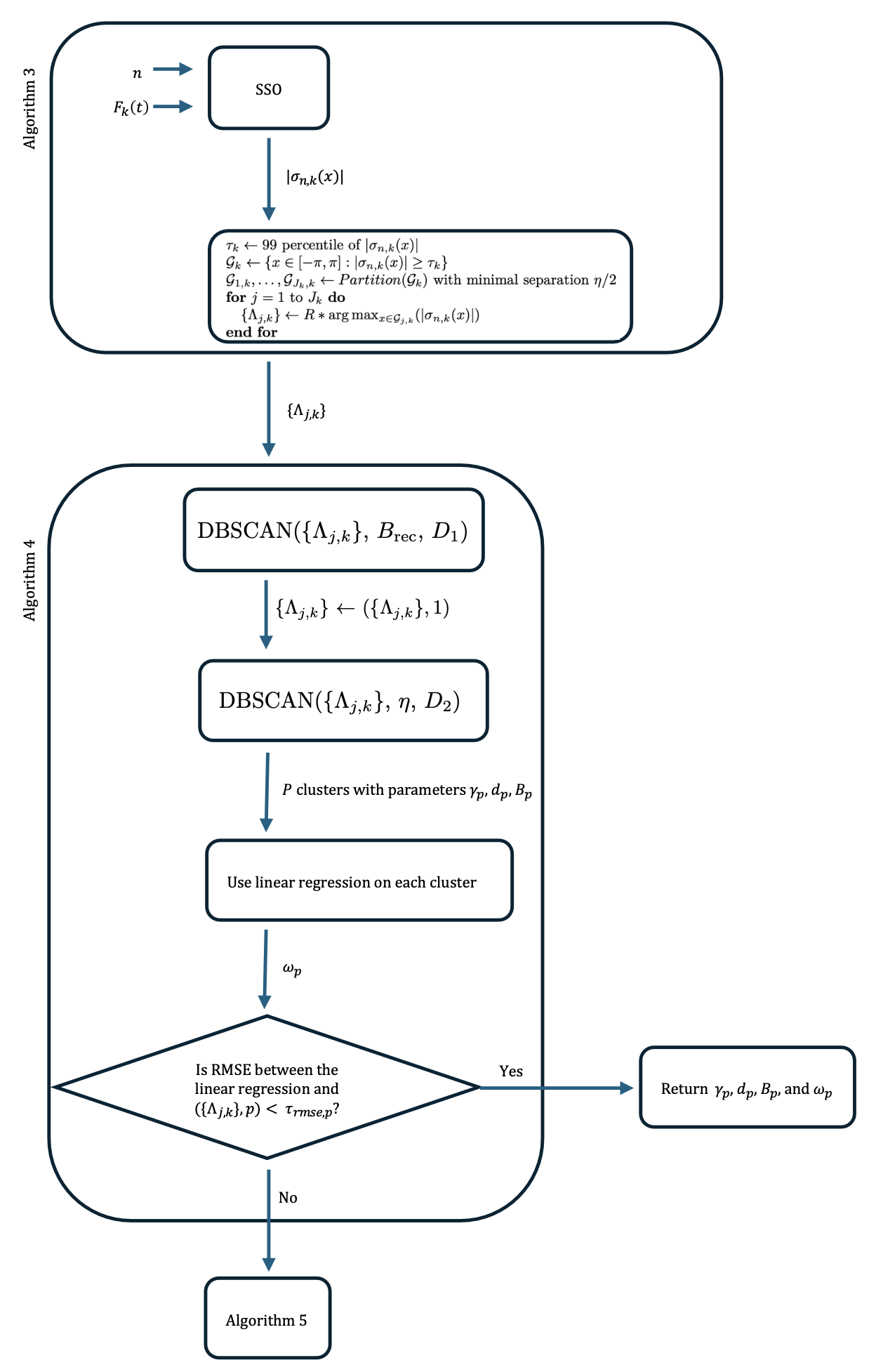}
\end{center}
\caption{The flowchart diagram for our Algorithms \ref{alg:univariate_navy} and \ref{alg:estimation_navy}.}
\label{fig:flowchart1}
\end{figure}

\begin{figure}[H]
\begin{center}
\includegraphics[width=.5\textwidth]{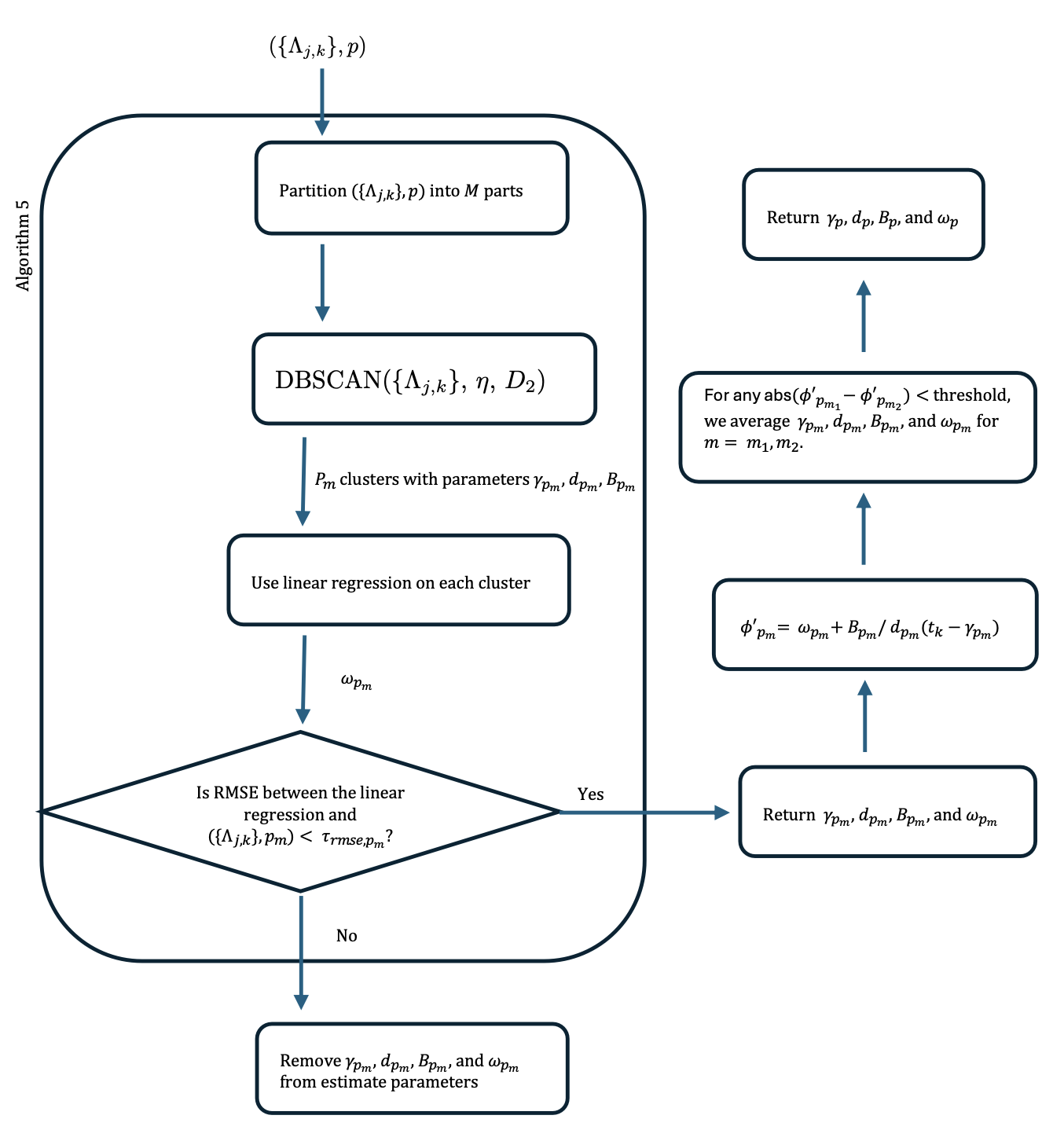}
\end{center}
\caption{The flowchart diagram for our Algorithm \ref{alg:crossover}.}
\label{fig:flowchart2}
\end{figure}

\begin{figure}[H]
\begin{center}
\includegraphics[scale=.2]{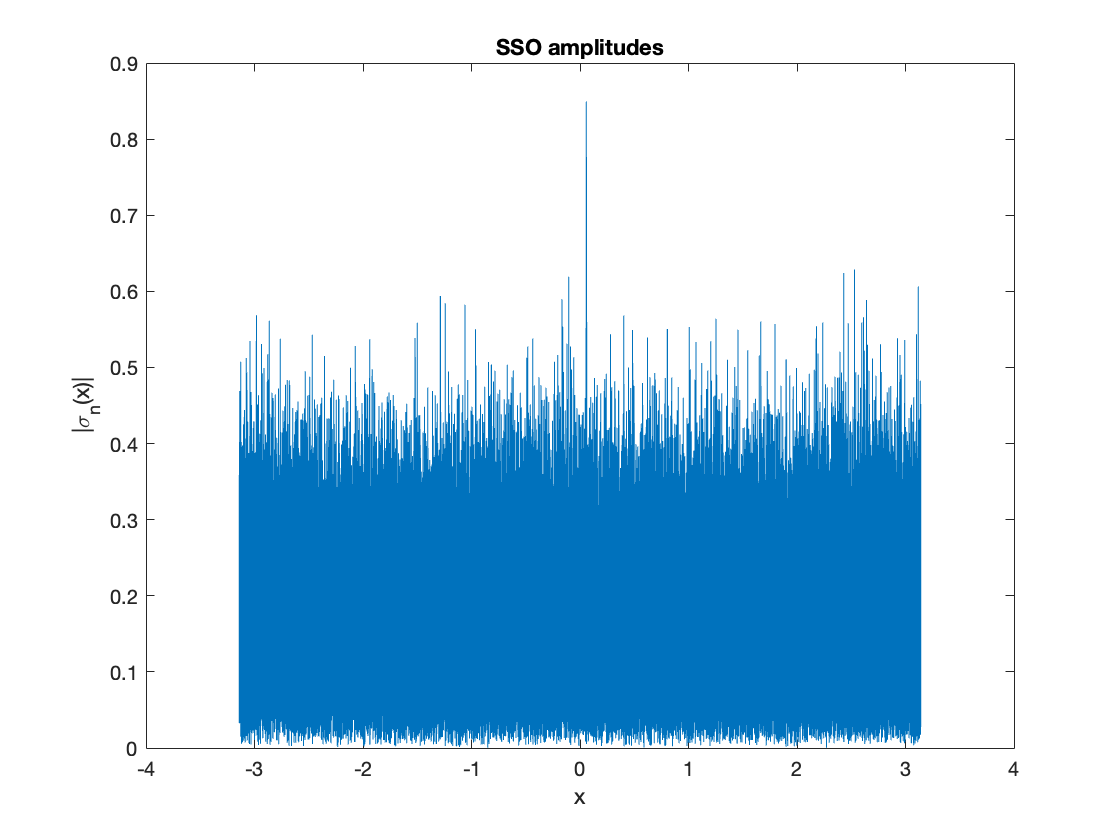}
\includegraphics[scale=.2]{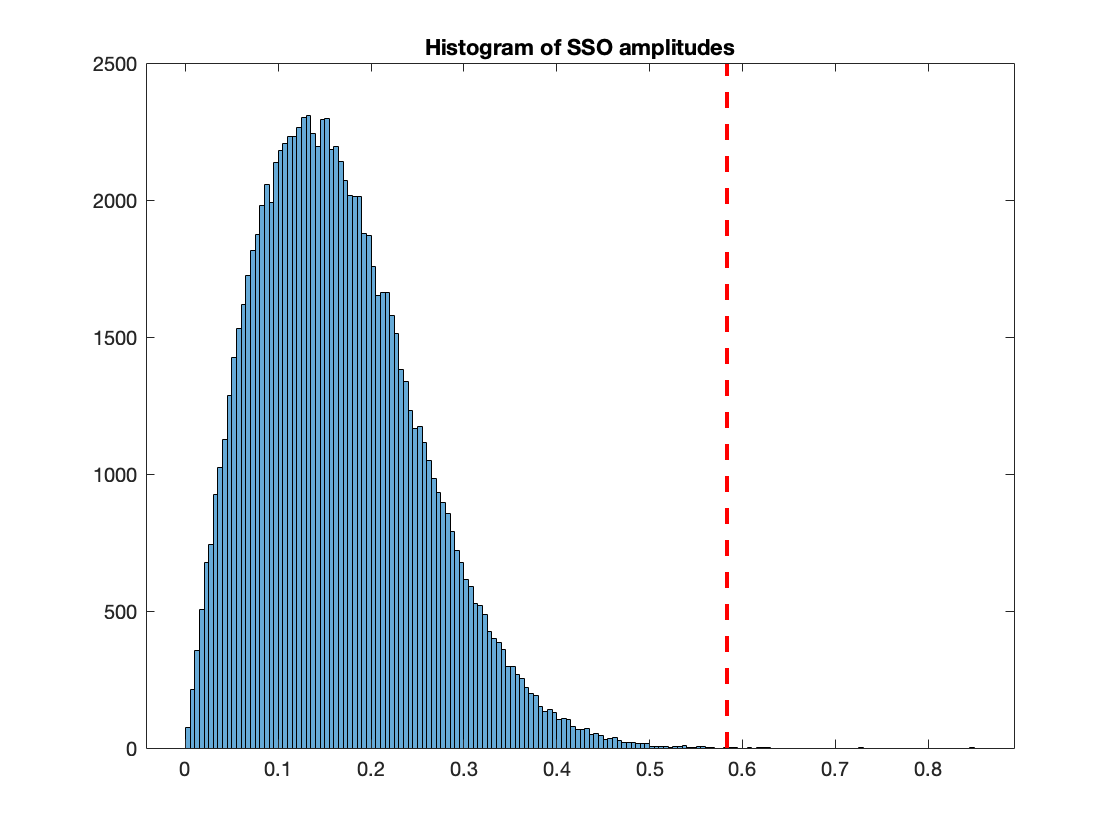}
\end{center}
\caption{(Left) The raw SSO result of $F_k(t)$ at $t=5 \times 10^{-5}$ with -10 dB SNR on 0.5 GHz sampling rate. (Right) The histogram of the SSO. The red dot line indicates the threshold $\tau_k = $ 99 percentile of the SSO result.}
\label{fig:noise_sso}
\end{figure}

\vspace*{-0.5cm}
\subsection{Finding instantaneous frequencies}\label{bhag:alg1}

In this section, we describe the first stage of the algorithm illustrated in Figure \ref{fig:flowchart1}, where we obtain the piecewise constant approximation to the the instantaneous frequencies in each snippet $F_k(t)$  for $k=1,\ldots,D$ (cf. \eqref{eq:ground_truth_reformed}) using the SSO for each such snippet.
The main challenge here is to determine the threshold in the definition \eqref{eq:levelsetbis}.
Since the value of the minimum amplitude $\mathfrak{m}$ is unknown, we run SSO on $F_k(t)$ and set the threshold $\tau_k$ to be $99$ percentile of the power spectrum (see Figure \ref{fig:noise_sso}). 
The percentile is the same for every interval $I_k$, but the actual value of the threshold $\tau_k$ will be different depending on the sampling frequency and the SNR.
 \begin{algorithm}[H]
 \begin{algorithmic}[1]
 \item[{\rm a)}] \textbf{Input:} Minimal separation $\eta$ and signal $F_k(t)$.
 \item[{\rm b)}] \textbf{Output:} $\{\Lambda_{j,k}\}$
 \STATE $\hbar_n \gets \left\{\sum_{|\ell|<n}H\left(\frac{|\ell|}{n}\right)\right\}^{-1}$
 \STATE $\sigma_{n,k}(x) \gets \hbar_n \sum_{|\ell|<n} H\left(\frac{|\ell|}{n}\right)F_k(t_{k}-\ell/R)e^{i\ell x}$
 \STATE $\tau_k \gets 99 \mbox{ percentile of } |\sigma_{n,k}(x)|$
 \STATE $\mathcal{G}_k \gets \{x\in [-\pi,\pi] : |\sigma_{n,k}(x)|\ge \tau_k\}$
 \STATE $\mathcal{G}_{1,k}, \ldots, \mathcal{G}_{J_k,k} \gets Partition(\mathcal{G}_k) \mbox{ with minimal separation } \eta/2$
  \FOR{$j=1$ to $J_k$}
 \STATE $\{\Lambda_{j,k}\} \gets R * \arg\max_{x\in \mathcal{G}_{j,k}} (|\sigma_{n,k}(x)|)$
 \ENDFOR
 \STATE \textbf{Return: } $\{\Lambda_{j,k}\}$ for $j=1,\ldots,J_k$
 \caption{Signal separation operator (SSO)}
 \label{alg:univariate_navy}
 \end{algorithmic}
 \end{algorithm}

\begin{figure}[H]
\begin{center}
\includegraphics[scale=.2]{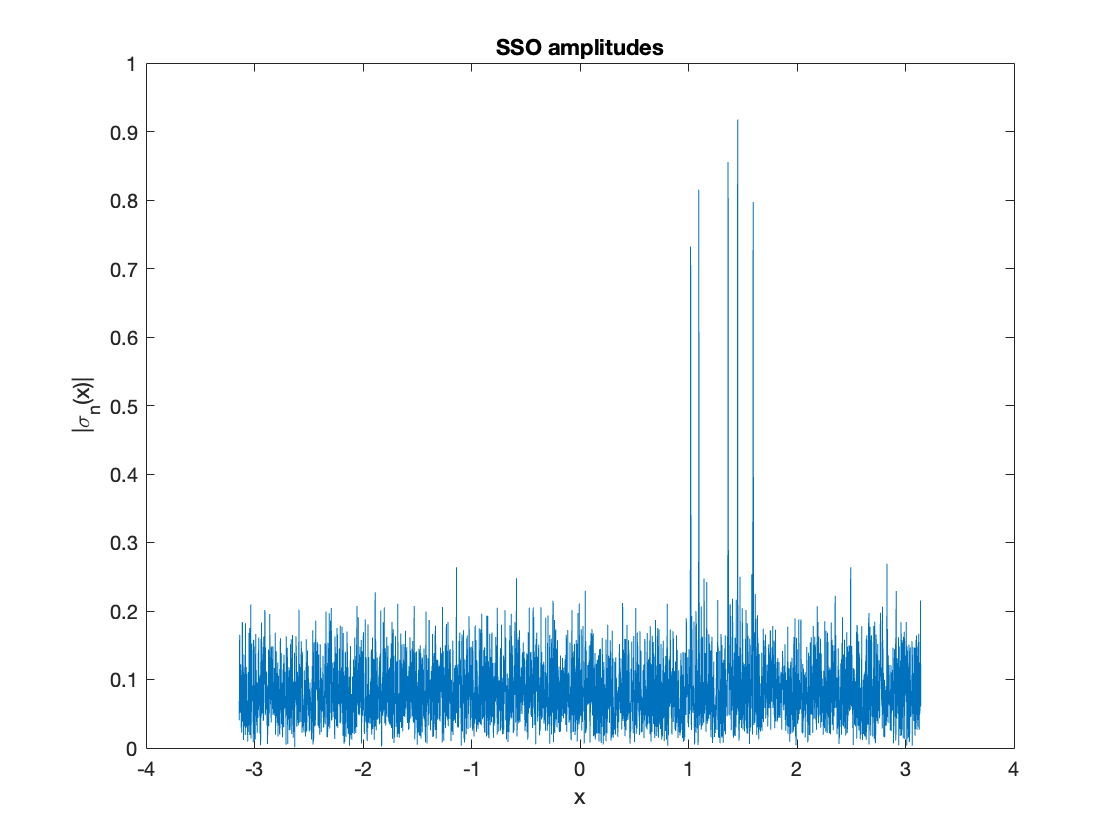}
\includegraphics[scale=.2]{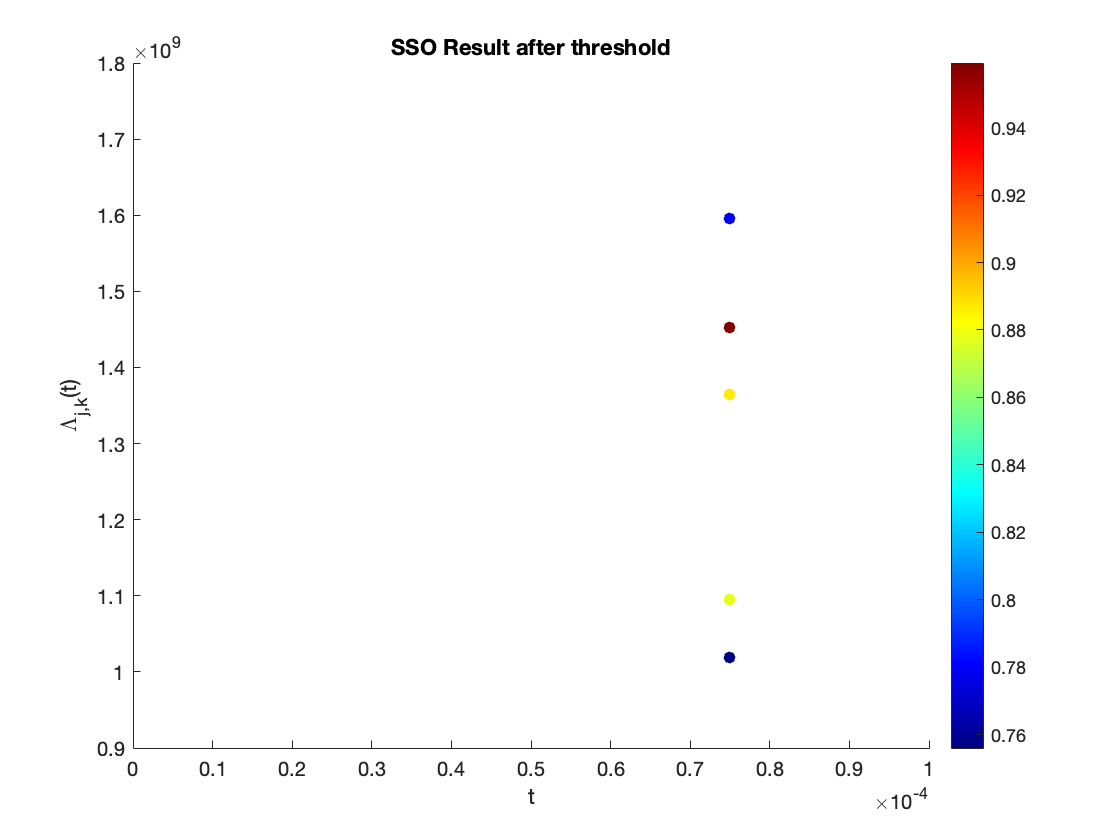}\\
\includegraphics[scale=.2]{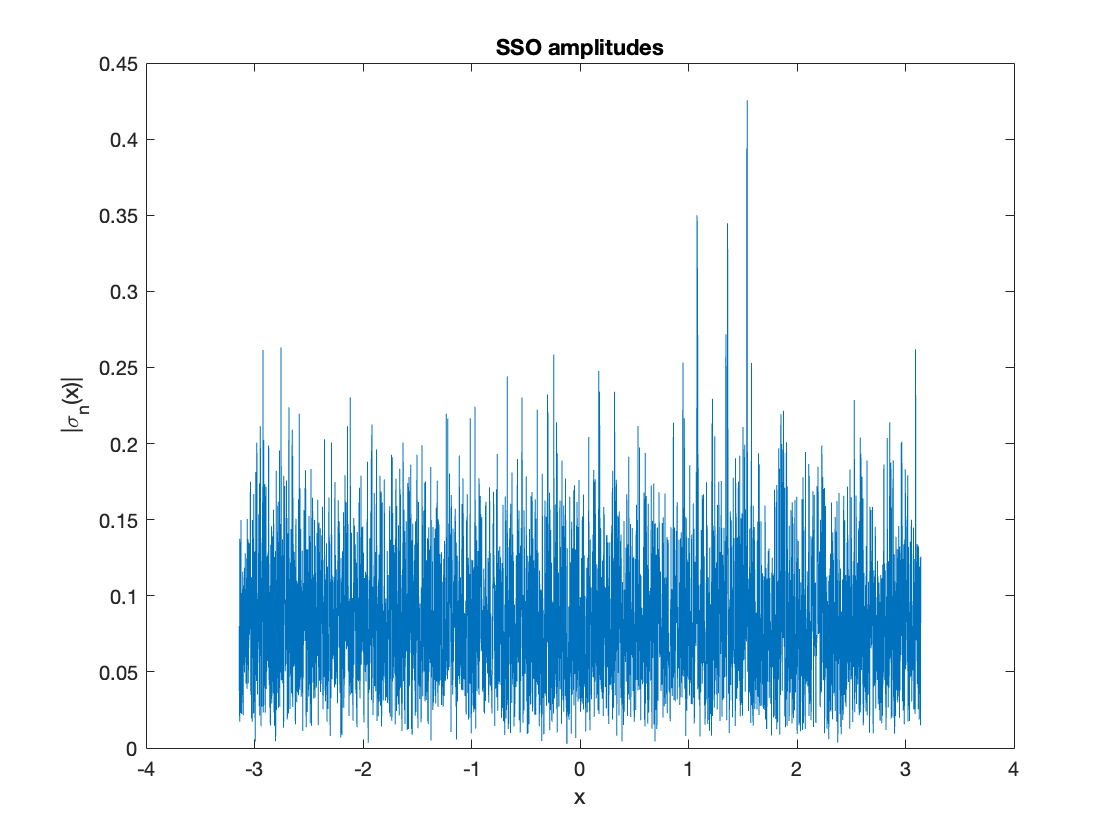}
\includegraphics[scale=.2]{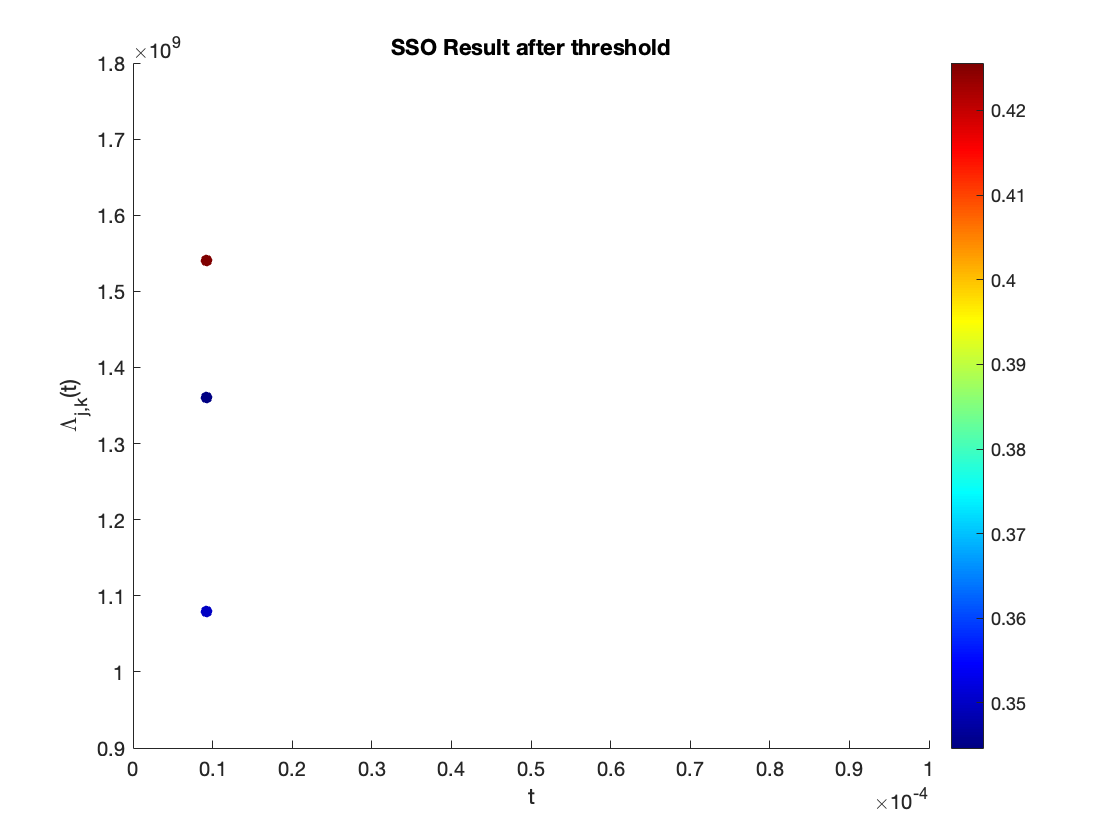}\\
\includegraphics[scale=.2]{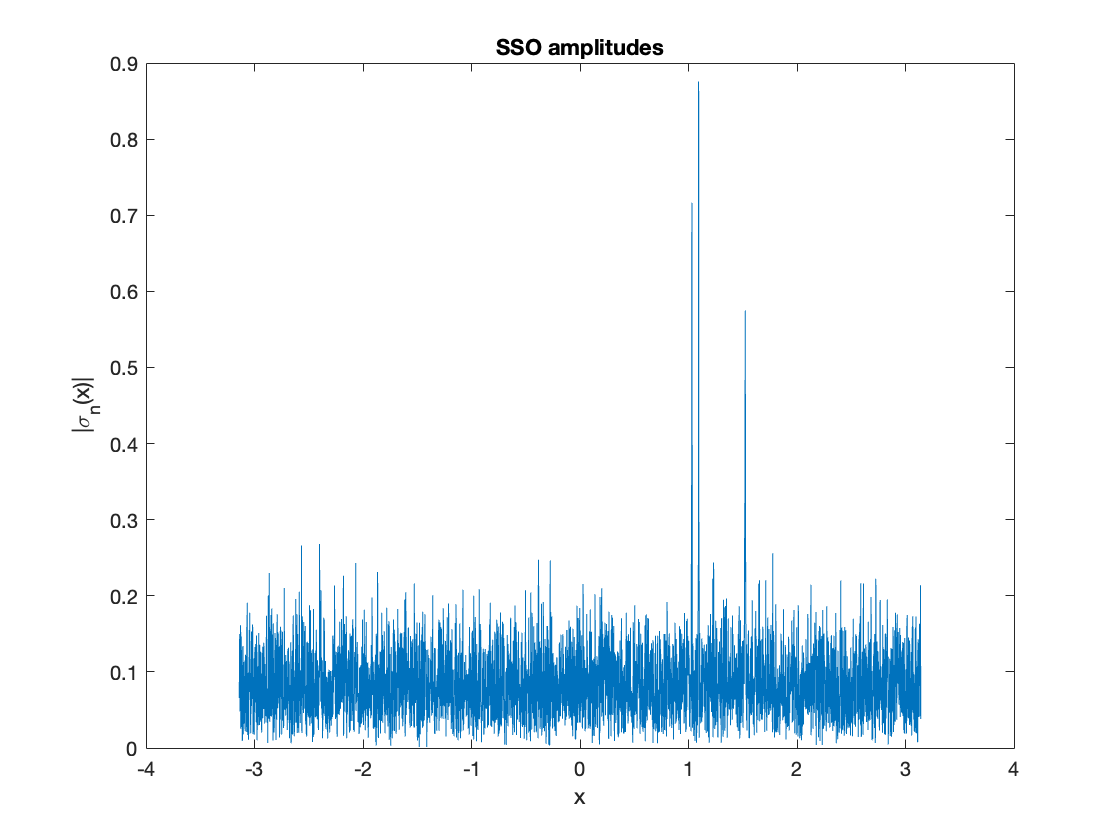}
\includegraphics[scale=.2]{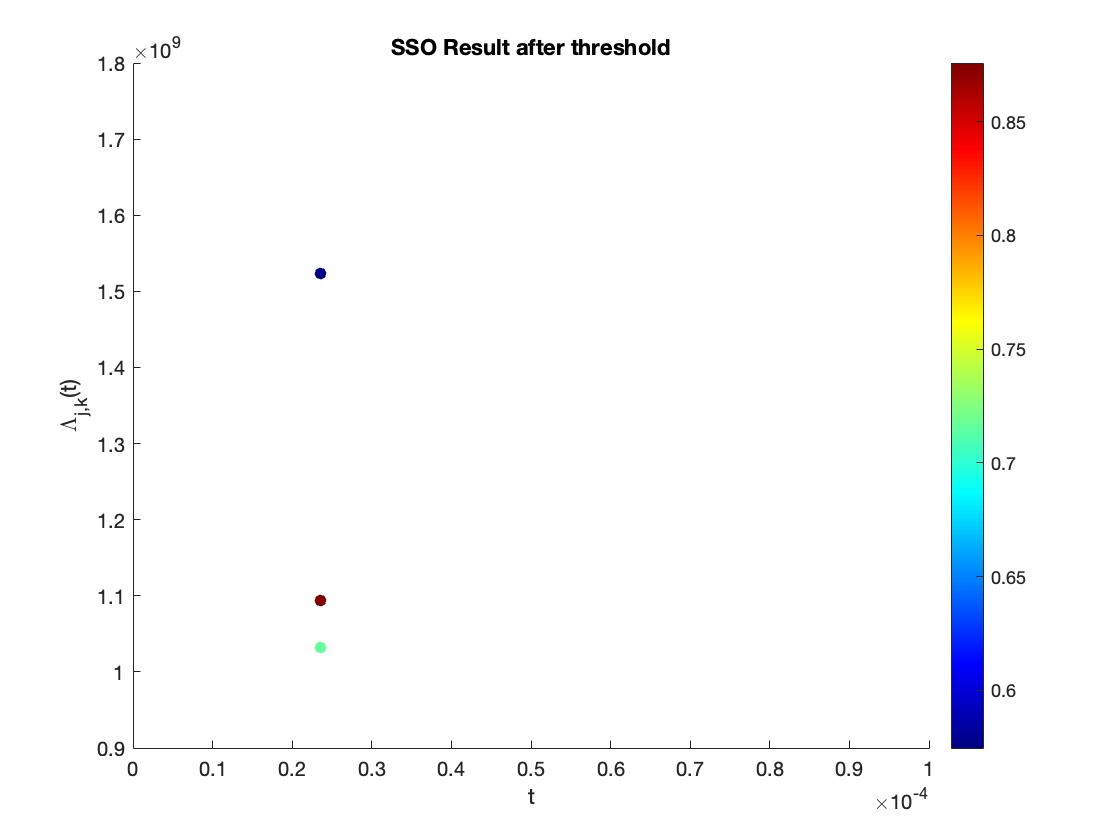}
\end{center}

\caption{(Left) $|\sigma_{n,k}(x)|$ for the interval $I_k$. (Right) SSO results at $t_k$.
(Top) The plots for $t_k =7.5 \times 10^{-5}$ where the signals pass from the beginning through the end of the interval $I_k$. (Middle) The plots for $t_k =1 \times 10^{-5}$ where the signal starts within the interval $I_k$. (Bottom) The plots for $t_k =2.35 \times 10^{-5}$ where there are 2 signals crossover within the interval $I_k$. }
\label{fig:freqattk}
\end{figure}

The output of this step is the values of the instantaneous frequencies at $t_k$:
\be
\{\Lambda_{j,k}\}=\omega_{j,k}+a_{j,k}(t_k-\gamma_{j,k}), \qquad j=1,\cdots, J_k.
\ee
This is illustrated in Figure~\ref{fig:freqattk}.
We see that the frequencies are determined quite accurately when there is no discontinuity in the signals, and the minimal separation is large, but there is a problem when a signal starts/stops somewhere within $I_k$, resulting in a discontinuity or when there is a cross-over frequency, resulting in a small minimal separation.

\subsection{Parameter estimation}\label{bhag:paramest}
The starting point of this step is the results of all the instantaneous frequencies in all the snippets.  The resulting diagram is called  the raw SSO diagram as shown in Figure~\ref{fig:rawsso} (Left). 
Knowing the receiver bandwidth $B_{\mbox{\scriptsize{rec}}}$, Algorithm~\ref{alg:estimation_navy} then concentrates on the relevant frequency part of the SSO diagram, which shows a piecewise constant approximation to the frequencies in the signal (see Figure \ref{fig:rawsso} (Right)).

 \begin{algorithm}[ht]
\begin{algorithmic}[1]
\item[{\rm a)}] \textbf{Input:} the receiver bandwidth $B_{\mbox{\scriptsize{rec}}}$, the minimum separation $\eta, \{\Lambda_{j,k}\}$ for $j=1,\ldots,J_k$, $k=1,\ldots,D$, and minimum number of neighbors $D_1, D_2$.
\item[{\rm b)}] \textbf{Output:} Estimation of $\omega_p, B_p, d_{p}$ and $\gamma_p$ for $p = 1, \ldots, P$.
\STATE $(\{\Lambda_{j,k}\},-1)$, $(\{\Lambda_{j,k}\},1)$ $\gets$ DBSCAN($\{\Lambda_{j,k}\}$, $B_{\mbox{\scriptsize{rec}}}$, $D_1$) to find the part of the $\{\Lambda_{j,k}\}$ relevant for our signal.
\STATE $\{\Lambda_{j,k}\}$ $\gets$ $(\{\Lambda_{j,k}\},1)$
\STATE $(\{\Lambda_{j,k}\}, -1)$, $(\{\Lambda_{j,k}\}, 1)$, $\ldots$, $(\{\Lambda_{j,k}\}, P)$ $\gets$ DBSCAN($\{\Lambda_{j,k}\}$, $\eta$, $D_2$) to separate the $\{\Lambda_{j,k}\}$ into $P$ clusters.
\FOR{$p=1$ to $P$}
\STATE $\Gamma_p \gets [(t_k, y_{j,k}) \mbox{ for } y_{j,k} \in (\{\Lambda_{j,k}\}, p)]$
\STATE $\gamma_p \gets$ $\min(t_k)$ for $(t_k, y_{j,k}) \in \Gamma_p$.
\STATE $d_p \gets$ $\max(t_k) - \min(t_k)$ for $(t_k, y_{j,k}) \in \Gamma_p$
\STATE $B_p \gets$ $(\max(y_{j,k}) - \min(y_{j,k}))/2$ for $(t_k, y_{j,k}) \in \Gamma_p$
\STATE Use linear regression on 50\% of $\Gamma_p$ that has highest $|\sigma_{n,k}(x)|$ to obtain the parameter estimation for $\omega_p$.
\STATE $\tau_{\mbox{\scriptsize{rmse}},p}$ $\gets$ $1\%$ of $\mbox{mean}(y_{j,k})$ for $(t_k, y_{j,k}) \in \Gamma_p$.
\IF{$\mbox{RMSE}_p$ defined in \eqref{eq:resresidue} $>$ $\tau_{\mbox{\scriptsize{rmse}},p}$}
\STATE Perform Algorithm \ref{alg:crossover}
\ENDIF
\ENDFOR
\STATE \textbf{Return: } $\omega_p, B_p, d_p$ and $\gamma_p$ for $j = 1, \ldots, P$.
 \caption{Parameter estimation}
 \label{alg:estimation_navy}
 \end{algorithmic}
 \end{algorithm}

Given the minimum separation, we can separate these approximations for different signal components using DBSCAN.
Algorithm~\ref{alg:estimation_navy} then focuses on this approximation in each component of the signal, and finds the corresponding parameters using linear regression.
The result of this operation is shown in Figure~\ref{fig:ssofound}.

\begin{figure}[H]
\begin{center}
\includegraphics[scale=.2]{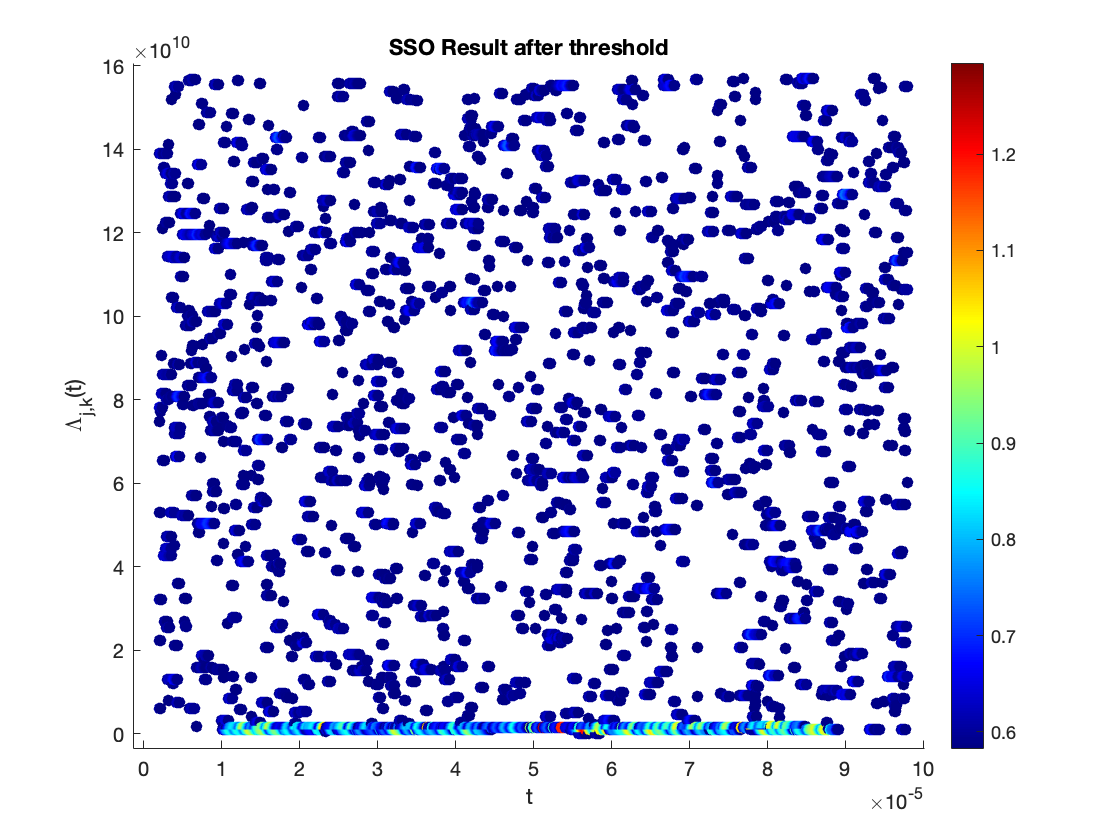}
\includegraphics[scale=.2]{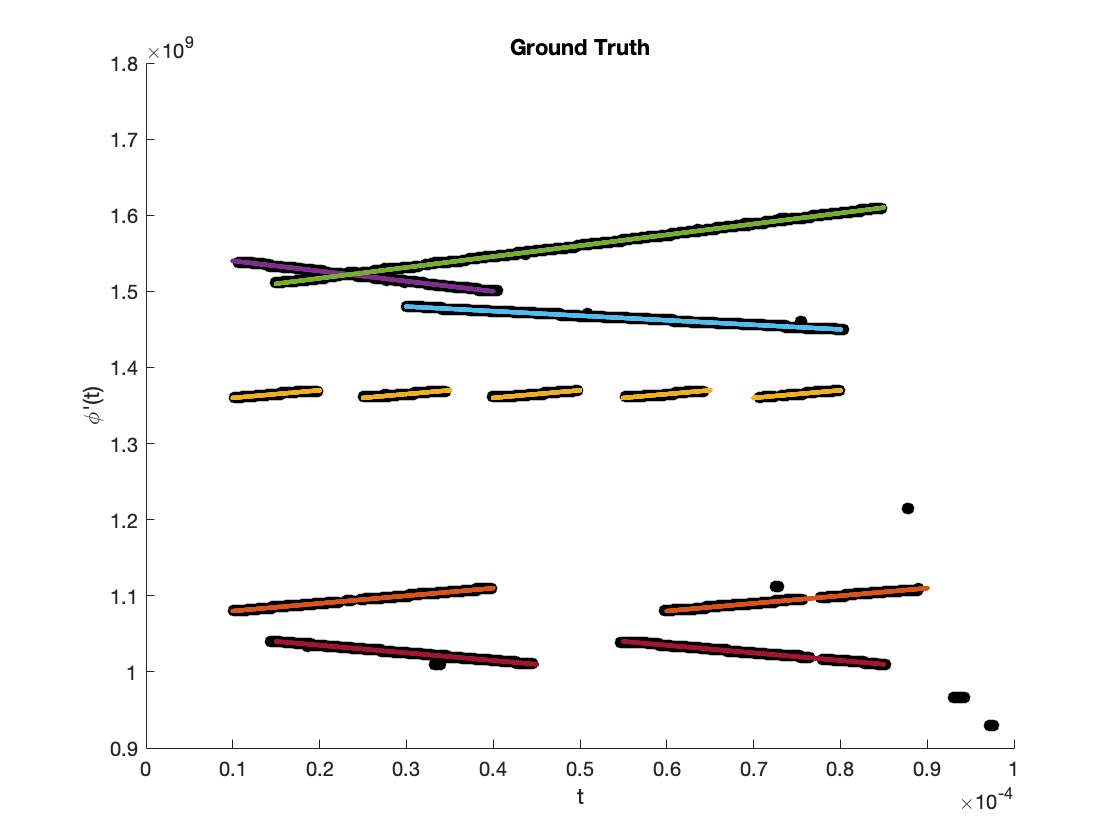}
\end{center}

\caption{(Left) The plot of threshold SSO result $\{\Lambda_{j,k}\}$ for $j=1,\ldots,J_k$, $k=1,\ldots,D$ by choosing $\Delta = 2\times 10^{-6}, D=2500, D_1=D/2, D_2=D/100,$ and $t_k$ are equidistant samples from 0 to $1\times 10^{-4}$. (Right) The plot of threshold SSO result after step 1 of Algorithm \ref{alg:estimation_navy}.}
\label{fig:rawsso}
\end{figure}

\subsection{Refinements}\label{bhag:refine}

Algorithm~\ref{alg:crossover}  focuses on working on the clusters in the SSO diagram where the signals crossover; i.e. the distance between two or more signals is less than the minimum separation $\eta$.

\begin{figure}[h]
\begin{center}
\includegraphics[scale=.2]{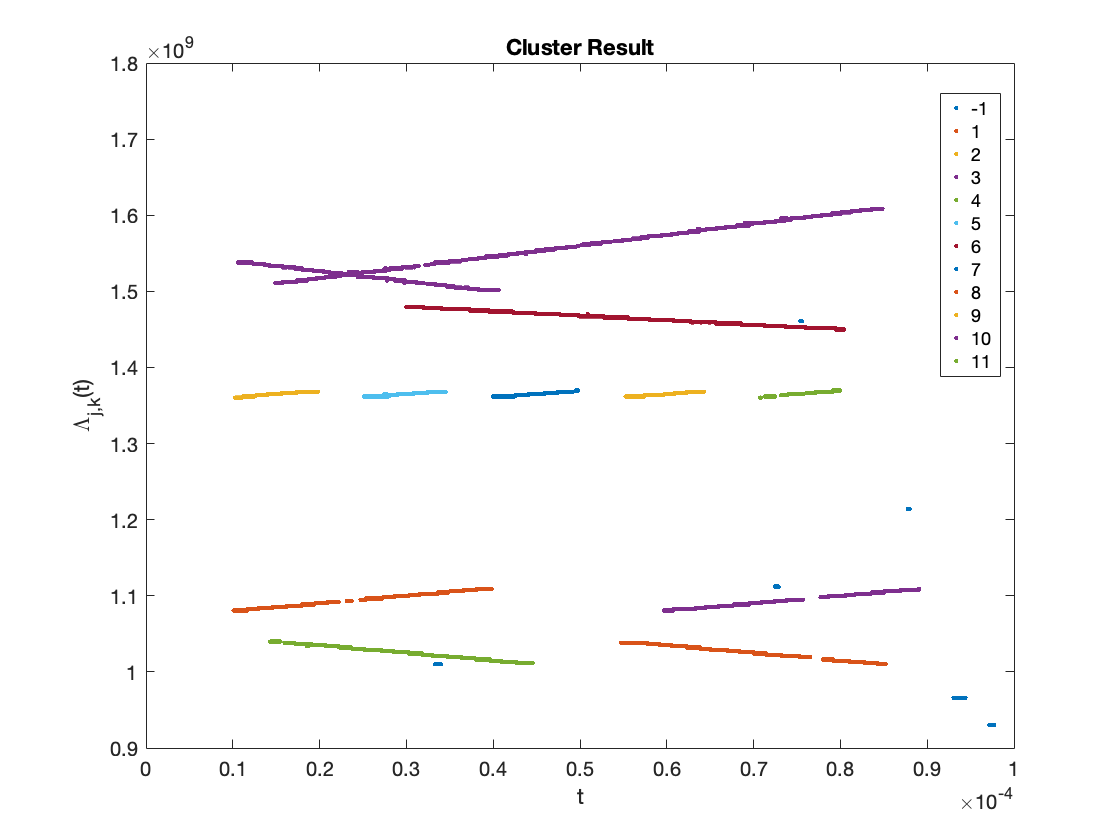}
\includegraphics[scale=.2]{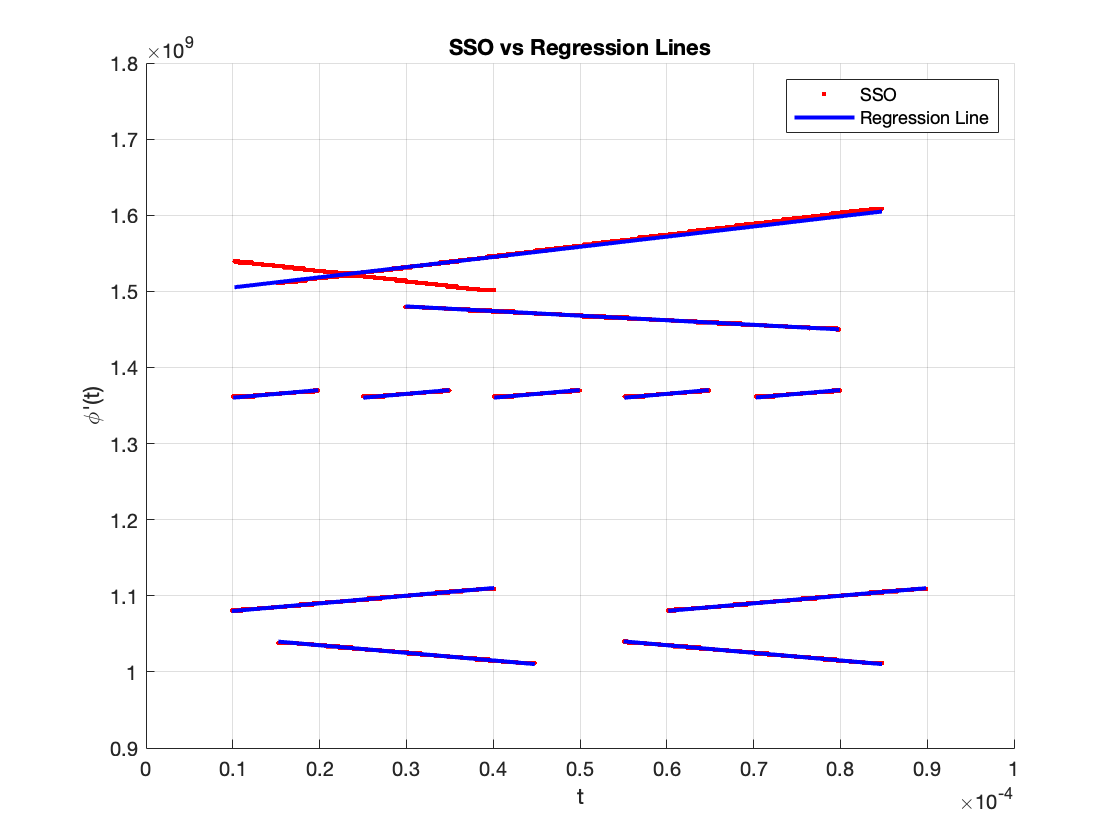}
\end{center}
\caption{(Left) The plot of signal components after using DBSCAN. (Right) The comparison plot between the SSO result vs the regression lines of $((\{\Lambda_{j,k}\},1), p)$. This is the result of Algorithm \ref{alg:estimation_navy} without performing Algorithm \ref{alg:crossover}; i.e. skipping steps 8 - 10 of Algorithm \ref{alg:estimation_navy}.}
\label{fig:ssofound}
\end{figure}

\begin{algorithm}[ht]
\begin{algorithmic}[1]
\item[{\rm a)}] \textbf{Input:} the minimum separation $\eta, (\{\Lambda_{j,k}\}, p)$ for $p$ when signals crossover, and minimum number of neighbors $D_2$, partition number $M$.
\item[{\rm b)}] \textbf{Output:} Estimation of $\omega_p, B_p, d_p$ and $\gamma_p$ for $p$ when signals crossover.
\STATE $\{\Lambda_{j,k}\}$ $\gets$ $(\{\Lambda_{j,k}\},p)$
\STATE partitions $\gets$ partition $\{\Lambda_{j,k}\}$ into $M$ parts.
\FOR{partition $m$ in  partitions}
\STATE $(\{\Lambda_{j,k}\}, -1)$, $(\{\Lambda_{j,k}\}, 1)$, $\ldots$, $(\{\Lambda_{j,k}\}, P_m)$ $\gets$ DBSCAN($\{\Lambda_{j,k}\}$, $\eta$, $D_2$) to separate the $\{\Lambda_{j,k}\}$ into $P_m$ clusters.
\FOR{$p_m=1$ to $P_m$}
\STATE $\Gamma_{p_m} \gets [(t_k, y_{j,k}) \mbox{ for } y_{j,k} \in (\{\Lambda_{j,k}\}, {p_m})]$
\STATE $\gamma_{p_m} \gets$ $\min(t_k)$ for $(t_k, y_{j,k}) \in \Gamma_{p_m}$.
\STATE $d_{p_m} \gets$ $\max(t_k) - \min(t_k)$ for $(t_k, y_{j,k}) \in \Gamma_{p_m}$
\STATE $B_{p_m} \gets$ $(\max(y_{j,k}) - \min(y_{j,k}))/2$ for $(t_k, y_{j,k}) \in \Gamma_{p_m}$
\STATE Use linear regression on 50\% of $\Gamma_{p_m}$ that has highest $|\sigma_{n,k}(x)|$ to obtain the parameter estimation for $\omega_{p_m}$.
\STATE $\tau_{\mbox{\scriptsize{rmse}},{p_m}}$ $\gets$ $1\%$ of $\mbox{mean}(y_{j,k})$ for $(t_k, y_{j,k}) \in \Gamma_{p_m}$.
\IF{$\mbox{RMSE}_{p_m}$ defined in \eqref{eq:resresidue} $>$ $\tau_{\mbox{\scriptsize{rmse}},{p_m}}$}
\STATE Remove $\omega_{p_m}, B_{p_m}, d_{p_m}$, and $\gamma_{p_m}$ from estimate parameters and return fail to detect this part of the signal.
\ENDIF
\ENDFOR
\ENDFOR
\STATE Compute $\phi'_{p_m}$ for each partition. If any two or more of the $\phi'_{p_m}$'s are within $10\%$ each other,  then average them and update these parameters to $\omega_{p}, B_{p}, d_{p}$ and $t_{p}$.
\STATE \textbf{Return: } $\omega_{p}, B_{p}, d_{p}$, and $\gamma_{p}$ for $p$ when signals crossover.


 \caption{Handling signals crossover}
 \label{alg:crossover}
 \end{algorithmic}
 \end{algorithm}

The part where signals crossover is detected by computing the RMSE between $y_{j,k}$ for $(t_k, y_{j,k}) \in \Gamma_{p}$ and the regression line using the formula
\bea
\mbox{residue}_{j,k,p} &=& |y_{j,k} - (\omega_{p}+B_{p}/d_{p}(t_k-\gamma_{p}))|, \label{eq:resresidue} \\
\mbox{RMSE}_p &=& \sqrt{\frac{1}{D}\sum_{k=1}^D\sum_{j=1}^{J_k} \mbox{residue}^2_{j,k,p}}. \label{eq:rmse}
\eea
If the $\mbox{RMSE}_p$ is larger that a small threshold, we then proceed to the refinement algorithm, Algorithm \ref{alg:crossover}.

To refine this part of the signal, we partition the signal into small intervals as shown in Figure \ref{fig:algo3}. Finally, we set a threshold on the $\phi_{p_m}'(t)$ and connect them if the differences between two or more $\phi_{p_m}'(t)$'s are within the threshold. The final result is shown in Figure \ref{fig:final}.

\begin{figure}[H]
\begin{center}
\includegraphics[scale=.2]{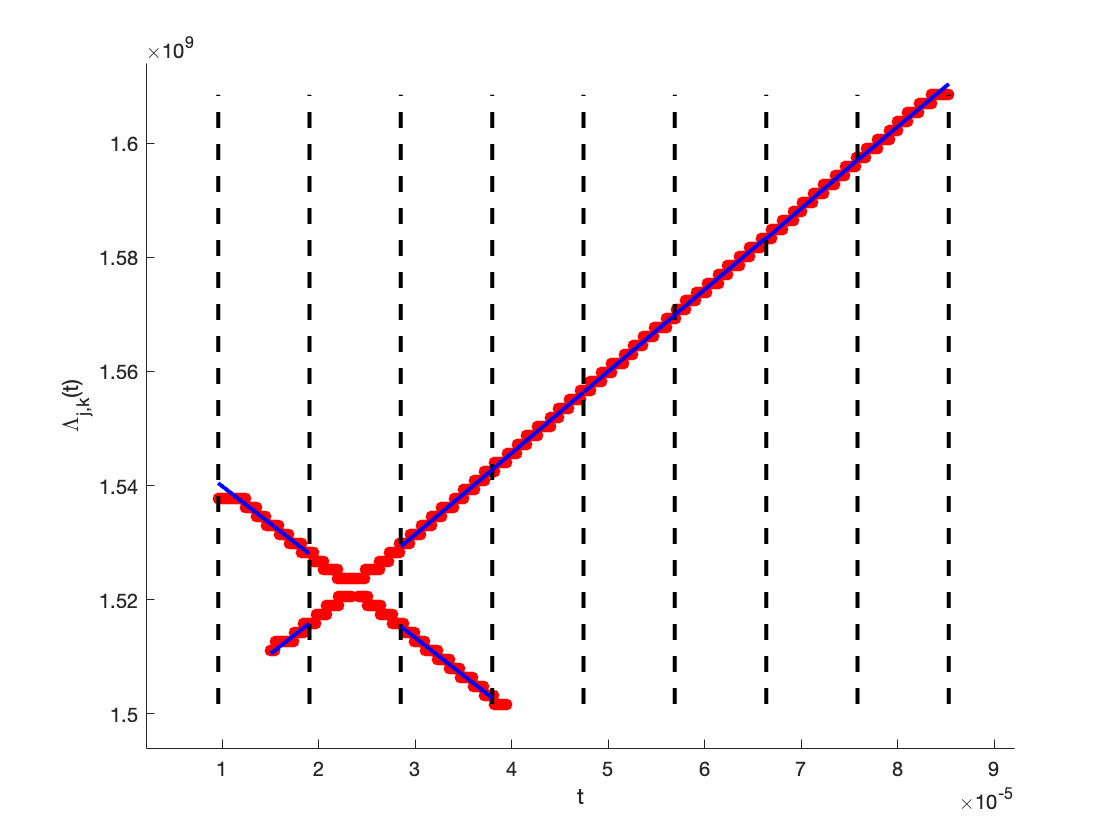}
\includegraphics[scale=.2]{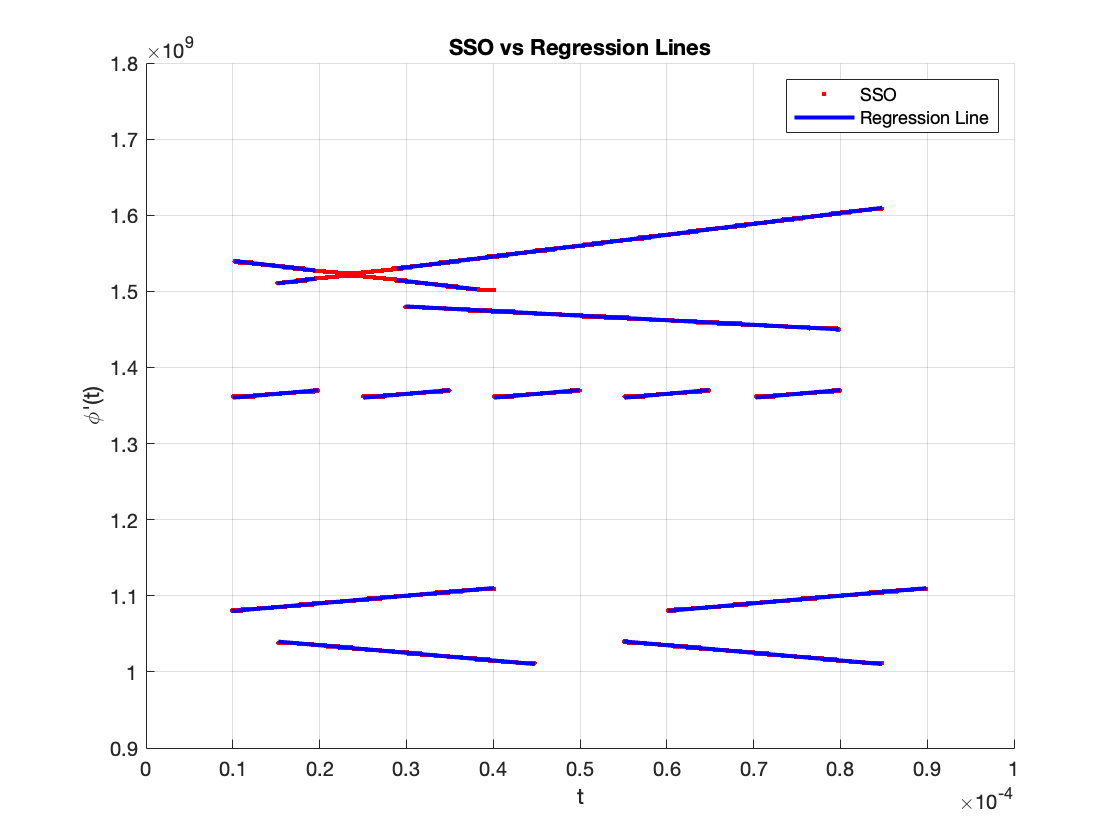}
\end{center}

\vspace*{-1cm}
\caption{(Left) The plot of crossing signals after performing step 1 - 13 of the Algorithm \ref{alg:crossover} by choosing $M = 8$ partitions. (Right) The comparison plot between the SSO result vs the regression lines of $\{\Lambda_{j,k}\}$ after performing step 1 - 13 of the Algorithm \ref{alg:crossover}.}
\label{fig:algo3}
\end{figure}

\begin{figure}[H]
\begin{center}
\includegraphics[scale=.2]{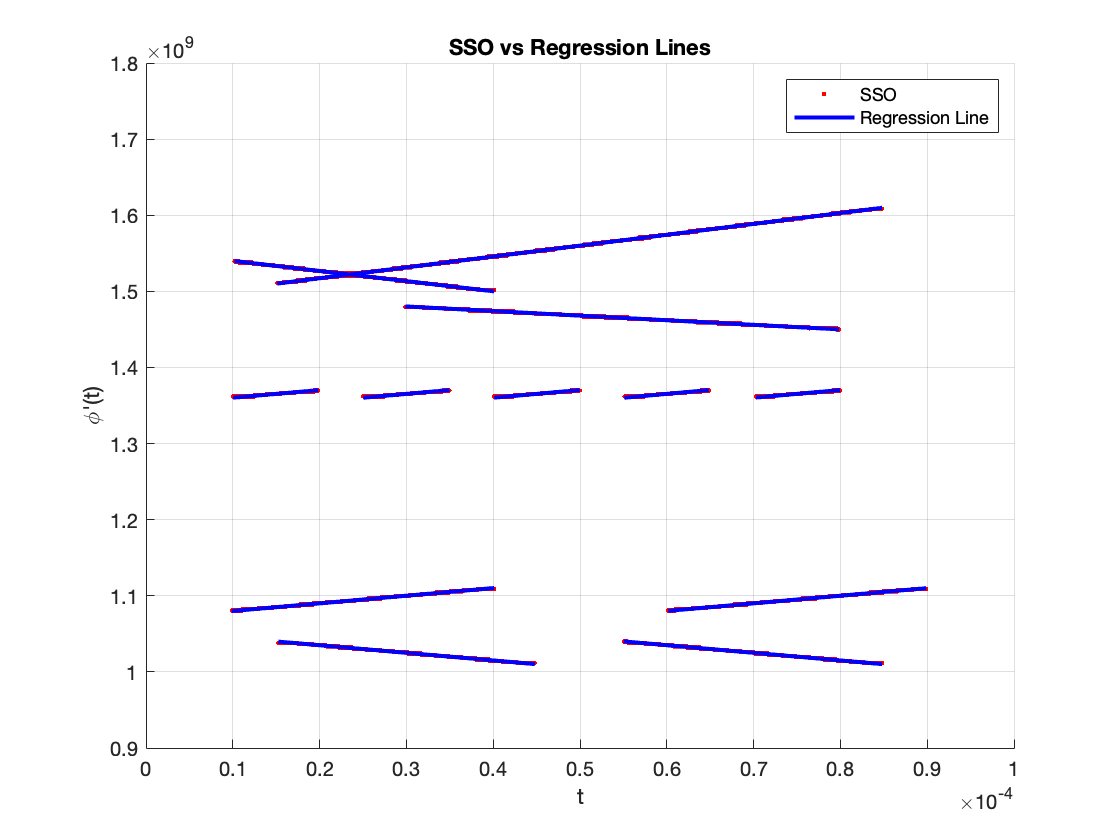}
\includegraphics[scale=.2]{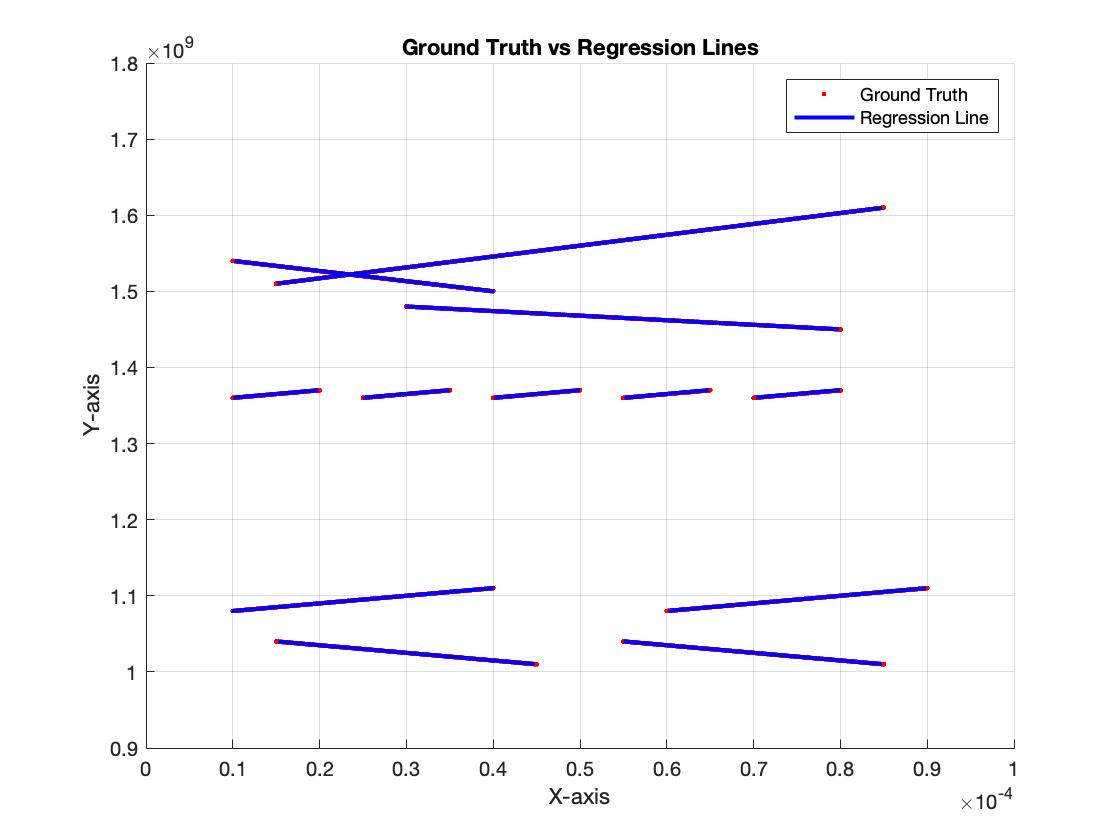}
\end{center}
\caption{The final comparison plot between the SSO result vs the regression lines of $\{\Lambda_{j,k}\}$ (left) and the ground truth vs the regression lines (right).}
\label{fig:final}
\end{figure}

\vspace*{-0.5cm}
\bhag{Numerical experiments}\label{bhag:numerical}

\subsection{Numerical results}\label{bhag:numerial_results}
To assess the robustness and performance of the proposed Signal Separation Operator (SSO) method, we conducted a series of controlled numerical experiments across a variety of scenarios. These experiments were designed to evaluate the impact of key signal characteristics, including minimal frequency separation, SNR, sampling rate, and the presence of frequency crossovers.

The details for all the experiments are listed in Tables~\ref{tab:result_table_1} and \ref{tab:result_table_2}.

\subsubsection{Minimal separation}\label{bhag:minsep}
The algorithm demonstrates strong performance when signal components are sufficiently separated in frequency. As illustrated in Figure \ref{fig:4results}, the method accurately recovers the instantaneous frequencies and associated parameters in both non-intersecting and cleanly intersecting cases. However, performance degrades in scenarios where the minimal separation between components becomes too small. In such cases, the spectral peaks associated with different components may become indistinct, leading to inaccurate clustering and parameter estimation. This limitation is primarily due to the inherent resolution bounds imposed by the kernel size and the sampling interval.

\begin{figure}[H]
\begin{center}
\includegraphics[scale=.20]{0_new_fig/1_result}
\includegraphics[scale=.20]{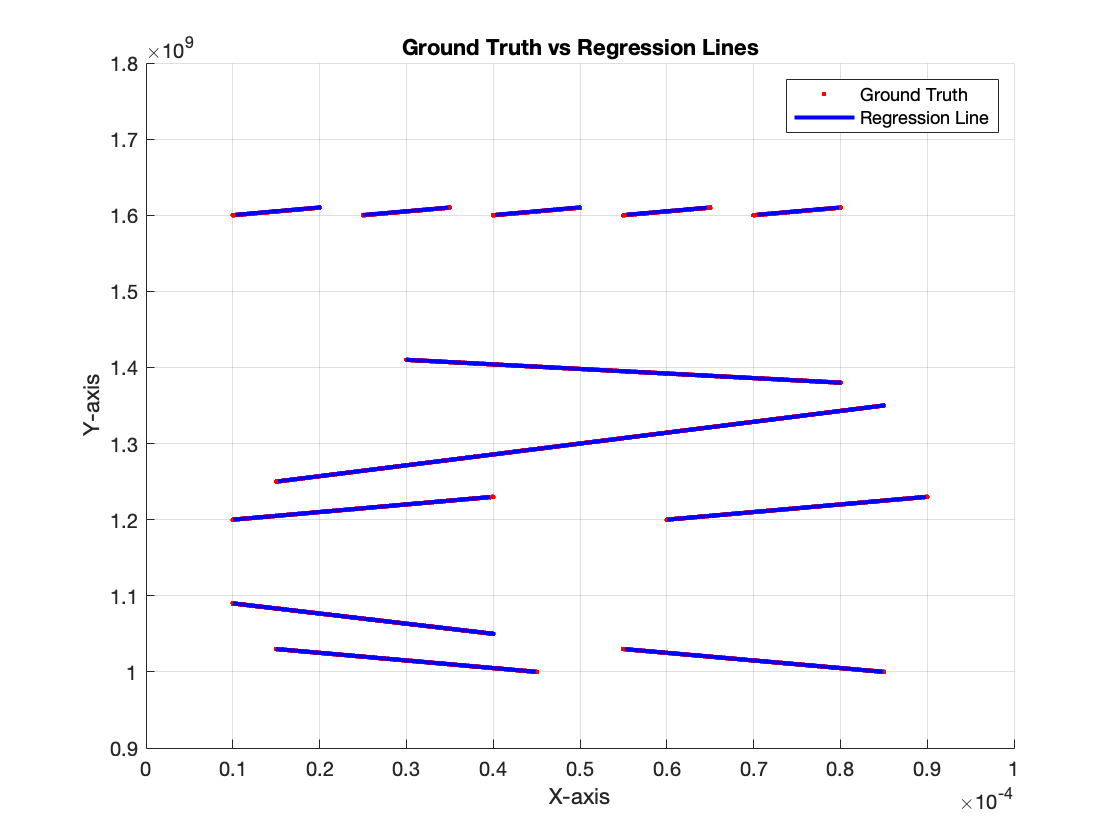}\\
\includegraphics[scale=.20]{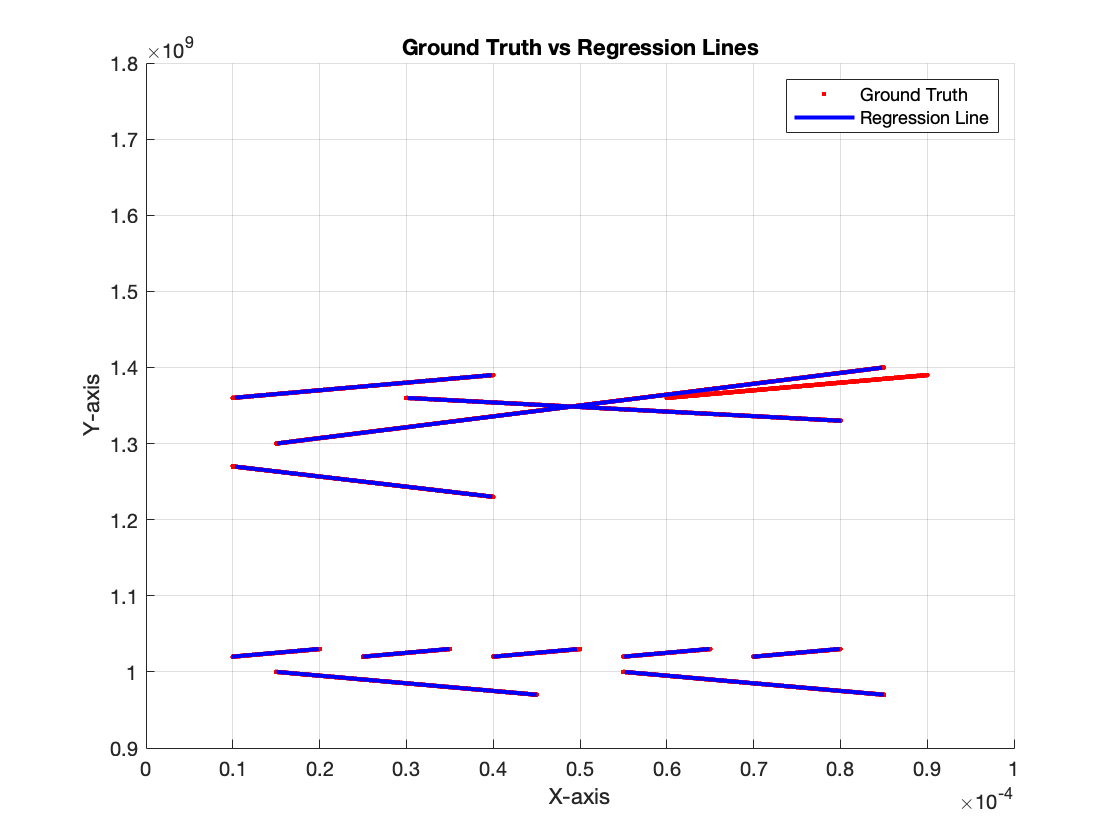}
\includegraphics[scale=.20]{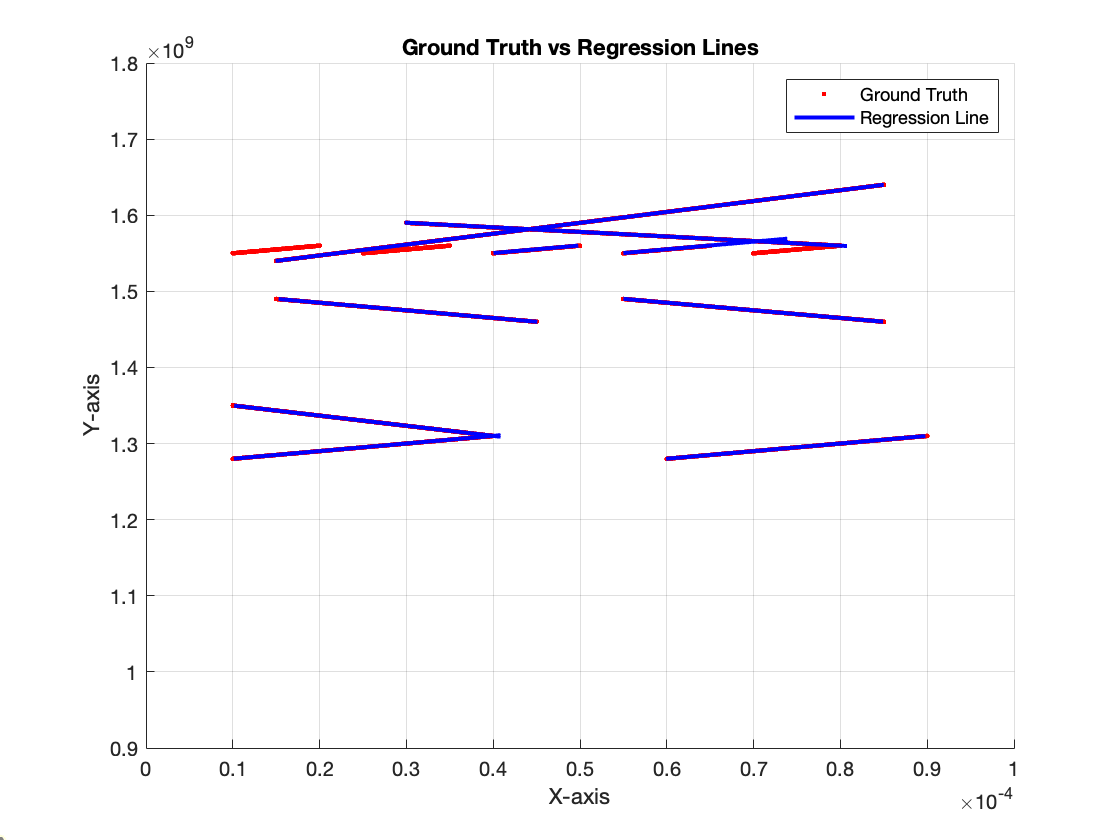}
\end{center}
\caption{Ground truth vs regression result plots for example 1 (top left), 2 (top right), 3 (bottom left), and 4 (bottom right) with SNR -10 dB and sampling rate 0.5 GHz. Our algorithm works well in case of no crossover and clear crossover signals, but does not work so well when the minimum separation among signals is low.}
\label{fig:4results}
\end{figure}

\vspace*{-0.5cm}
\subsubsection{Robustness to noise}\label{bhag:robust}
The method remains effective even under high noise conditions. 
To illustrate, we compute the the root mean square error (RMSE) in each experiment by
\be\label{eq:exptrmse}
\mbox{RMSE}=\mbox{mean}\left(\sqrt{\frac{1}{D}\sum_{k=1}^D\sum_{j=1}^{J_k}\left(\frac{\phi_{j,k}'(t_k)-\widehat{\phi_{j,k}'(t_k)}}{\phi_{j,k}'(t_k)}\right)^2}\right).
\ee
In our experiments, we took the mean over 16 trials for each choice of the signal, the sampling rate, and SNR.
As shown in Figure \ref{fig:4rmse}, the root mean square error (RMSE) remains low for SNR levels ranging from $10$ dB to $-30$ dB. 
The algorithm maintains RMSE values within acceptable bounds, provided that the sampling rate is sufficient. 
The combination of localized kernel averaging and peak detection in the frequency domain allows the method to suppress noise effectively.

\begin{figure}[H]
\begin{center}
\includegraphics[scale=.20]{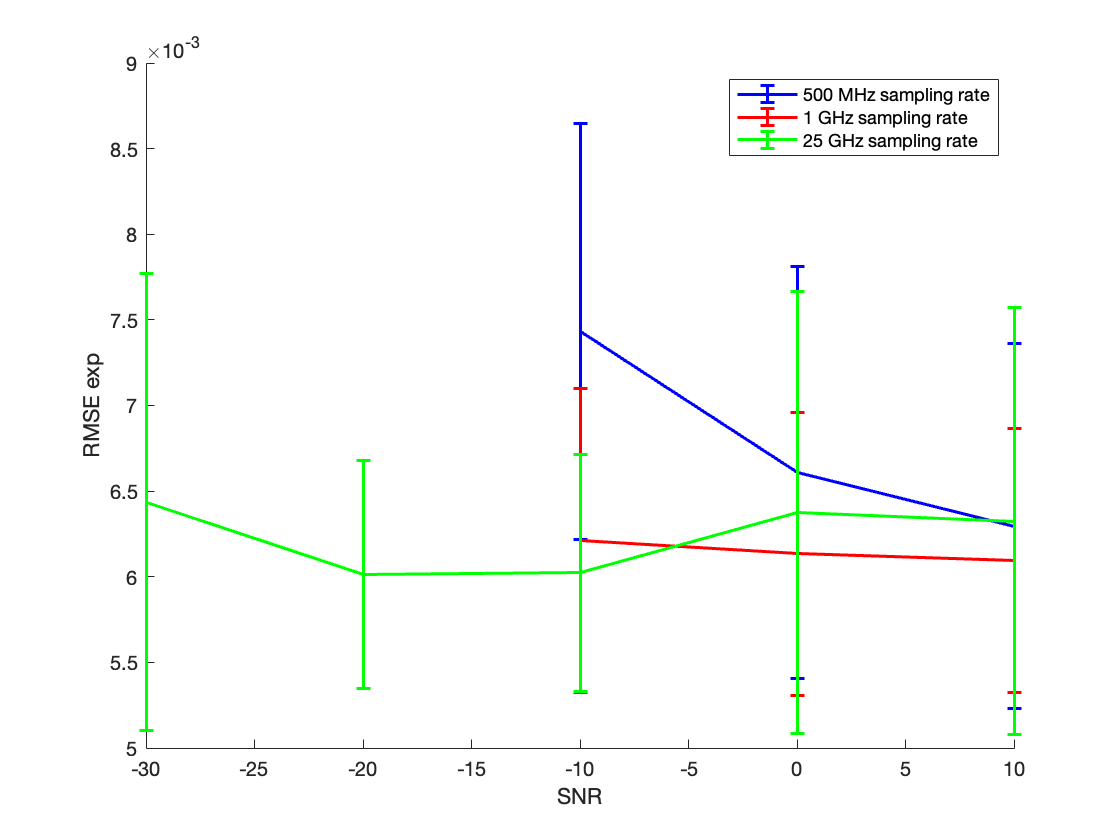}
\includegraphics[scale=.20]{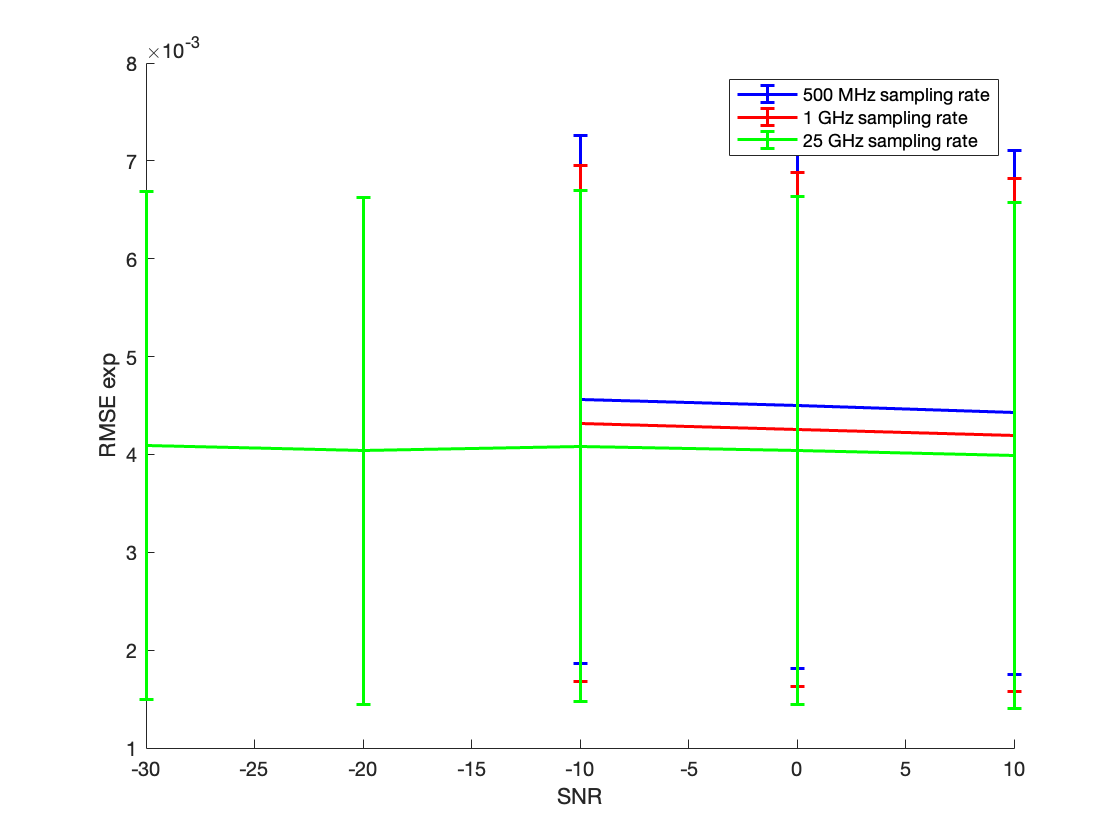}\\
\includegraphics[scale=.20]{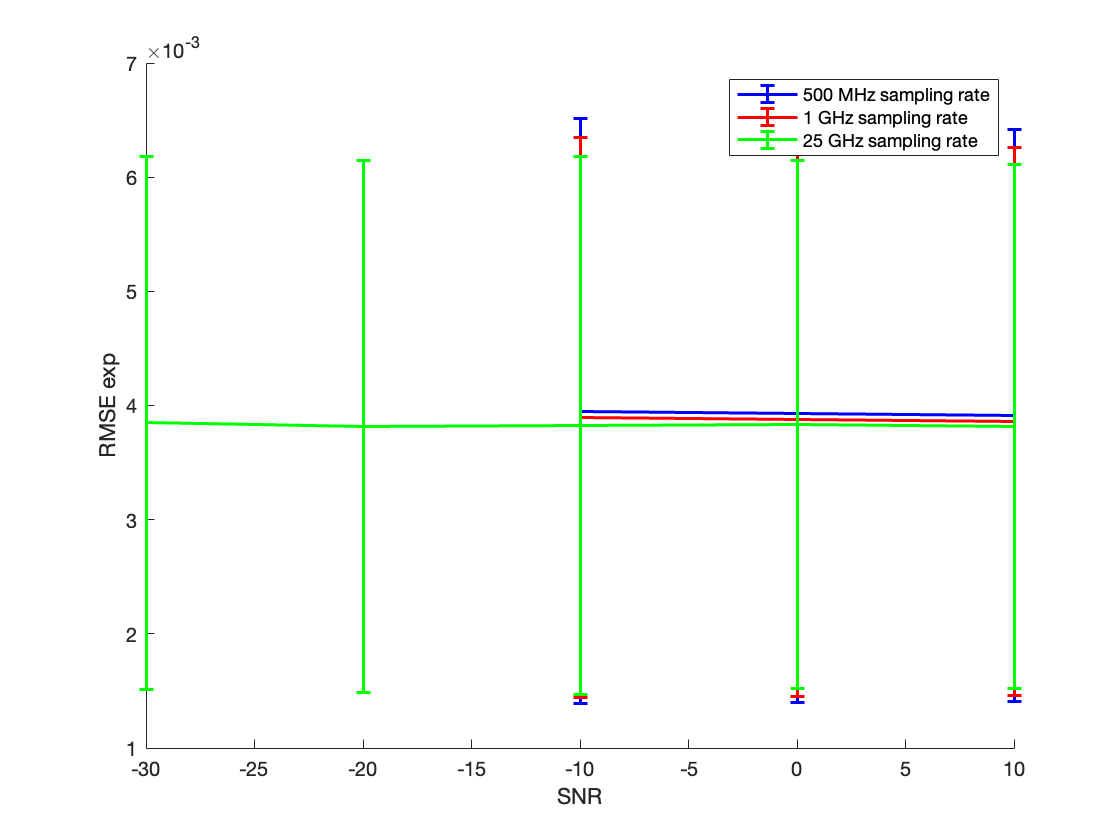}
\includegraphics[scale=.20]{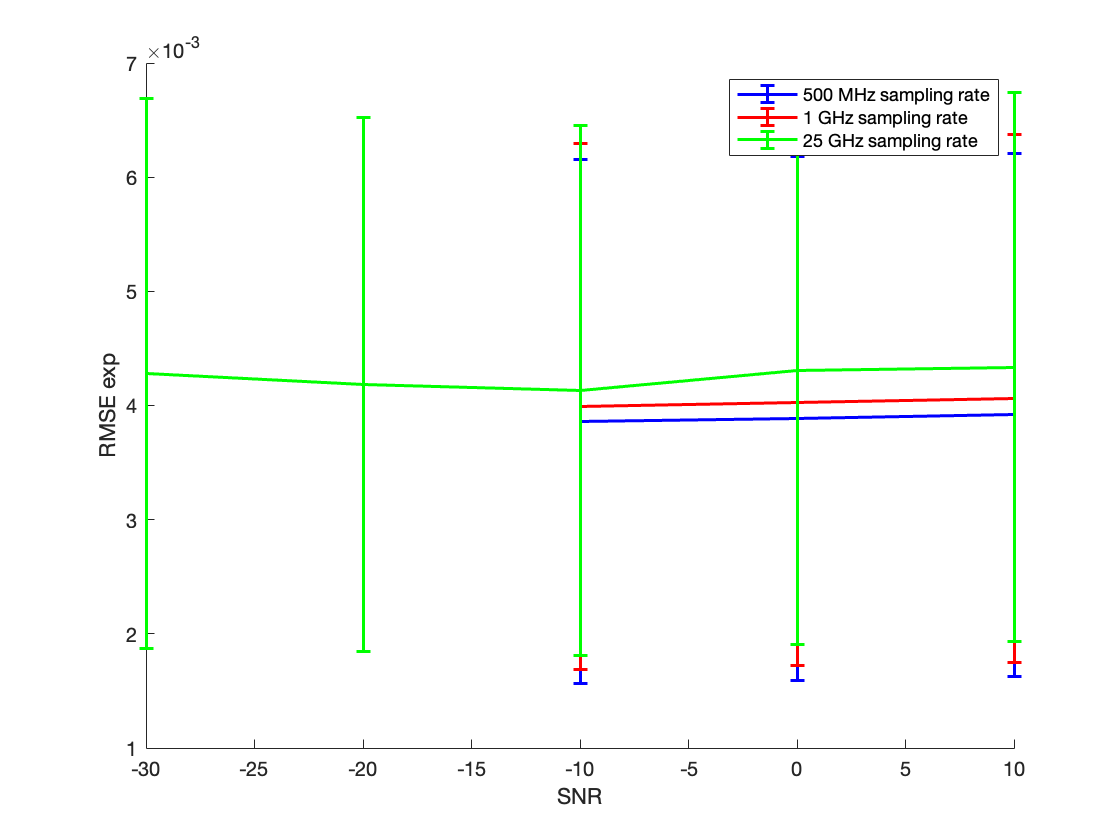}
\end{center}
\caption{The sampling rate performance plots of vary SNR vs RMSE for example 1 (top left), 2 (top right), 3 (bottom left), and 4 (bottom right). The RMSE was comparing between the ground truth and the estimation parameters of $\phi(t)$ in \eqref{eq:pulsephasedef}. }
\label{fig:4rmse}
\end{figure}

\vspace*{-0.5cm}
\subsubsection{Sampling Rate}\label{bhag:samplingrate}
Sampling rate plays a critical role in frequency resolution and noise resilience. As shown in Figure \ref{fig:4heatmap}, increasing the sampling rate from 0.5 GHz to 25 GHz significantly improves the accuracy of instantaneous frequency estimation, particularly in cases with minimal separation or closely spaced crossovers. In this figure, a heat map is used to visualize the estimation error computed in equation $\eqref{eq:resresidue}$, where the color indicates the magnitude of error. This color-coded representation helps to clearly highlight the improvement in performance with higher sampling rates. Higher sampling rates yield a finer time grid, which enhances the operator's ability to resolve subtle frequency variations and distinguish between overlapping components.

\begin{figure}[H]
\begin{center}
\includegraphics[scale=.14]{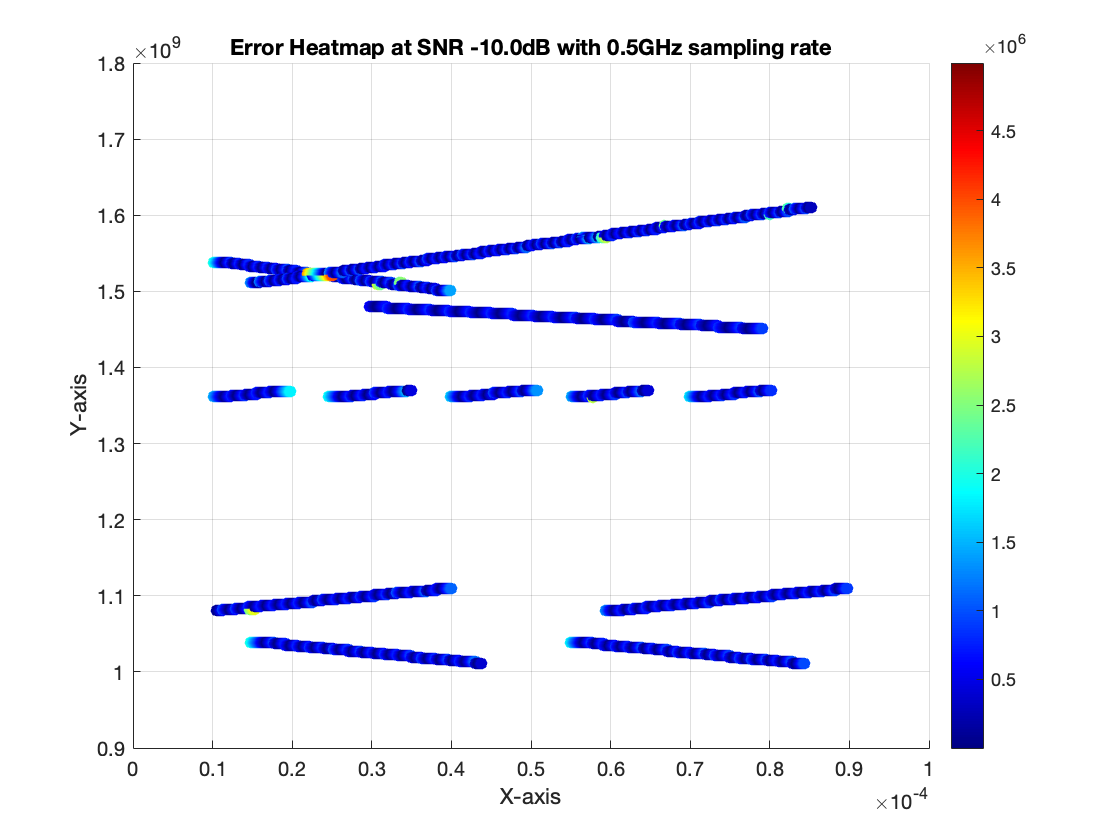}
\includegraphics[scale=.14]{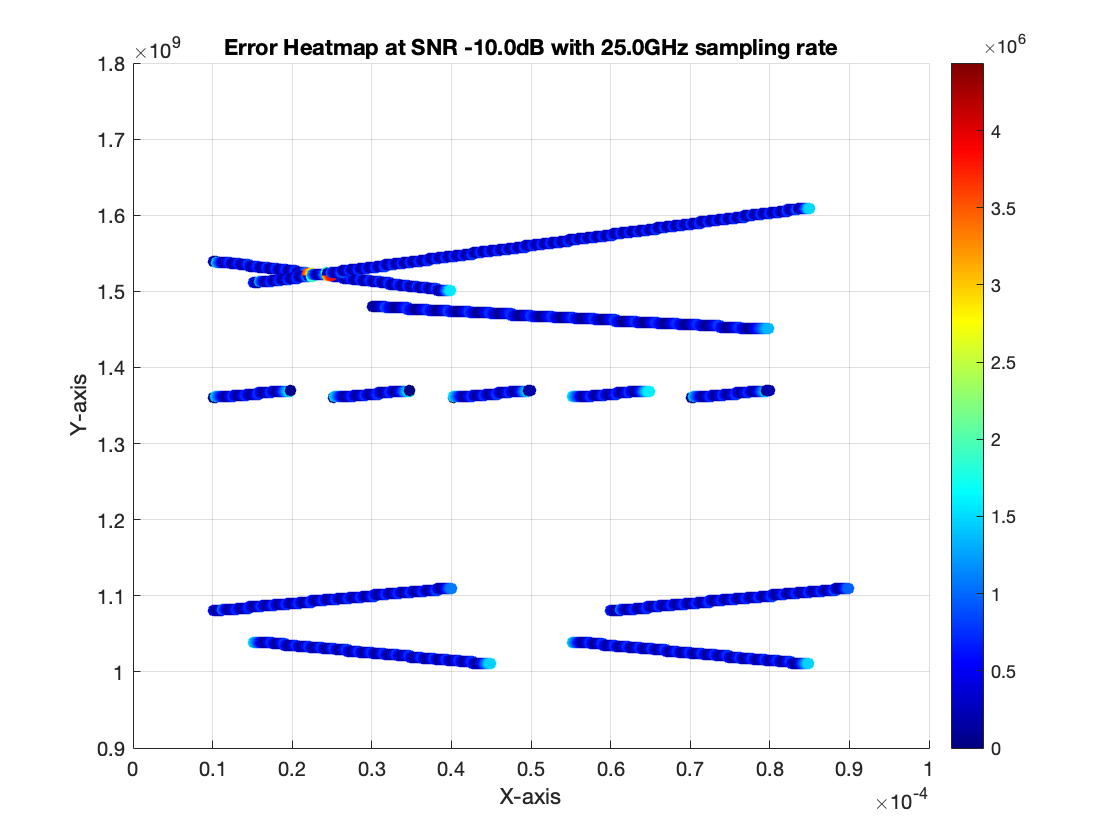}\\
\includegraphics[scale=.14]{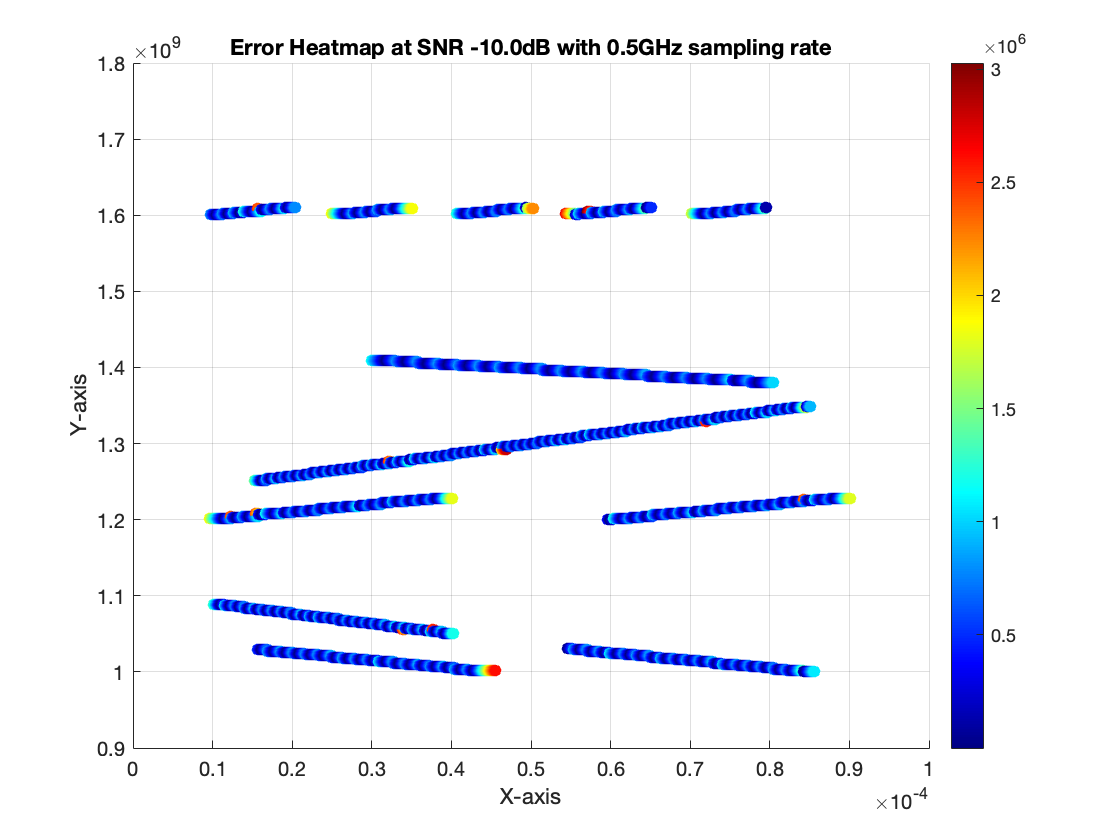}
\includegraphics[scale=.14]{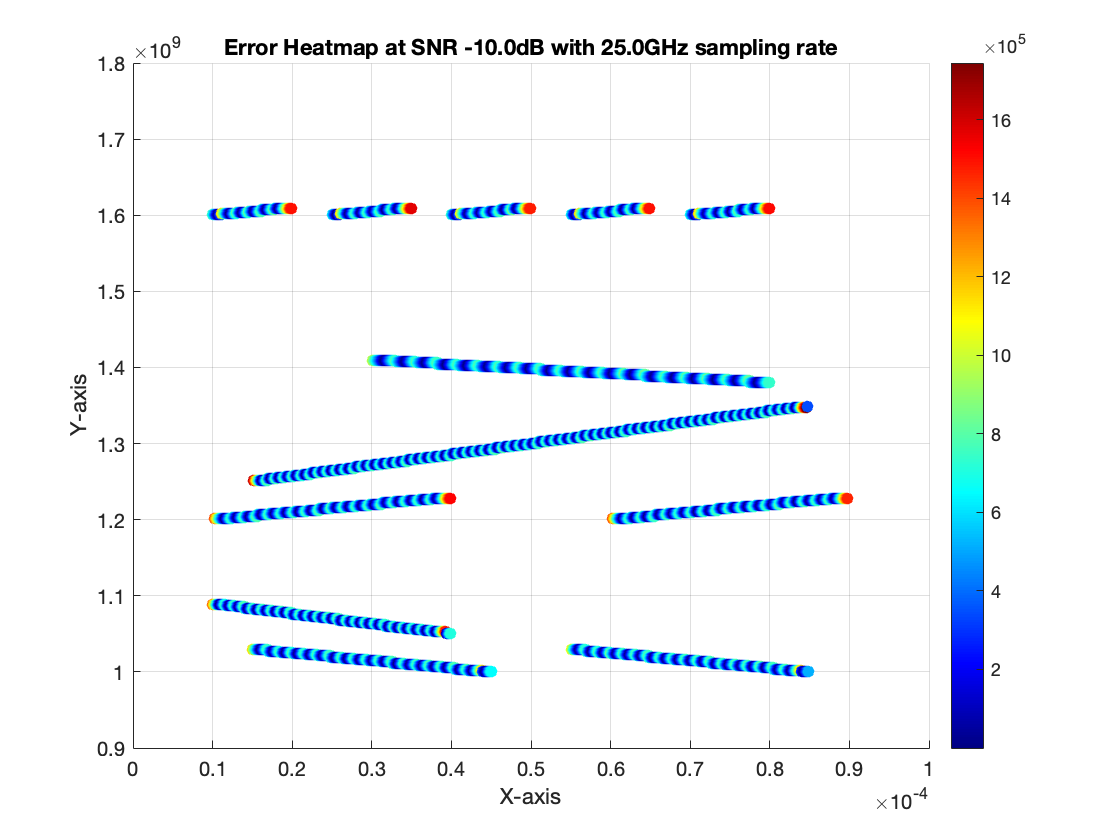}\\
\includegraphics[scale=.14]{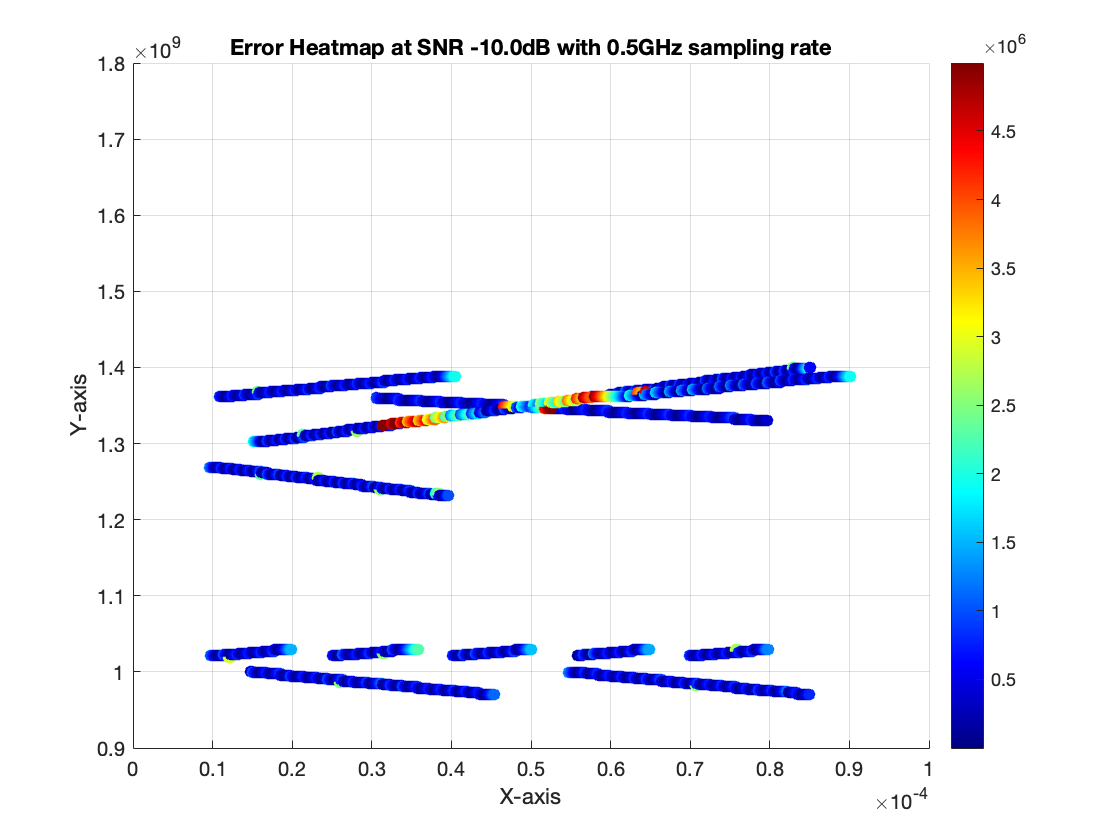}
\includegraphics[scale=.14]{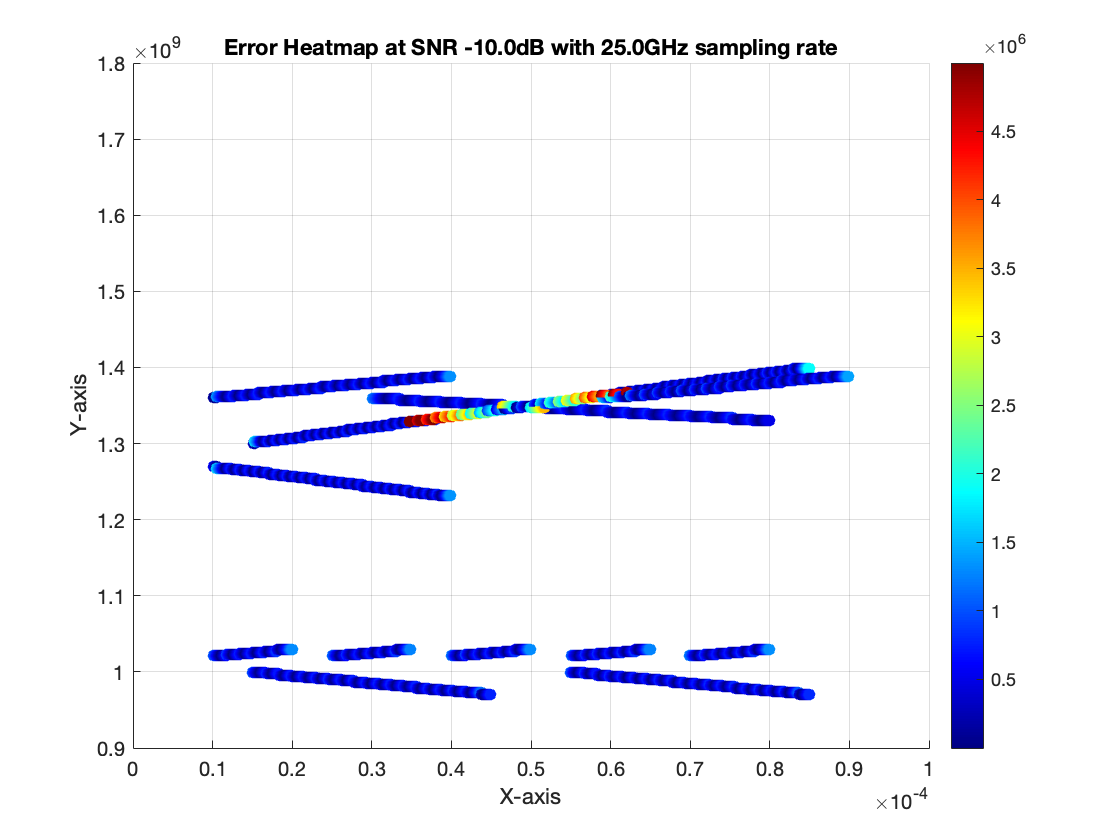}\\
\includegraphics[scale=.14]{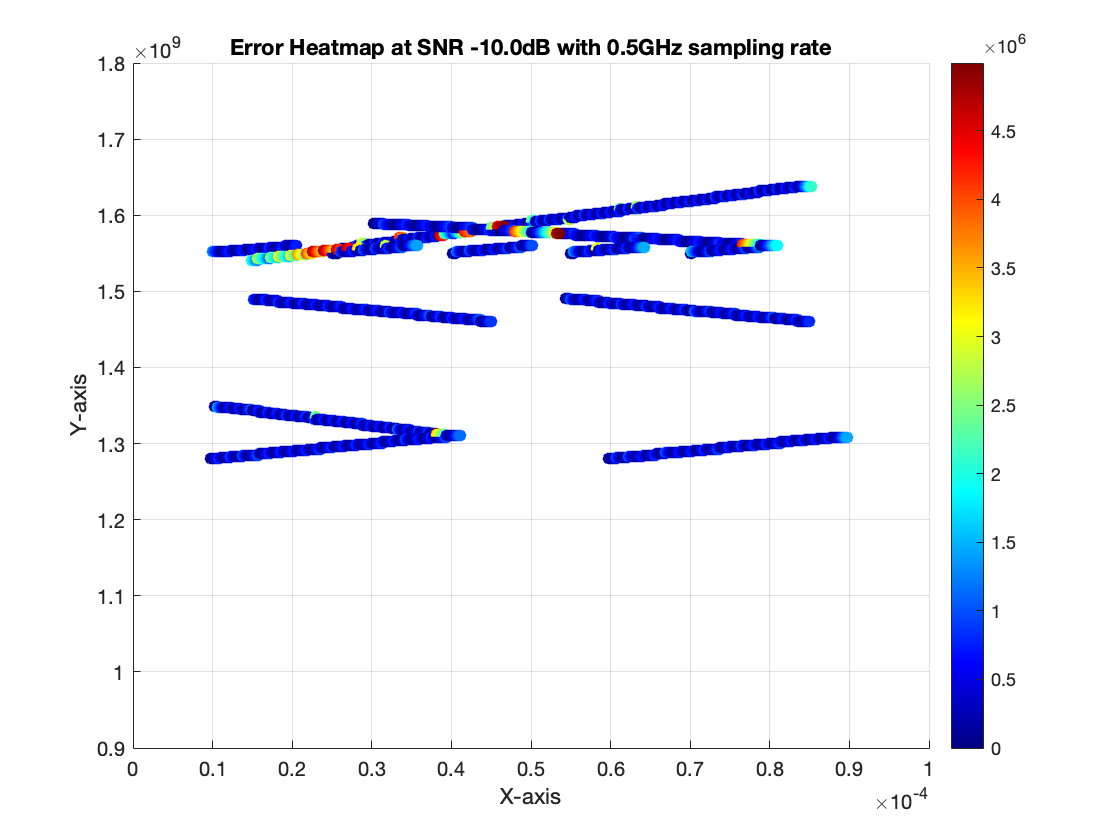}
\includegraphics[scale=.14]{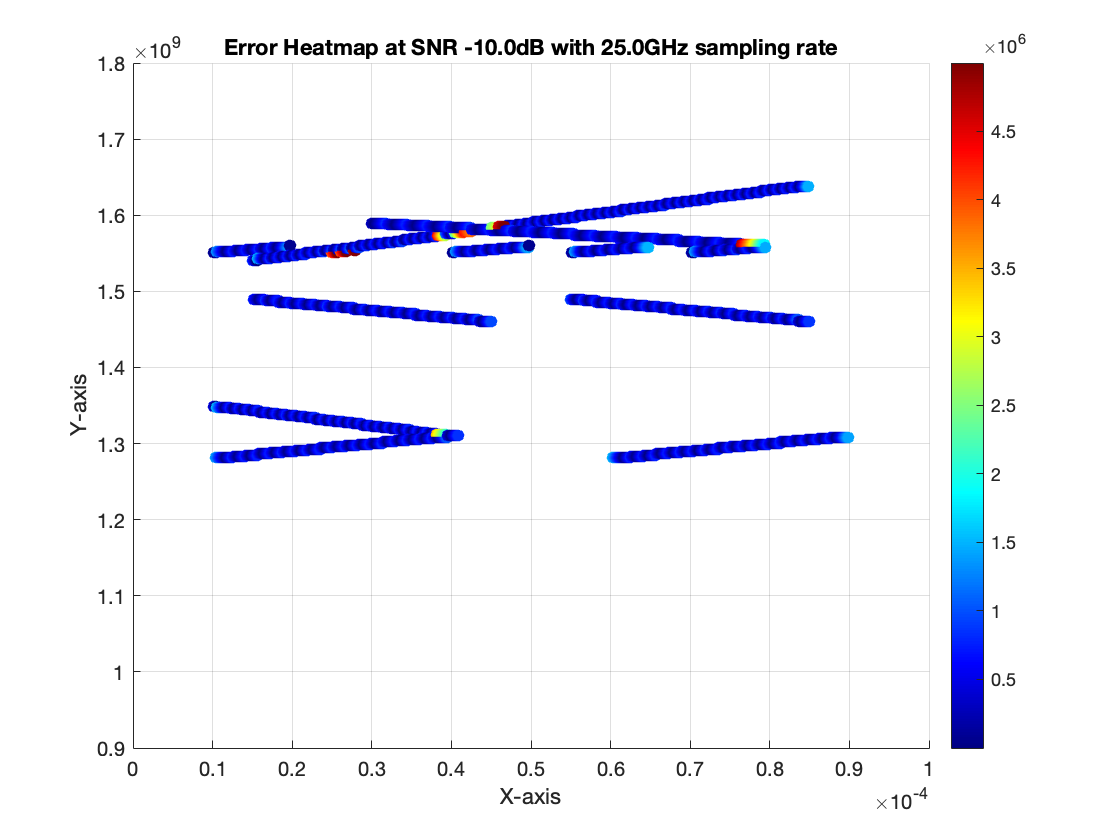}
\end{center}
\caption{Heat map of the errors for example 1, 2, 3, and 4 with SNR -10 dB at sampling rate 0.5 GHz (left) vs sampling rate 25 GHz (right). Ones may notice that the higher sampling rate examples can deal with the smaller minimal separation among signals better.
}
\label{fig:4heatmap}
\end{figure}

\subsubsection{Comparison with SST}\label{bhag:sst}
We recall that the most often used method for finding IF's is EMD, which is purely heuristic.
The method known as synchrosqueezing transform method (SST) is based on a solid mathematical foundation, and works better than EMD.
In this section we compare our results with those obtained by using SST.
As mentioned earlier, there are many versions of SST. 
We use the implementation of the SST algorithm as described in Section III of \cite{thakur2013synchrosqueezing}.
Figure \ref{fig:base_sst} shows that our implementation of this algorithm works for certain signals $f(t) = e^{2\pi i(15\times 10^7 t-5\times10^{12}t^2)} + e^{2\pi i(5\times 10^7t+5\times10^{12}t^2)}$ at high SNR levels and deteriorates at lower levels.
We observed that the maximum frequency that the SST method can take is $0.2$ GHz at sampling rate $0.5$ GHz while our dataset has maximum frequency at $1.6$ GHz.
Figure \ref{fig:4sst} illustrates that the SST fails to work in the regime in which we are interested in this chapter.

\begin{figure}[H]
\begin{center}
\includegraphics[scale=.12]{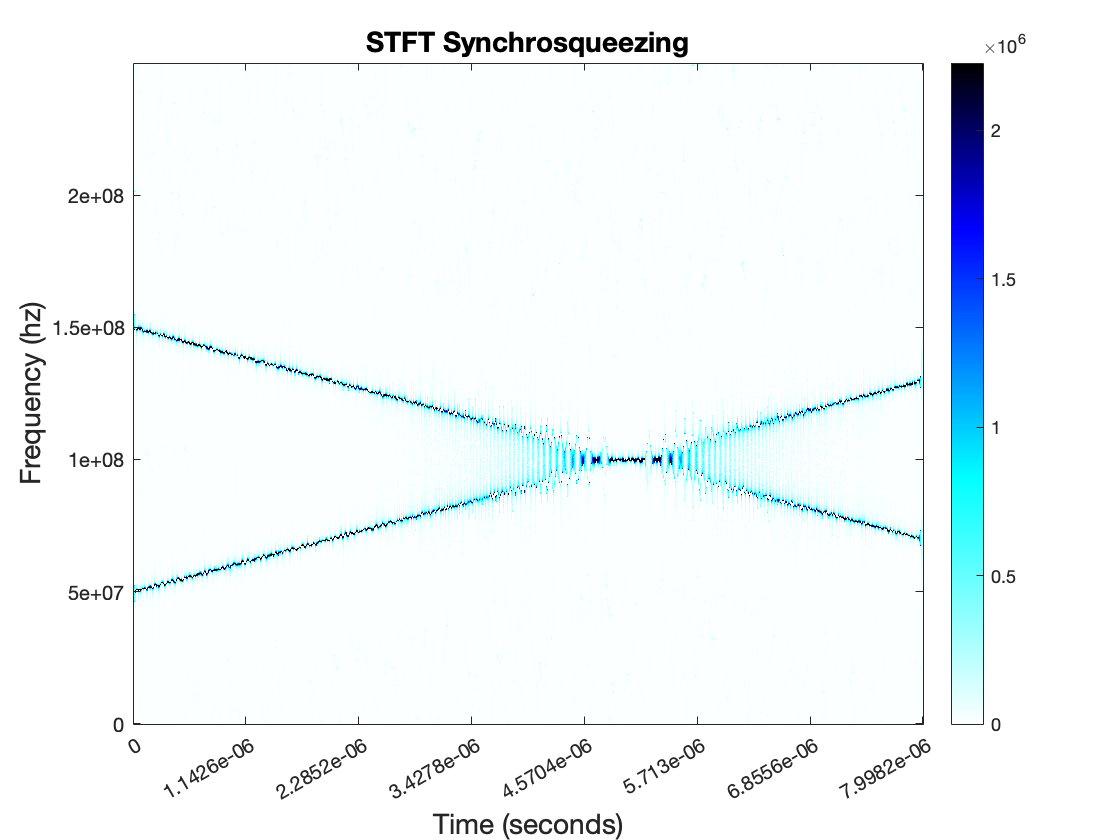}
\includegraphics[scale=.12]{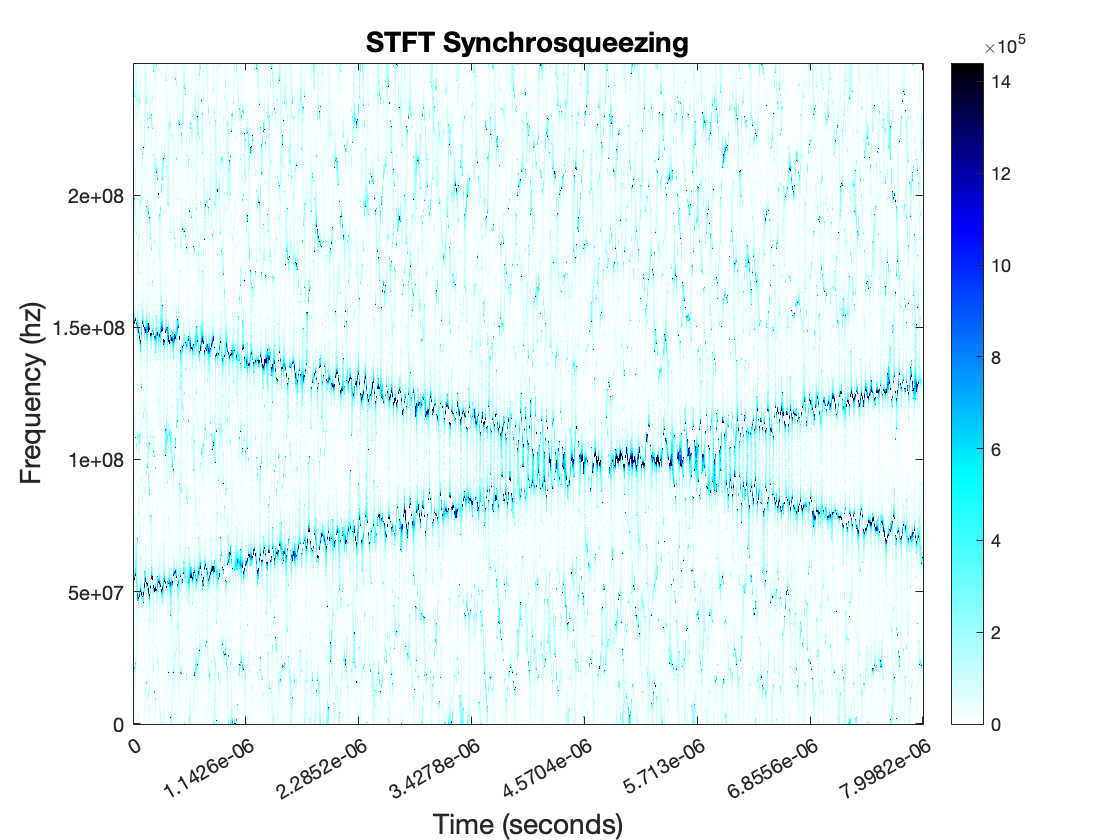}
\includegraphics[scale=.12]{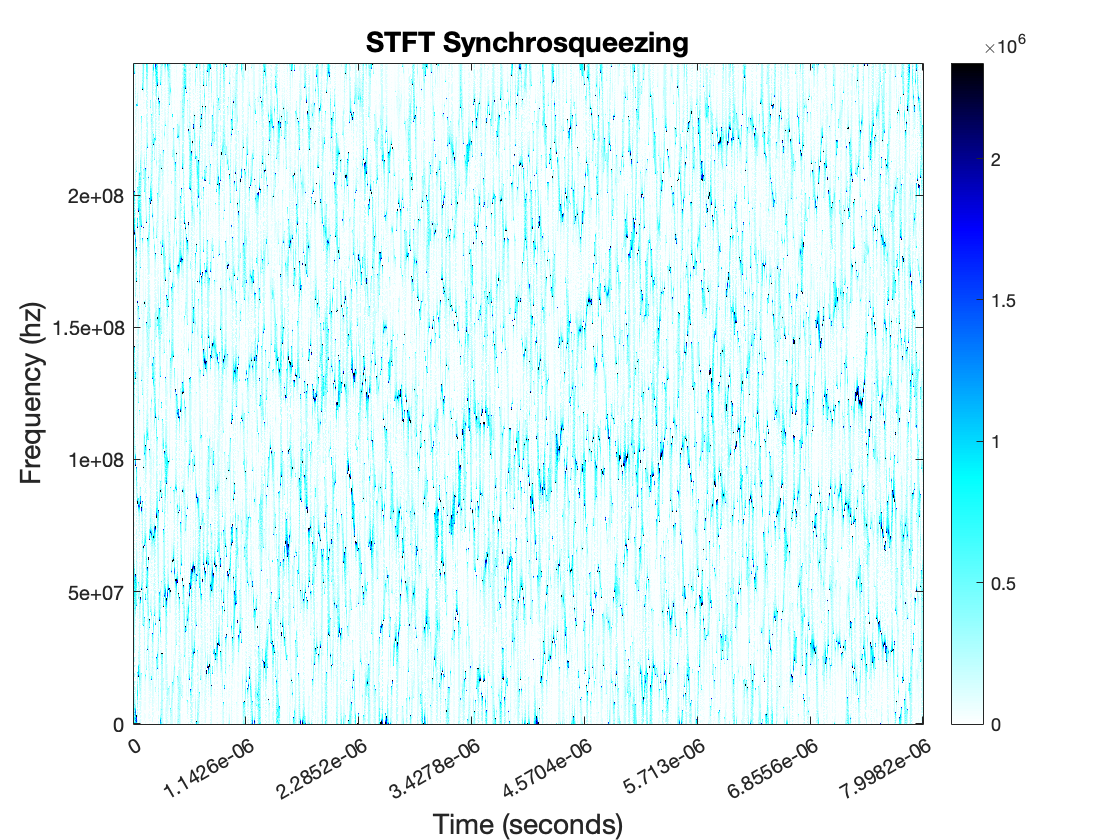}
\end{center}
\caption{The examples of SST plot with sampling rate 0.5 GHz at SNR 20 dB (left), SNR 0 dB (middle), and SNR -10 dB (right) for signals $f(t)$.}
\label{fig:base_sst}
\end{figure}

\begin{figure}[H]
\begin{center}
\includegraphics[scale=.10]{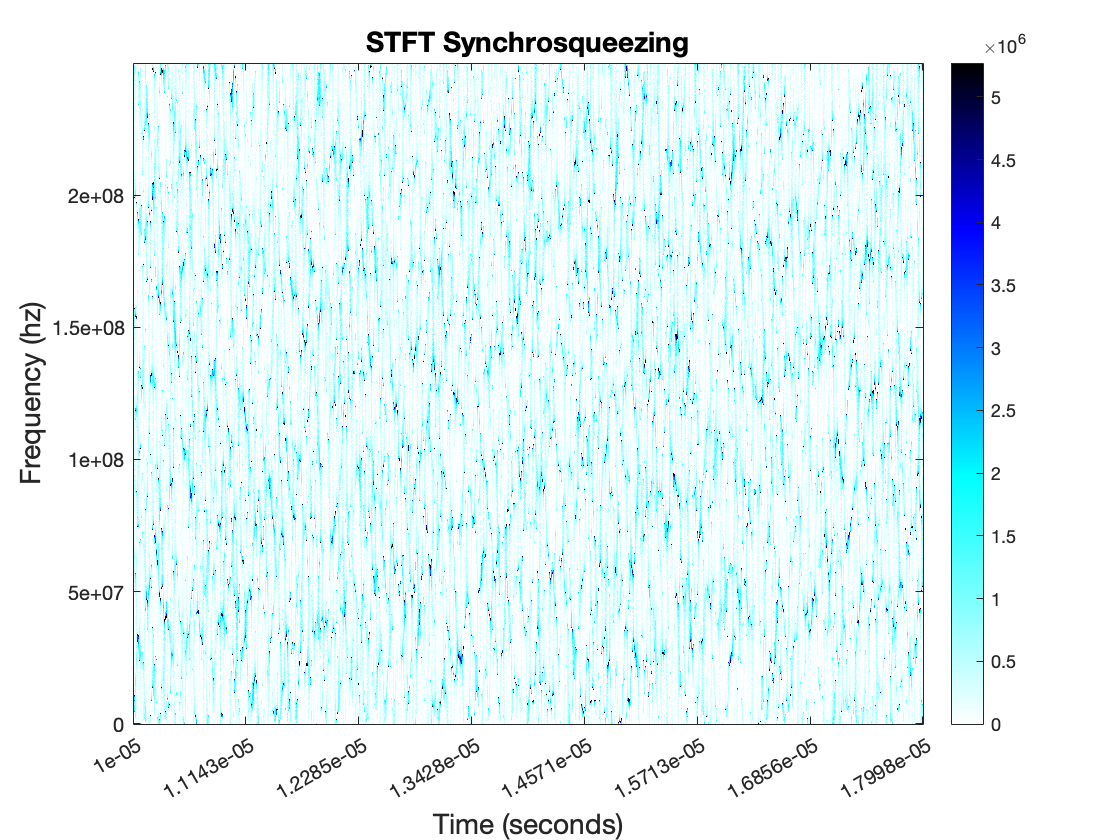}
\includegraphics[scale=.10]{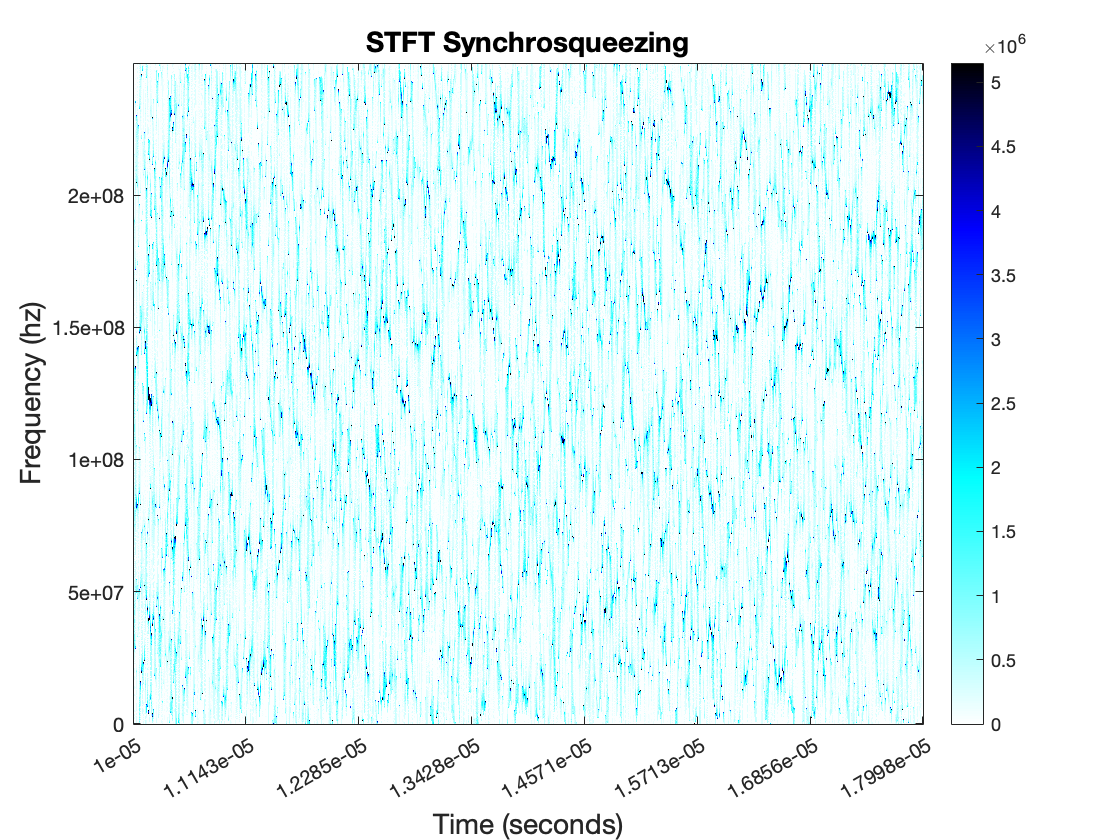}
\includegraphics[scale=.10]{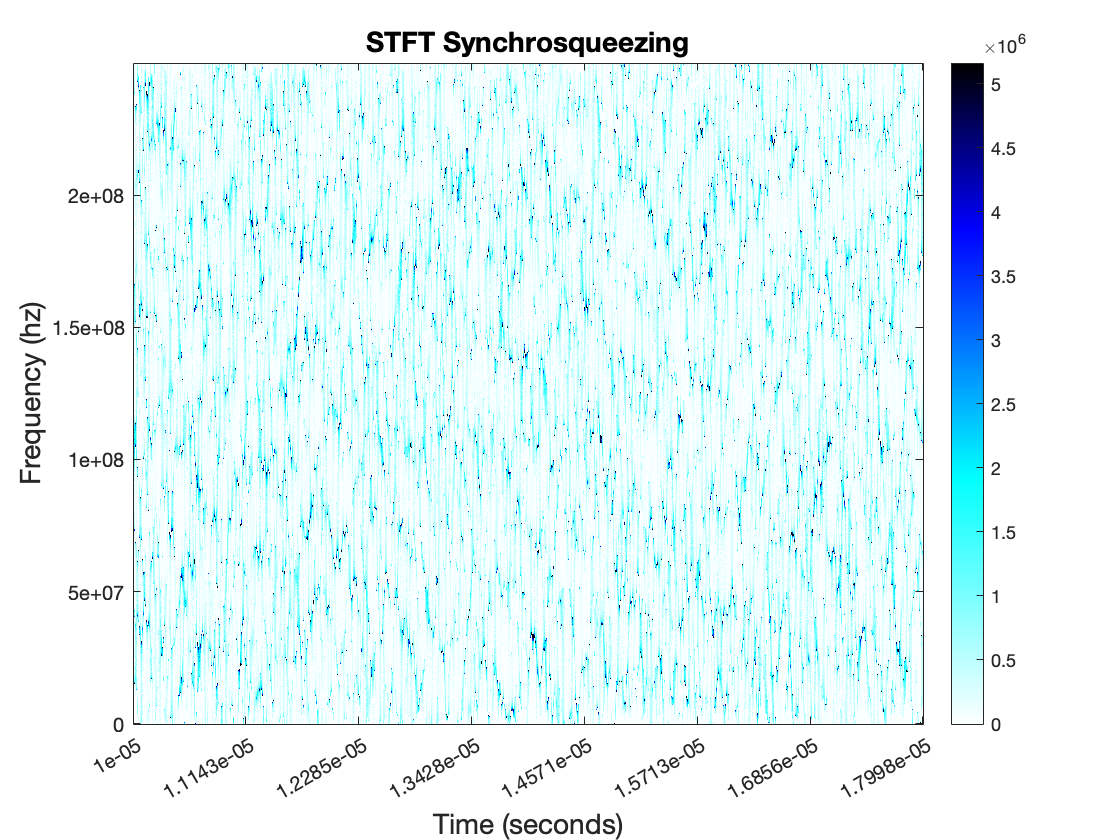}
\includegraphics[scale=.10]{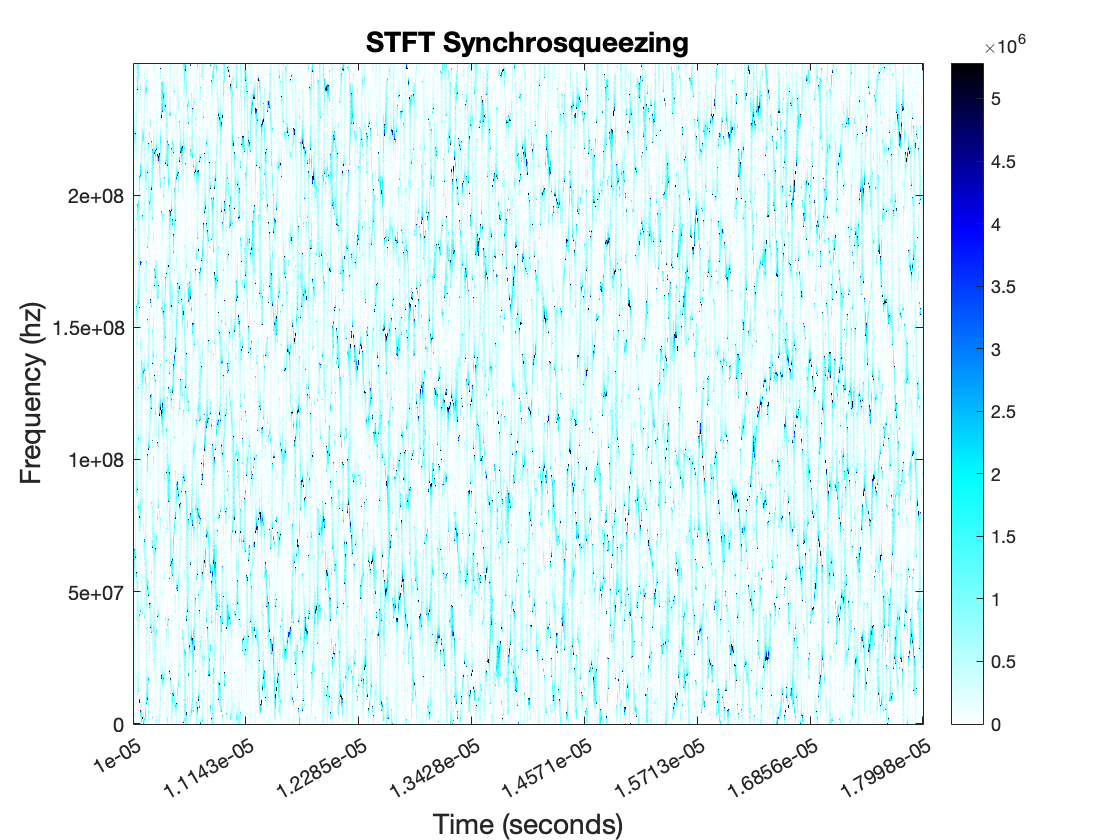}
\end{center}
\caption{The SST plots for example 1, 2, 3, and 4 with SNR 0 dB and sampling rate 0.5 GHz.}
\label{fig:4sst}
\end{figure}

\vspace*{-0.5cm}
\bhag{Conclusions}\label{bhag:conclusion}
\vspace*{-0.5cm}
Separation of a superposition of blind source signals is an important problem in audio and radar signal processing with  many applications, including electronic warfare and electronic intelligence.
In this chapter, we have modified a method proposed in \cite{bspaper} for solving this problem in the case when the instantaneous frequencies of the components are linear.
The method is based on a filtered FFT, and hence, is very fast.
While the paper \cite{bspaper} assumes the frequencies to be slowly and continuously varying, our modification includes the case when the signals are discontinuous and when there are crossover frequencies present. 
We analyze theoretically the behavior of SSO in the  presence of noise, giving insight into the relationship between the sampling frequency, desired  accuracy, and noise.
We have illustrated the effectiveness our method with a few well chosen examples, and tested it on a small database of 7 signals, demonstrating a consistent behavior. 
Our method outperforms the SST algorithm given in \cite{thakur2013synchrosqueezing}.

\bhag{Appendix}\label{bhag:appendix}
We report the statistics for 16 trials for each of the SNR and sampling rates for each signal, where the RMSE is computed using \eqref{eq:exptrmse}.
\begin{table}[H]
\begin{center}
\resizebox{0.455\textwidth}{!}{%
\begin{tabular}{ |c|c|c|c|c|c|c|c|c| } 
 \hline
 SNR & Example & Sampling & Total & Detected  & RMSE & Standard \\
 (dB) & label & rate (GHz) & signals & signals & & Deviation \\
 \hline
10 & 1 & 0.500000 & 12 & 12 & 0.006295 & 0.001066 \\
10 & 1 & 1.000000 & 12 & 12 & 0.006095 & 0.000770 \\
10 & 1 & 25.000000 & 12 & 12 & 0.006325 & 0.001249 \\
0 & 1 & 0.500000 & 12 & 12 & 0.006610 & 0.001203 \\
0 & 1 & 1.000000 & 12 & 12 & 0.006133 & 0.000828 \\
0 & 1 & 25.000000 & 12 & 12 & 0.006375 & 0.001291 \\
-10 & 1 & 0.500000 & 12 & 12 & 0.007433 & 0.001214 \\
-10 & 1 & 1.000000 & 12 & 12 & 0.006211 & 0.000888 \\
-10 & 1 & 25.000000 & 12 & 12 & 0.006022 & 0.000690 \\
-20 & 1 & 25.000000 & 12 & 12 & 0.006012 & 0.000666 \\
-30 & 1 & 25.000000 & 12 & 12 & 0.006436 & 0.001336 \\
10 & 2 & 0.500000 & 12 & 12 & 0.004428 & 0.002676 \\
10 & 2 & 1.000000 & 12 & 12 & 0.004197 & 0.002619 \\
10 & 2 & 25.000000 & 12 & 12 & 0.003990 & 0.002588 \\
0 & 2 & 0.500000 & 12 & 12 & 0.004496 & 0.002685 \\
0 & 2 & 1.000000 & 12 & 12 & 0.004255 & 0.002624 \\
0 & 2 & 25.000000 & 12 & 12 & 0.004038 & 0.002593 \\
-10 & 2 & 0.500000 & 12 & 12 & 0.004565 & 0.002698 \\
-10 & 2 & 1.000000 & 12 & 12 & 0.004315 & 0.002633 \\
-10 & 2 & 25.000000 & 12 & 12 & 0.004086 & 0.002607 \\
-20 & 2 & 25.000000 & 12 & 12 & 0.004038 & 0.002590 \\
-30 & 2 & 25.000000 & 12 & 12 & 0.004088 & 0.002597 \\
10 & 3 & 0.500000 & 12 & 11 & 0.003915 & 0.002504 \\
10 & 3 & 1.000000 & 12 & 11 & 0.003859 & 0.002402 \\
10 & 3 & 25.000000 & 12 & 11 & 0.003818 & 0.002293 \\
0 & 3 & 0.500000 & 12 & 12 & 0.003934 & 0.002532 \\
0 & 3 & 1.000000 & 12 & 11 & 0.003877 & 0.002427 \\
0 & 3 & 25.000000 & 12 & 11 & 0.003834 & 0.002314 \\
-10 & 3 & 0.500000 & 12 & 11 & 0.003950 & 0.002561 \\
-10 & 3 & 1.000000 & 12 & 11 & 0.003895 & 0.002452 \\
-10 & 3 & 25.000000 & 12 & 11 & 0.003825 & 0.002355 \\
-20 & 3 & 25.000000 & 12 & 11 & 0.003814 & 0.002332 \\
-30 & 3 & 25.000000 & 12 & 11 & 0.003850 & 0.002335 \\
10 & 4 & 0.500000 & 12 & 9 & 0.003918 & 0.002293 \\
10 & 4 & 1.000000 & 12 & 9 & 0.004061 & 0.002310 \\
10 & 4 & 25.000000 & 12 & 10 & 0.004335 & 0.002404 \\
0 & 4 & 0.500000 & 12 & 9 & 0.003886 & 0.002292 \\
0 & 4 & 1.000000 & 12 & 9 & 0.004025 & 0.002306 \\
0 & 4 & 25.000000 & 12 & 10 & 0.004310 & 0.002406 \\
-10 & 4 & 0.500000 & 12 & 9 & 0.003861 & 0.002296 \\
-10 & 4 & 1.000000 & 12 & 9 & 0.003994 & 0.002305 \\
-10 & 4 & 25.000000 & 12 & 10 & 0.004132 & 0.002319 \\
-20 & 4 & 25.000000 & 12 & 10 & 0.004184 & 0.002339 \\
-30 & 4 & 25.000000 & 12 & 10 & 0.004283 & 0.002408 \\
 \hline
\end{tabular}
}
\end{center}
 	\caption{The table above shows full performances of our algorithm.} \label{tab:result_table_1}
\end{table}

\begin{table}[H]
\begin{center}
\resizebox{0.45\textwidth}{!}{%
\begin{tabular}{ |c|c|c|c|c|c|c|c|c| } 
 \hline
 SNR & Example & Sampling & Total & Detected  & RMSE & Standard \\
 (dB) & label & rate (GHz) & signals & signals & & Deviation \\
 \hline
10 & 5 & 0.500000 & 12 & 12 & 0.005752 & 0.002110 \\
10 & 5 & 1.000000 & 12 & 12 & 0.005199 & 0.002565 \\
10 & 5 & 25.000000 & 12 & 12 & 0.004693 & 0.002723 \\
0 & 5 & 0.500000 & 12 & 12 & 0.005926 & 0.001830 \\
0 & 5 & 1.000000 & 12 & 12 & 0.005325 & 0.002474 \\
0 & 5 & 25.000000 & 12 & 12 & 0.004781 & 0.002699 \\
-10 & 5 & 0.500000 & 12 & 10 & 0.006114 & 0.001440 \\
-10 & 5 & 1.000000 & 12 & 12 & 0.005460 & 0.002364 \\
-10 & 5 & 25.000000 & 12 & 12 & 0.004971 & 0.002684 \\
-20 & 5 & 25.000000 & 12 & 12 & 0.004873 & 0.002693 \\
-30 & 5 & 25.000000 & 12 & 12 & 0.004874 & 0.002666 \\
10 & 6 & 0.500000 & 12 & 10 & 0.004385 & 0.002373 \\
10 & 6 & 1.000000 & 12 & 10 & 0.004349 & 0.002319 \\
10 & 6 & 25.000000 & 12 & 10 & 0.004416 & 0.002305 \\
0 & 6 & 0.500000 & 12 & 10 & 0.004383 & 0.002387 \\
0 & 6 & 1.000000 & 12 & 9 & 0.004365 & 0.002332 \\
0 & 6 & 25.000000 & 12 & 10 & 0.004406 & 0.002313 \\
-10 & 6 & 0.500000 & 12 & 10 & 0.004380 & 0.002400 \\
-10 & 6 & 1.000000 & 12 & 9 & 0.004380 & 0.002345 \\
-10 & 6 & 25.000000 & 12 & 10 & 0.004337 & 0.002298 \\
-20 & 6 & 25.000000 & 12 & 10 & 0.004344 & 0.002290 \\
-30 & 6 & 25.000000 & 12 & 8 & 0.004395 & 0.002322 \\
10 & 7 & 0.500000 & 12 & 10 & 0.004440 & 0.002269 \\
10 & 7 & 1.000000 & 12 & 10 & 0.004427 & 0.002231 \\
10 & 7 & 25.000000 & 12 & 11 & 0.004435 & 0.002193 \\
0 & 7 & 0.500000 & 12 & 9 & 0.004436 & 0.002280 \\
0 & 7 & 1.000000 & 12 & 10 & 0.004441 & 0.002241 \\
0 & 7 & 25.000000 & 12 & 11 & 0.004436 & 0.002204 \\
-10 & 7 & 0.500000 & 12 & 9 & 0.004435 & 0.002290 \\
-10 & 7 & 1.000000 & 12 & 9 & 0.004453 & 0.002252 \\
-10 & 7 & 25.000000 & 12 & 10 & 0.004409 & 0.002211 \\
-20 & 7 & 25.000000 & 12 & 11 & 0.004409 & 0.002202 \\
-30 & 7 & 25.000000 & 12 & 9 & 0.004441 & 0.002214 \\
 \hline
\end{tabular}
}
\end{center}
 	\caption{The table above shows full performances of our algorithm (continue).} \label{tab:result_table_2}
\end{table}


\chapter{Conclusion and Future Work}\label{summary}

This dissertation introduced two robust, tractable, and theoretically grounded methods for the decomposition and analysis of structured signals in both stationary and nonstationary settings. The central theme across both contributions is the use of localized trigonometric polynomial kernels to recover signal parameters efficiently and accurately under noisy condition.

\bhag{Summary of results}

This dissertation presents two algorithmic frameworks designed for signal recovery in distinct but related contexts: parameter estimation in multidimensional exponential sums, and separation of linear chirp components from short-time signal observations. Both approaches are grounded in localized trigonometric polynomial kernels and rely on efficient FFT-based computations.

In Chapter \ref{1paper}, we developed a multidimensional exponential analysis framework for estimating frequencies and amplitudes from a finite sum of complex exponentials sampled on a grid. Our method uses localized kernel reconstructions to directly estimate parameters while bypassing traditional subspace or eigenstructure methods. We established theoretical guarantees under sub-Gaussian noise assumptions and demonstrated the algorithm's superiority over classical methods like MUSIC and ESPRIT, particularly in high-dimensional and low-SNR regimes.

In Chapter \ref{2paper}, we extended the methodology based on localized kernels to nonstationary signals modeled as linear chirps. The proposed Signal Separation Operator (SSO) extracts instantaneous frequencies from short-time signal segments by filtering in the frequency domain. The method is automatic, noise-robust, and well-suited for handling crossover signals. Through synthetic radar scenarios, we demonstrated its accuracy in recovering complex chirp trajectories without prior information about the number of components.

Together, these techniques provide a unified and scalable framework for signal recovery in structured and noisy environments.

\bhag{Future research directions}

Several promising directions emerge from this work:

\begin{itemize}
  \item \textbf{Generalization to Nonlinear Chirps:} Extending the SSO method to handle higher-order chirps could broaden its application to more realistic radar and communication signal models.
  
  \item \textbf{Adaptive Kernel Design:} Future work could explore adaptive or data-driven kernel designs that optimize localization properties for specific signal models or noise conditions. For example, wider kernels may be beneficial in high-noise regimes to suppress spurious peaks, whereas narrower kernels could improve resolution in closely spaced frequency clusters. Designing such adaptive filters in a principled way could enhance both the robustness and resolution of the parameter recovery process.
  
  \item \textbf{Learning-Based Classification and Prediction:} By converting recovered signal parameters into feature representations, these methods can be integrated into pipelines for supervised classification or prediction tasks.
  
  \item \textbf{Application to Other Domains:} The underlying mathematical ideas may be applied to problems in biomedical signal processing, wireless sensor networks, and geophysical exploration, where sparse and structured signal recovery is essential.
\end{itemize}

By advancing tractable and robust techniques grounded in harmonic analysis, this dissertation contributes new tools to the broader effort of reliable signal inference in the presence of uncertainty and complexity.

\bibliographystyle{abbrv}
\bibliography{refs,hrushikesh,mason}

\end{document}